\documentclass{article}
\usepackage{kr}
\usepackage{times}
\usepackage{mathalfa}
\usepackage{xfrac}
\usepackage{stmaryrd}
\usepackage{iftex}

\usepackage{bm} 
\usepackage{soul}
\usepackage{url}
\usepackage[shortlabels]{enumitem}
\usepackage[hidelinks]{hyperref}
\usepackage[T1]{fontenc}
\usepackage[utf8]{inputenc}
\usepackage[small]{caption}
\usepackage{graphicx}
\usepackage{mathtools}
\usepackage{amsmath}
\usepackage{amssymb}
\usepackage{amsthm}
\usepackage{booktabs}
\usepackage{algorithm}
\usepackage{algorithmic}
\usepackage{amsfonts}
\usepackage{xcolor}
\usepackage{multicol}
\usepackage{multirow}
\usepackage{verbatim}
\usepackage{thm-restate}
\usepackage[normalem]{ulem}
\newcommand{\CPL}{\mathsf{CPL}}
\newcommand{\Luk}{{\normalfont{\textsf{\L}}}}
\newcommand{\Lukinttriangle}{\Luk^\Qmbb_\triangle}

\newcommand{\FP}{\mathsf{FP}}
\newcommand{\FPRLtriangle}{\mathsf{FP}_k(\mathsf{R}\Luk_\triangle)}

\newcommand{\cvalue}{\overline{\mathsf{c}}}
\newcommand{\dvalue}{\overline{\mathsf{d}}}
\newcommand{\evalue}{\overline{\mathsf{e}}}
\newcommand{\zero}{\overline{0}}
\newcommand{\one}{\overline{1}}

\newcommand{\conp}{\mathsf{coNP}}
\newcommand{\np}{\mathsf{NP}}
\newcommand{\DP}{\mathsf{DP}}
\newcommand{\pspace}{\mathsf{PSACE}}

\newcommand{\Var}{\mathsf{Var}}

\newcommand{\LLukint}{\mathcal{L}^\mathbb{Q}_\Luk}
\newcommand{\LProbint}{\mathcal{L}^\mathbb{Q}_\Prob}
\newcommand{\LCPL}{\mathcal{L}_{\CPL}}

\newcommand{\Cmsf}{{\mathsf{C}}}

\newcommand{\Hmsf}{{\mathsf{H}}}

\newcommand{\Pmsf}{{\mathsf{P}}}

\newcommand{\Smsf}{{\mathsf{S}}}

\newcommand{\Vmsf}{{\mathsf{V}}}

\newcommand{\Mfrak}{\mathfrak{M}}
\newcommand{\pfrak}{\mathfrak{p}}

\newcommand{\Vmbf}{\mathbf{V}}

\newcommand{\Embb}{\mathbb{E}}
\newcommand{\Hmbb}{\mathbb{H}}

\newcommand{\Nmbb}{\mathbb{N}}
\newcommand{\Pmbb}{\mathbb{P}}
\newcommand{\Qmbb}{\mathbb{Q}}
\newcommand{\Rmbb}{\mathbb{R}}

\newcommand{\Emc}{\mathcal{E}}

\newcommand{\Hmc}{\mathcal{H}}
\newcommand{\Imc}{\mathcal{I}}

\newcommand{\Omc}{{\mathcal{O}}}

\newcommand{\Vmc}{{\mathcal{V}}}

\newcommand{\consvDashLuk}{\models^{\mathsf{cons}}_\Luk}
\newcommand{\consvDashFPLuk}{\models^{\mathsf{cons}}_{\FP}}
\newcommand{\consvDashCPL}{\models^{\mathsf{cons}}_\CPL}

\newcommand{\Prob}{\mathsf{Pr}}
\newcommand{\Probfrak}{\mathfrak{Pr}}

\newcommand\myoverset[2]{\overset{\scriptstyle #1\mathstrut}{\scriptstyle #2\mathstrut}}
\newcommand{\ihsbminus}{{\textsf{IHS-B}^-}}
\newcommand{\rank}{\mathsf{rank}}

\newtheorem{proposition}{Proposition}
\newtheorem{lemma}{Lemma}

\newtheorem{example}{Example}
\newtheorem{definition}{Definition}
\newtheorem{convention}{Convention}
\theoremstyle{remark}

\DeclareFontFamily{U}{mathb}{}
\DeclareFontShape{U}{mathb}{m}{n}{<-5.5> mathb5 <5.5-6.5> mathb6 
<6.5-7.5> mathb7 <7.5-8.5> mathb8 <8.5-9.5> mathb9 <9.5-11> mathb10 
<11-> mathb12}{}
\DeclareSymbolFont{mathb}{U}{mathb}{m}{n}
\DeclareFontSubstitution{U}{mathb}{m}{n}

\DeclareMathSymbol{\blackdiamond}{\mathbin}{mathb}{"0C}

\title{Probabilistic Abduction in a~Fuzzy Logic Framework}

\author{Tommaso Flaminio$^1$ \and Katsumi Inoue$^2$ \and Daniil Kozhemiachenko$^3$\\
\affiliations
$^1$IIIA --- CSIC, Campus de la UAB, Bellaterra, Barcelona, Spain\\
$^2$National Institute of Informatics, Japan
$^3$Aix Marseille Univ, CNRS, LIS, Marseille, France\\
\emails
tommaso@iiia.csic.es, inoue@nii.ac.jp, daniil.kozhemiachenko@lis-lab.fr}

\begin{document}
\maketitle
\begin{abstract}
We study the problem of explaining observations about the probabilities of events such as ‘it rains $20\%$ of the time’, ‘rain and snow are equally likely’, etc. We explain these statements with a~probability distribution or a~statement about probabilities of (other) events that are consistent with our knowledge and entail the observation. We formalise this problem in a~fuzzy probabilistic logic $\FP$. We define and motivate the notions of abduction problems and their solutions. Our main technical contribution is a~comprehensive study of the complexity of solution recognition and existence for a~given abduction problem in~$\FP$ for the case of full language and its disjunctive-clause fragments. We also obtain a~translation of classical probabilistic abduction (finding the most likely explanation of a~given event) to~$\FP$.
\end{abstract}
\allowdisplaybreaks
\section{Introduction\label{sec:introduction}}
Deduction, induction, and abduction are the main forms of reasoning~\cite{FlachKakas2000}. Abduction, or inference of~\emph{explanations}, has found multiple applications in artificial intelligence. Among the most prominent ones are diagnosis \cite{Pople1973,ConsoleTorasso1990,ElAyebMarquisRusinowitch1993,JosephsonJosephson1994,Koitz-HristovWotawa2018}, commonsense reasoning~\cite{Paul1993,BhagavatulaLeBrasMalaviyaSakaguchiHoltzmanRashkinDowneyYihChoi2020}, formalisation of scientific reasoning~\cite{Magnani2001}, and machine learning~\cite{DaiXuYuZhou2019}, in particular, explainable~AI~\cite{IgnatievNarodytskaMarques-Silva2019}.

In logic-based abduction~\cite{EiterGottlob1995}, the reasoning is usually conducted in the framework of \emph{abduction problems} that contain a~\emph{theory}~$\Gamma$ formalising the background knowledge and an \emph{observation}~$\delta$. The goal is to find an \emph{explanation} of~$\delta$ \emph{relative to}~$\Gamma$ using a~pre-defined set of \emph{hypotheses}. That is, a~formula~$\eta$ s.t.\ $\Gamma,\eta\not\models\bot$ and $\Gamma,\eta\models\delta$ (in other words, $\eta$ and~$\Gamma$ should \emph{consistently entail}~$\delta$).

In practice, however, our data may be unreliable, and thus the knowledge might bear a~degree of uncertainty. In such contexts, one might search for \emph{probabilistic} explanations of an observation. Hence, to provide explanations in the context of uncertainty, one may employ \emph{probabilistic abduction}.

\paragraph{Probabilistic Abduction}
Probabilistic abduction was first introduced by Poole~\shortcite{Poole1993} as probabilistic Horn abduction (PHA). This framework computes probabilities of Bayesian networks by definite clauses with probabilistic hypotheses, but does not compute probabilities based on probability measures over possible worlds as defined by Fenstad~\shortcite{Fenstad1967}.

Sato~\shortcite{Sato1995} uses definite clauses with probabilistic hypotheses like PHA, but proposes the distribution semantics, which defines probabilistic measures over possible worlds without assumptions of PHA. Sato et al.~\shortcite{SatoIshihataInoue2011} extended this framework to allow arbitrary clauses in the background theory and any ground formulas in observations as constraint-based probabilistic modelling. The distribution semantics were also adopted in probabilistic logic programming. In particular, de~Raedt et al.~\shortcite{DeRaedtKimmigToivonen2007} present ProbLog, which considers a probability for each definite clause, i.e., a~probabilistic hypothesis to select each clause, and Azzolini et al.~\shortcite{AzzoliniBellodiferilliRiguzziZese2022} propose algorithms for computing explanations to probabilistic logic programs.

Other frameworks include a Bayesian view of prior distributions in abduction~\cite{DuboisGilioKern-Isberner2008} and abduction in Markov logic networks~\cite{KateMooney2009}. They, however, do not strictly provide a~logical formalisation of abduction to explain observations.

\paragraph{Probabilistic Reasoning in Fuzzy Logics}
Probabilistic abduction requires formalisation of reasoning about uncertainty. One of the approaches to this task is via the employment of so-called \emph{probabilistic logics}. Probabilistic logic in the context of~AI was first proposed by Nilsson~\shortcite{Nilsson1986}. The propositional fragment of Nilsson's logic was then formalised by Fagin et al.~\shortcite{FaginHalpernMegiddo1990}. H\'{a}jek et al.~\shortcite{HajekGodoEsteva1995,HajekGodoEsteva2000} formalised reasoning about probabilities using \emph{fuzzy logics}, i.e., logics where formulas have values in the real-valued interval~$[0,1]$. The intuitive idea was to consider a~(non-nesting) modality $\Prob$ such that, for a classical formula~$\phi$, the modal formula $\Prob(\phi)$ is interpreted as the probability measure of the event corresponding to~$\phi$.

Connections between probabilities and fuzzy logics have been extensively studied. H\'{a}jek and Tulipani~\shortcite{HajekTulipani2001} considered the complexity of probabilistic reasoning, in particular, with conditional probabilities. Flaminio~\shortcite{Flaminio2007} proposed a~simpler formalisation of reasoning with conditional probabilities. Baldi et al.~\shortcite{BaldiCintulaNoguera2019,BaldiCintulaNoguera2020} showed that fuzzy probabilistic logics and the probabilistic logics proposed by Fagin et al.\ have the same expressive power. Corsi et al.~\shortcite{CorsiFlaminioGodoHosni2023} studied epistemic probabilistic logics and Kozhemiachenko and Sedl\'{a}r~\shortcite{KozhemiachenkoSedlar2025} considered formalisation of reasoning about probabilities and actions.

\paragraph{Abduction in Fuzzy Logics}
Abduction in various fuzzy logics has long attracted interest. The first formal presentation was given by Yamada and Mukaidono~\shortcite{YamadaMukaidono1995} in the framework of Łukasiewicz logic. Vojta\v{s}~\shortcite{Vojtas1999} expanded this approach to G\"{o}del and Product logics. Miyata et al.~\shortcite{MiyataFuruhashiUchikawa1998} considered abductive reasoning with observations of different degrees. Solutions to abduction problems in multiple fuzzy logics were further systematised by~d'Allones et al.~\shortcite{dAllonnesAkdagBouchon-Meunier2007} and Chakraborty et al.~\shortcite{ChakrabortyKonarPalJain2013}. Recently, Inoue and Kozhemiachenko~\shortcite{InoueKozhemiachenko2025} provided a~comprehensive study of the complexity of abduction in Łukasiewicz logic.

Abduction in fuzzy logic has found numerous applications. Namely, Vojta\v{s}~\shortcite{Vojtas2001} and Ebrahim~\shortcite{Ebrahim2001} proposed fuzzy logic programming. Bergadano et al.~\shortcite{BergadanoCutelloGunetti2000} studied fuzzy abduction in the context of machine learning. Among other applications are decision-making and learning in the presence of incomplete information~\cite{MellouliBouchon-Meunier2003,Tsypyschev2017} and robot perception~\cite{Shanahan2005}.

\paragraph{Contributions}
Although probabilistic abduction and applications of fuzzy logics to probabilistic reasoning have been extensively studied, there has been, to the best of our knowledge, no formalisation of probabilistic abduction in a~fuzzy logic setting. Furthermore, in most probabilistic abduction frameworks, the main task is to find \emph{the most likely explanation} of an observed event. On the other hand, in some situations, we are interested in \emph{why the event occurred with the observed frequency}. In such contexts, an explanation can be a~\emph{probability distribution} consistent with our knowledge about probabilities of events that explains the frequency of the occurrence of the observed event. Consider the following running example for an illustration.
\begin{example}\label{example:probabilities1}
We study weather records of a~town and notice that there was very bad weather (it rained and was cold on the same day) at least $20\%$ of the time. The records are incomplete, so we do not know how likely it was to rain or to be cold on each given day. Moreover, there could be cloudy days without rain and days with mild weather (i.e., events ‘it rains’ and ‘it is sunny’ as well as ‘it is cold’ and ‘it is warm’ are \emph{not complementary}). We know, however, that good weather (sunny and warm days) was twice as rare as bad weather (rainy or cold days). Furthermore, cold days were $1.5$ times more often than warm days, and sunny days were rarer than rainy days. Now, we aim to \emph{explain the observed frequency} of days with very bad weather.
\end{example}

We consider probabilistic abductive reasoning within the framework of a~fuzzy probabilistic logic $\FP$ that extends Łukasiewicz logic. We will formally define these logics in the next section. In our proofs, we will rely on the results for abduction in (propositional) Łukasiewicz logic~\cite{InoueKozhemiachenko2025}.
%
Our contribution is threefold. (1)~We define and motivate abduction problems and their solutions that adequately represent the contexts, as in the example above. (2)~We study the complexity of probabilistic abduction in the full logic and its disjunctive-clause fragments. We consider \emph{solution recognition} (given a~solution candidate~$\zeta$ and a~problem~$\Pmbb$, check if $\zeta$~solves~$\Pmbb$) and \emph{solution existence} (given~$\Pmbb$, check if it has solutions). (3)~We formalise finding most likely explanations in~$\FP$.

Omitted proofs can be found in the appendix of the technical report~\cite{FlaminioInoueKozhemiachenko2026arxiv}.
\section{Probabilistic Łukasiewicz Logic\label{sec:FPLuk}}
In this section, we present the fuzzy probability logic $\FP$~--- a~probabilistic expansion of Łukasiewicz logic~$\Luk$ introduced by H\'{a}jek et al. in~\shortcite{HajekGodoEsteva1995}. The logic $\FP$ has a \emph{two-layered} language made of \emph{probabilistic atoms} of the form $\Prob(\phi)$. E.g., if $\phi=r\vee h$, and $r$~stands for ‘it is raining’ and $h$~for ‘it is hot outside’, the value of $\Prob(r\vee h)$ corresponds to the probability of the event ‘it is raining, or it is hot outside’. \emph{Compound probability formulas} are connected by operators of~$\Luk$. To improve its expressivity, we extend its language with the $\triangle$~operator of Baaz~\shortcite{Baaz1996} and truth constants. We fix a~countable set $\Var$ of propositional variables, let $p\in\Var$, $c\in\Qmbb\cap[0,1]$, $\diamond\in\{\leq,<,>,\geq\}$, and define~$\LLukint$:
\begin{align*}
\phi&\coloneqq p\mid p\diamond\cvalue\mid\triangle\phi\mid\neg\phi\mid(\phi\odot\phi)\mid(\phi\oplus\phi)\mid(\phi\!\rightarrow\!\phi)
\end{align*}
\begin{convention}\label{conv:notation}
$\LCPL$ denotes the $\{\neg,\wedge,\vee\}$-lan\-gu\-age of the classical propositional logic. For a~set of formulas $\Gamma$ and a~formula~$\phi$, $\Var(\phi)$ and $\Var[\Gamma]$ are the set of all variables occurring in $\phi$ and~$\Gamma$, respectively.

$\Rmbb$ and $\Qmbb$ denote the sets of real and rational numbers, respectively. A~square (round) bracket means that the endpoint is included in (excluded from) the interval. Lower index~$_\Qmbb$ means that the interval contains rational numbers only. E.g., $(\tfrac{1}{2},\tfrac{2}{3}]_\Qmbb=\{x\mid x\in\Qmbb,x>\tfrac{1}{2},x\leq\tfrac{2}{3}\}$.
\end{convention}
\begin{figure*}
\centering
\begin{align*}
v(\neg\phi)&=1{-}v(\phi)&v(p\diamond\cvalue)&=\begin{cases}1\text{ if }v(p)\diamond c\\0\text{ otherwise}\end{cases}&v(\triangle\phi)&=\begin{cases}1\text{ if }v(\phi){=}1\\0\text{ otherwise}\end{cases}\\
v(\phi{\odot}\chi)&=\max(0,v(\phi){+}v(\chi){-}1)&v(\phi{\oplus}\chi)&=\min(1,v(\phi){+}v(\chi))&v(\phi\!\rightarrow\!\chi)&=\min(1,1{-}v(\phi){+}v(\chi))
\end{align*}
\caption{Semantics of $\Lukinttriangle$.}
\label{fig:Luksemantics}
\end{figure*}
\begin{definition}[Semantics of Łukasiewicz logic]\label{def:Luksemantics}
An~\emph{$\Luk$-va\-lu\-a\-tion} is a~function $v:\Var\rightarrow[0,1]$ extended to the complex formulas as shown in Figure~\ref{fig:Luksemantics} (cf.~the top of the next page).

We say that $\phi\in\LLukint$ is \emph{$\Luk$-valid ($\Luk\models\phi$)} if $v(\phi)=1$ in all $\Luk$-valuations; $\phi$ is \emph{$\Luk$-satisfiable} if $v(\phi)=1$ in some $\Luk$-valuation. Given a~finite $\Phi\!\subseteq\!\LLukint$, \emph{$\Phi$ entails $\chi$ in~$\Luk$ ($\Phi\models_\Luk\chi$)} iff $v(\chi)\!=\!1$ in every $v$ s.t.\ $v(\phi)\!=\!1$ for all $\phi\!\in\!\Phi$. \emph{$\Phi$~consistently entails $\chi$ in~$\Luk$ ($\Phi\!\consvDashLuk\!\chi$)} iff $\Phi\models_\Luk\chi$ and $\Phi\not\models_\Luk\bot$.
\end{definition}

The language $\LProbint$ of $\FP$ is defined as follows (here, $\phi{\in}\LCPL$, $\heartsuit{\in}\{\neg,\triangle\}$, $\circ{\in}\{\odot,\oplus,\rightarrow\}$, and $\diamond{\in}\{\leq,<,>,\geq\}$).
\begin{align*}
\alpha&\coloneqq\Prob(\phi)\mid\Prob(\phi)\diamond\cvalue\mid\heartsuit\alpha\mid(\alpha\circ\alpha)
\end{align*}
We will interpret $\LProbint$-formulas on \emph{probabilistic} models following~\cite{HajekTulipani2001}. The only change is that, to simplify some proofs, instead of evaluating variables at each state, we define states of models as sets of variables.
\begin{definition}[Probability measures and distributions]\label{def:probabilitymeasure}
Given a~finite $W\neq\varnothing$,
\begin{itemize}[noitemsep,topsep=0pt]
\item a~map $\mu:2^W\rightarrow[0,1]$ is a~\emph{probability measure on~$2^W$} iff $\mu(W)=1$, $\mu(\varnothing)=0$, and $\mu(X\cup Y)=\mu(X)+\mu(Y)$ for every $X,Y\subseteq W$ s.t.\ $X\cap Y=\varnothing$;
\item a~map $\partial:W\rightarrow[0,1]$ is a~\emph{probability distribution on~$W$} iff $\sum_{w\in W}\partial(w)=1$.
\end{itemize}
\end{definition}
\begin{definition}[Probabilistic models]
Given a~finite set $\Vmbf\subseteq\Var$, a~\emph{probabilistic model for~$\Vmbf$} is a~tuple $\Mfrak=\langle2^\Vmbf,\mu\rangle$ s.t.\ \mbox{$\mu:2^{2^\Vmbf}\rightarrow[0,1]$} is a~\emph{probability measure on~$2^{2^\Vmbf}$}. For $p{\in}\Vmbf$, $\phi{\in}\LCPL$ s.t.\ $\Var(\phi)\subseteq\Vmbf$, and a~probabilistic model $\Mfrak{=}\langle2^\Vmbf,\mu\rangle$, we define the \emph{truth set of~$\phi$} ($\|\phi\|_\Mfrak$) as follows:
\begin{itemize}[noitemsep,topsep=0pt]
\item $\|p\|_\Mfrak=\{X\in2^\Vmbf\mid p\in X\}$;
\item $\|\neg\phi\|=2^{2^\Vmbf}{\setminus}\|\phi\|_\Mfrak$;
\item $\|\phi\wedge\chi\|_\Mfrak=\|\phi\|_\Mfrak\cap\|\chi\|_\Mfrak$.
\end{itemize}
\end{definition}
\begin{definition}[Semantics of $\FP$]\label{def:FPLuksemantics}
The \emph{$\FP$-interpretation} induced by~$\Mfrak$ is a~map $\Imc_\Mfrak:\LProbint\rightarrow[0,1]$ s.t.\ $\Imc_\Mfrak(\Prob(\phi))=\mu(\|\phi\|_\Mfrak)$. The interpretation is extended onto complex $\LProbint$-formulas using the conditions of Definition~\ref{def:Luksemantics} (cf.~Fig.~\ref{fig:Luksemantics}).
\end{definition}

The notions of $\FP$-satisfiability, validity, and entailment generalise those in Łukasiewicz logic (cf.~Definition~\ref{def:Luksemantics}).
\begin{definition}\label{def:FPentailment}
Let $\Gamma\cup\{\alpha,\delta\}\subseteq\LProbint$, $\Vmbf=\Var(\alpha)$, and $\Vmbf'=\Var[\Gamma\cup\{\delta\}]$. We say that
\begin{itemize}[noitemsep,topsep=0pt]
\item $\alpha$ is \emph{$\FP$-satisfiable} iff $\Imc_\Mfrak(\alpha)=1$ for some $\Mfrak=\langle2^\Vmbf,\mu\rangle$;
\item $\alpha$ is \emph{$\FP$-valid} ($\FP{\models}\alpha$) iff $\Imc_\Mfrak(\alpha){=}1$ in each $\Mfrak{=}\langle2^\Vmbf,\mu\rangle$;
\item $\Gamma$ \emph{$\FP$-entails} $\delta$ ($\Gamma\models_\FP\delta$) iff $\Imc_{\Mfrak'}(\delta)=1$ for all $\Mfrak'=\langle2^{\Vmbf'},\mu\rangle$ s.t.\ $\Imc_{\Mfrak'}(\gamma)=1$ for all $\gamma\in\Gamma$;
\item $\Gamma$ \emph{consistently $\FP$-entails} $\delta$ ($\Gamma\consvDashFPLuk\delta$) iff $\Gamma\models_\FP\delta$ and $\Gamma\not\models_\FP\bot$.
\end{itemize}
\end{definition}
\begin{convention}\label{conv:innerouterformulas}
$\LProbint$-formulas are called \emph{outer formulas}. Formulas occurring in probabilistic atoms are \emph{inner formulas} or \emph{events}. We denote $\LProbint$ formulas with $\alpha$, $\beta$, \ldots, $\lambda$, $\LCPL$ and $\LLukint$ formulas with $\pi$, $\varrho$, \ldots, $\psi$. We denote \emph{sets} of $\LProbint$ formulas with $\Gamma$, $\Delta$, \ldots, $\Xi$, and \emph{sets} of $\LCPL$ or $\LLukint$ formulas with $\Upsilon$, $\Phi$, and~$\Psi$. For $\Gamma{\cup}\{\alpha\}{\subseteq}\LProbint$, we set $\Emc(\alpha)=\{\phi\in\LCPL\mid\Prob(\phi)\, \text{occurs in}\, \alpha\}$ and $\Emc[\Gamma]{=}\bigcup_{\gamma\in\Gamma}\Emc(\gamma)$.

For $\Phi\subseteq\LCPL$ and a~probabilistic model $\Mfrak=\langle2^\Vmbf,\mu\rangle$, we will write $\mu(\|\Phi\|_\Mfrak)$ to denote $\mu(\|\bigwedge_{\phi\in\Phi}\phi\|_\Mfrak)$.

An \emph{$\LCPL$-term} ($\LCPL$-clause) is a~formula of the form $\bigwedge^n_{i=1}l_i$ ($\bigvee^n_{i=1}l_i$) with $l$'s being literals. We write $\alpha\leftrightarrow\beta$ as a~shorthand for $(\alpha{\rightarrow}\beta)\odot(\beta{\rightarrow}\alpha)$, $\Prob(\phi)\blackdiamond\cvalue$ and $p\blackdiamond\cvalue$ as shorthands for $\neg(\Prob(\phi)\diamond\cvalue)$ and $\neg(p\diamond\cvalue)$, and $p{\approx}\cvalue$ and $\Prob(\phi){\approx}\cvalue$ as shorthands for $(p{\geq}\cvalue)\odot(p{\leq}\cvalue)$ and $(\Prob(\phi){\geq}\cvalue)\odot(\Prob(\phi){\leq}\cvalue)$.
\end{convention}

Let us remark on $\FP$. First, $\Emc(\alpha)$ contains only the events explicitly mentioned in~$\alpha$ (not their combinations or subformulas). E.g., if $\alpha=\Prob(p\wedge q)\rightarrow\Prob(r)$, $\Emc(\alpha)=\{p\wedge q,r\}$ (i.e., $p,q\notin\Emc(\alpha)$). Thus, $\Emc(\alpha)$ is linear in the length of~$\alpha$. Second, $\triangle$ allows for a~reduction between validity and satisfiability: $\alpha$ is $\FP$-valid iff $\neg\triangle\alpha$ is $\FP$-unsatisfiable; $\alpha$~is $\FP$-satisfiable iff $\neg\triangle\alpha$ is not $\FP$-valid. Third, we obtain the following statement by a~straight\-forward expansion of~\cite[Theorem~6]{Flaminio2007}.
\begin{proposition}\label{prop:FPLukNPcoNP}
$\FP$-satisfiability is $\np$-complete, and $\FP$-entailment is $\conp$-complete.
\end{proposition}
\section{Probabilistic Abduction in~$\FP$\label{sec:FPabduction}}
Let us now present abduction in $\FP$. We begin with the notion of abduction problems.
\begin{definition}\label{def:FPLukabduction}
An \emph{$\FP$-abduction problem} ($\FP$ AP) is a~tuple $\Pmbb=\langle\Gamma,\delta,\Hmsf,\Vmc\rangle$ s.t.\ $\Gamma\cup\{\delta\}\subseteq\LProbint$, $\Hmsf$ is a~finite set of $\LCPL$-terms s.t.\ $\Var[\Hmsf]\subseteq\Var[\Gamma\cup\{\delta\}]$, and $\Vmc=\{0,\tfrac{1}{n},\ldots,\tfrac{n-1}{n},1\}$ with $n\in\Nmbb$. We call $\Gamma$ \emph{theory}, $\delta$~\emph{ob\-ser\-va\-tion}, members of~$\Hmsf$ and~$\Vmc$ \emph{hypotheses} and \emph{values}.
\end{definition}

Let us represent Example~\ref{example:probabilities1} as an $\FP$-abduction problem.
\begin{example}\label{example:probabilities2}
Recall the context of Example~\ref{example:probabilities1}. We use $c$, $r$, $s$, and $w$ to denote cold, rainy, sunny, and warm days, respectively, and define $\Pmbb_\mathsf{rain}=\langle\Gamma_\mathsf{rain},\delta_\mathsf{rain},\Hmsf_\mathsf{rain},\Vmc_\mathsf{rain}\rangle$.
\begin{align*}
\Gamma_\mathsf{rain}&=\left\{\begin{matrix}\Prob(s\!\wedge\!w)\!\leq\!\overline{\tfrac{1}{2}},\Prob(w)\!\leq\!\overline{\tfrac{2}{3}},\neg\triangle(\Prob(r)\!\rightarrow\!\Prob(s)),\\
(\Prob(s\wedge w)\oplus\Prob(s\wedge w))\leftrightarrow\Prob(c\vee r),\\
(\Prob(x)\oplus\Prob(x))\leftrightarrow\Prob(w),\\
\Prob(c)\leftrightarrow(\Prob(w)\oplus\Prob(x))
\end{matrix}\right\}\nonumber\\
\delta_\mathsf{rain}&=\Prob(c\!\wedge\!r)\!\geq\!\overline{\tfrac{1}{5}};~
\Hmsf_\mathsf{rain}=\{s,w,s\wedge w\};\\
\Vmc_\mathsf{rain}&=\{\tfrac{i}{100}\!\mid0\leq i\!\leq\!100\}\nonumber
\end{align*}
Let us briefly discuss the formalisation. By Definitions~\ref{def:Luksemantics} and~\ref{def:FPLuksemantics}, \mbox{$\Imc_\Mfrak(\neg\triangle(\alpha\rightarrow\beta))=1$ iff $\Imc_\Mfrak(\alpha)>\Imc_\Mfrak(\beta)$.} Thus, $\neg\triangle(\Prob(r)\!\rightarrow\!\Prob(s))$ means ‘the probability of a~rainy day is greater than the probability of a~sunny day’. $\Prob(s\!\wedge\!w)\!\leq\!\overline{\tfrac{1}{2}}$ and $(\Prob(s\wedge w)\oplus\Prob(s\wedge w))\leftrightarrow\Prob(c\vee r)$ imply $\Imc_\Mfrak(\Prob(s\wedge r))=\Imc_\Mfrak(\Prob(c\vee r))\cdot\tfrac{1}{2}$, i.e., rainy or cold days are two times likelier than days with warm and sunny weather.

We also add an auxiliary variable~$x$ to express $\Imc_\Mfrak(\Prob(w))=2\cdot\Imc_\Mfrak(\Prob(x))$ via $(\Prob(x)\oplus\Prob(x))\leftrightarrow\Prob(w)$ and $\Prob(w)\!\leq\!\overline{\tfrac{2}{3}}$. Now, ‘a~cold day is $1.5$ times likelier than a~warm one’ can be formalised as $\Prob(c)\leftrightarrow(\Prob(w)\oplus\Prob(x))$. To provide an informative explanation of the frequency of days with cold and rainy weather, we need hypotheses $s$, $w$, and~$s\wedge w$ allowing us to express probabilities of sunny and warm days. Finally, we choose~$\Vmc$ expressing percentages.
\end{example}

Now, to explain the observed frequency, we will need to assume probabilities of sunny and warm days. This means that the solution to an $\FP$-abduction problem will not be a~conjunction of hypotheses, but \emph{a~statement about probabilities of hypotheses}. Note that we need $s{\wedge} w$ because, in general, probability measures are not truth functional. In particular, $\mu(\|s{\wedge} w\|_\Mfrak)$ is not uniquely determined by $\mu(\|s\|_\Mfrak)$ and~$\mu(\|w\|_\Mfrak)$. E.g., consider $\Vmbf=\{s,w\}$, $\Mfrak_1=\langle2^\Vmbf,\mu_1\rangle$, and $\Mfrak_2=\langle2^\Vmbf,\mu_2\rangle$ being such that $\mu_1(\{s,w\})=\mu_1(\{s\})=\mu_1(\{w\})=\mu_1(\{\varnothing\})=\tfrac{1}{4}$, $\mu_2(\{s,w\})=\mu_2(\{\varnothing\})=\tfrac{1}{2}$, and $\mu_2(\{s\})=\mu_2(\{w\})=0$ (recall that the measure is defined on $2^{2^\Vmbf}$). One can see that $\mu_1(\|s\|_{\Mfrak_1}){=}\mu_1(\|w\|_{\Mfrak_1}){=}\mu_2(\|s\|_{\Mfrak_2}){=}\mu_2(\|w\|_{\Mfrak_2})=\tfrac{1}{2}$, but $\mu_1(\|s\wedge w\|_{\Mfrak_1})=\tfrac{1}{4}$ and $\mu_2(\|s\wedge w\|_{\Mfrak_2})=\tfrac{1}{2}$.

In~$\LProbint$, statements about probabilities of events can be formalised using formulas of the form $\Prob(\phi)\geq\cvalue$ (‘probability of~$\phi$ is at least~$c$’) or $\Prob(\phi)\rightarrow\Prob(\chi)$ (‘probability of $\phi$ is at least as high as the probability of~$\chi$’). Note, however, that explanations are usually expressed as (sets of) facts (i.e., conjunctions of literals). As we are explaining observed frequencies of events, it makes sense to use probabilities of facts as solutions. To do that, we introduce \emph{probabilistic interval literals} and \emph{terms} that generalise (propositional) interval literals and terms proposed by Inoue and Kozhemiachenko~\shortcite{InoueKozhemiachenko2025} for Łukasiewicz logic.
\begin{definition}[Probabilistic interval literals and terms]\label{def:probabilisticintervalliterals}~
\begin{itemize}[noitemsep,topsep=0pt]
\item A~\emph{probabilistic interval literal} (PIL) is a~formula $\Prob(\tau)\diamond\cvalue$ with $\tau$ being an $\LCPL$-term.
\item A~\emph{probabilistic interval term} (PIT) is a~formula $\bigodot^n_{i=1}\lambda_i$ with $\lambda$'s being PILs.
\end{itemize}

For a~PIL $\Prob(\tau){\diamond}\dvalue$, we define its \emph{set of permitted values}:
\begin{align*}
\Vmsf_\Prob(\Prob(\tau)\diamond\dvalue)&=\!\{\Imc_\Mfrak(\Prob(\tau))\!\mid\!\Imc_\Mfrak(\Prob(\tau)\diamond\dvalue)\!=\!1\}
\end{align*}

Given $\Vmbf=\{p_1,\ldots,p_m\}$ and $\Vmc=\{0,\tfrac{1}{n-1},\ldots,1\}$, we call a~PIT $\theta=\bigodot^k_{i=1}\Prob(\tau_i){\leq}\cvalue_i\odot\bigodot^k_{i=1}\Prob(\tau_i){\geq}\cvalue_i$ \emph{$\langle\Vmbf,\Vmc\rangle$-complete} if $\Var[\tau_i]=\Vmbf$ and $c_i\in\Vmc$ for every $i\in\{1,\ldots,k\}$, and $\sum^k_{i=1}c_i=1$.
\end{definition}


Observe briefly that all $\langle\Vmbf,\Vmc\rangle$-complete PITs are $\FP$-satisfiable. Let us now define the notion of a~solution to an $\FP$-abduction problem.
\begin{definition}\label{def:solutions}
Let $\Pmbb=\langle\Gamma,\delta,\Hmsf,\Vmc\rangle$ be an $\FP$ AP and let further $\Vmbf=\Var[\Pmbb]$.
\begin{itemize}[noitemsep,topsep=0pt]
\item A~\emph{sufficient solution} is a~PIT $\eta=\bigodot^n_{i=1}\Prob(\tau_i)\diamond\cvalue_i$ s.t.\ $\Emc(\eta)\subseteq\Hmsf$, $\cvalue_i\in\Vmc$ for every $\cvalue_i$, and \mbox{$\Gamma,\eta\consvDashFPLuk\delta$.}
\item A~\emph{full solution} is a~$\langle\Vmbf\!,\!\Vmc\rangle$-complete PIT~$\theta$ s.t.\ \mbox{$\Gamma,\!\theta\consvDashFPLuk\delta$.}
\end{itemize}
\end{definition}
To guarantee the finiteness of the set of solutions, we restrict the set of \emph{values} of probabilistic atoms containing hypotheses. Intuitively, this can be interpreted as an explanation with a~given precision. Furthermore, full solutions correspond to \emph{probability distributions} on~$2^\Vmbf$. Hence, there is \emph{only one} probabilistic model that satisfies $\theta$. I.e., if~$\theta$ is a~full solution, there is no full solution~$\theta'$ s.t.\ $\theta'\models_\FP\theta$ and $\theta\not\models_\FP\theta'$. Besides, a~full solution~$\theta$ is not necessarily a~sufficient solution since in general, $\Emc(\theta)\not\subseteq\Hmsf$. Moreover, while the length of sufficient solutions is bounded by $\Emc[\Pmbb]$ (i.e., is linear in the length of~$\Pmbb$), the length of full solutions is not. One could interpret \emph{full solutions} as probabilistic counterparts of \emph{most specific explanations (MSEs)}~\cite{Poole1985,Stickel1990,SakamaInoue1995}. The difference is that MSEs specify as many values of hypotheses as possible, and full solutions specify probabilities of \emph{every event}.

Let us now solve the problem from Example~\ref{example:probabilities2}.
\begin{example}\label{example:probabilities3}
We begin with a~sufficient solution: 
\begin{align*}
\eta_\mathsf{rain}&=(\Prob(w){\approx}\overline{0.4})\odot(\Prob(s){\approx}\overline{0.4})\odot(\Prob(s\wedge w){\approx}\overline{0.3})
\end{align*}
By Definition~\ref{def:FPLuksemantics}, we obtain $\Imc_\Mfrak(\Prob(c))\!=\!\tfrac{3}{5}$, $\Imc_\Mfrak(\Prob(r))\in(\tfrac{2}{5},\tfrac{3}{5}]$, and $\Imc_\Mfrak(\Prob(c\!\vee\!r))\!=\!\tfrac{3}{5}$ from Example~\ref{example:probabilities2}. Thus, $\Imc(\Prob(c\wedge r))>\tfrac{2}{5}$ in every probabilistic model~$\Mfrak$ that satisfies $\Gamma_\mathsf{rain}$ and~$\eta_\mathsf{rain}$. Note that $\eta_\mathsf{rain}$ is not a~full solution: we cannot infer probabilities of \emph{all} events over $\{c,r,s,w,x\}$. But $\eta_\mathsf{rain}$ \emph{can be extended to a~full solution}, e.g., as follows:
\begin{align*}
\theta_\mathsf{rain}=&\Prob(s\wedge w\wedge \neg x\wedge\neg c\wedge\neg r){\approx}\overline{0.3}\odot\\
&\Prob(\neg s\wedge\neg w\wedge\neg x\wedge c\wedge r){\approx}\overline{0.5}\odot\\
&\Prob(s\wedge\neg w\wedge x\wedge c\wedge\neg r){\approx}\overline{0.1}\odot\\
&\Prob(\neg s\wedge w\wedge x\wedge\neg c\wedge r){\approx}\overline{0.1}%
\end{align*}
By Definitions~\ref{def:FPLuksemantics}, \ref{def:FPentailment}, and~\ref{def:solutions}, $\theta_\mathsf{rain}$ is indeed a~full solution to~$\Pmbb_\mathsf{rain}$. Clearly, $\Gamma_\mathsf{rain},\theta_\mathsf{rain}\consvDashFPLuk\delta_\mathsf{rain}$. Furthermore, all $\tau\in\Emc(\theta_\mathsf{rain})$ correspond to some $X\subseteq\Vmbf$ and the values of $\tau$'s add up to~$1$. Thus $\theta_\mathsf{rain}$ is $\langle\Vmbf,\Vmc\rangle$-complete. Moreover, $\theta_\mathsf{rain}\models_\FP\eta_\mathsf{rain}$. Hence, $\theta_\mathsf{rain}$ can be interpreted as an extension of~$\eta_\mathsf{rain}$.
\end{example}

We finish the section with a~short observation. From Definition~\ref{def:solutions}, one can see that if a~problem has a~sufficient solution, then it has a~full solution. On the other hand, the converse is not necessarily true. E.g., if we modify the set of hypotheses in Example~\ref{example:probabilities3} as follows: $\Hmsf^-_\mathsf{rain}=\{w\}$, $\Pmbb^-=\langle\Gamma_\mathsf{rain},\delta_\mathsf{rain},\Hmsf^-_\mathsf{rain},\Vmc_\mathsf{rain}\rangle$ will not have sufficient solutions. Still~$\theta_\mathsf{rain}$ remains a~valid full solution.
\section{Complexity of $\FP$-Abduction\label{sec:complexity}}
Before proceeding to establish the complexity of abduction in~$\FP$, let us first define some preliminary notions. Since most of our hardness results will be obtained by reductions from $\Luk$-abduction, we recall here the definitions of $\Luk$-abduction problems and their solutions from~\cite{InoueKozhemiachenko2025}.
\begin{definition}[$\Luk$-abduction problems and their solutions]\label{def:Lukabduction}~
\begin{itemize}[noitemsep,topsep=0pt]
\item An \emph{$\Luk$-abduction problem} is a~tuple $\Pmbb=\langle\Phi,\chi,\Hmsf\rangle$ with $\Phi\cup\{\chi\}\subseteq\LLukint$ being finite and $\Hmsf$ a~finite set of \emph{interval literals} s.t.\ $\Var[\Hmsf]\subseteq\Var[\Phi\cup\{\chi\}]$. We call $\Phi$ a~\emph{theory}, $\chi$~an \emph{observation}, and members of $\Hmsf$ \emph{hypotheses}.
\item An \emph{$\Luk$-solution} (simply called a~solution) of $\Pmbb$ is an~\emph{interval term} $\tau$ composed of hypotheses s.t.\ $\Phi,\tau\consvDashLuk\chi$.
\end{itemize}
\end{definition}

We note that $\Lukinttriangle$ is the fragment of~$\FP$ with \emph{atomic} events. The following statement is a~straightforward generalisation of~\cite[Proposition 5.6]{FlaminioPretoUgolini2023} and~\cite[Theorem 5.3]{FlaminioUgolini2024}.
\begin{definition}\label{def:FPLuktoLukembedding}
For $\phi\in\LLukint$, $\sharp\in\{\neg,\!\triangle\}$, and $\circledast{\in}\{\odot,\!\oplus,\!\rightarrow\}$, we define the \emph{probabilistic counterpart} $\phi^\Prob$ as follows:
\begin{align*}
p^\Prob&=\Prob(p)&(p\diamond\cvalue)^\Prob&=\Prob(p)\diamond\cvalue\\
(\sharp\chi)^\Prob&=\sharp(\chi^\Prob)&(\chi\circledast\psi)^\Prob&=\chi^\Prob\circledast\psi^\Prob
\end{align*}
For $\Phi\subseteq\LLukint$, we set $\Phi^\Prob=\{\phi^\Prob\mid\phi\in\Phi\}$.
\end{definition}
\begin{restatable}{proposition}{FPLuktoLukembedding}\label{prop:LukisfragmentofFPLuk}
For all $\Phi\cup\{\chi\}\subseteq\LLukint$, it holds that $\Phi\models_\Luk\chi$ iff $\Phi^\Prob\models_\FP\chi^\Prob$.
\end{restatable}

We are now ready to consider the complexity of abductive reasoning in~$\FP$. We will deal with two problems:
\begin{itemize}[noitemsep,topsep=0pt]
\item \emph{solution recognition} --- given a~problem~$\Pmbb$ and a~PIT~$\zeta$, determine whether~$\zeta$ is a~(sufficient or full) solution to~$\Pmbb$;
\item \emph{solution existence} --- given a~problem~$\Pmbb$, determine whether it has (sufficient or full) solutions.
\end{itemize}
\subsection{Solution Recognition and Existence\label{ssec:solutionrecognition}}
We begin by determining the complexity of solution recognition. The next two statements show that while recognition of \emph{sufficient} solutions is $\DP$-complete (same as in classical or Łukasiewicz logic), recognising a~\emph{full} solution is polynomially tractable.
\begin{restatable}{theorem}{sufficientsolutionrecognitionDP}\label{theorem:sufficientsolutionrecognitionDP}
Given an $\FP$ AP $\Pmbb=\langle\Gamma,\delta,\Hmsf,\Vmc\rangle$ and a~PIT~$\eta$, it is $\DP$-complete to check if $\eta$~is a~sufficient solution to~$\Pmbb$.
\end{restatable}
\begin{restatable}{theorem}{fullsolutionrecognitionpolynomial}\label{theorem:fullsolutionrecognitionpolynomial}
Given an $\FP$ AP $\Pmbb=\langle\Gamma,\delta,\Hmsf,\Vmc\rangle$ with $\Vmbf=\Var[\Pmbb]$ and a~$\langle\Vmbf,\Vmc\rangle$-complete PIT~$\theta$, it requires polynomial time (w.r.t.\ the total number of symbols in~$\theta$ and~$\Pmbb$ together) to check if $\theta$~is a~full solution to~$\Pmbb$.
\end{restatable}
The reason for this discrepancy is that a~\emph{full} solution determines \emph{exactly one} probabilistic model where $\Gamma$~and~$\delta$ are true. Moreover, a~full solution~$\theta$ corresponds to a~probability distribution on~$2^\Vmbf$. Hence, to calculate the value of $\Prob(\phi)$ using~$\theta=\bigodot^k_{i=1}\Prob(\tau_i){\leq}\cvalue_i\odot\bigodot^k_{i=1}\Prob(\tau_i){\geq}\cvalue_i$, one needs to compute $\sum_{\tau_i\models_\CPL\phi}c_i$. Since $\tau_i$'s determine the values of all variables, $\tau_i\models_\CPL\phi$ can be decided in polynomial time, and this sum takes polynomial time to compute. On the other hand, \emph{sufficient} solutions can leave probabilities of some events undefined and thus do not correspond to a~unique probability distribution. This means that we need two $\np$-oracle calls: to verify $\Gamma,\eta\not\models_{\FP}\bot$ and $\Gamma,\eta\models_{\FP}\delta$. Since these calls are independent, $\DP$-membership follows.

Lastly, \emph{full} solutions are $\langle\Vmbf,\Vmc\rangle$-complete PITs, i.e., they can be exponentially long w.r.t.\ $\Var[\Pmbb]$ (and $\Emc[\Pmbb]$). Thus, we take into account lengths of both~$\theta$ and~$\Pmbb$ in Theorem~\ref{theorem:fullsolutionrecognitionpolynomial}.

Let us now turn to the complexity of \emph{solution existence}. As expected from the results in the previous section, determining the existence of \emph{sufficient} solutions will be harder than determining the existence of \emph{full} solutions. Namely, $\Sigma^\Pmsf_2$-hardness of \emph{sufficient} solution existence follows from Proposition~\ref{prop:LukisfragmentofFPLuk} and $\Sigma^\Pmsf_2$-hardness of solution existence for $\Luk$-abduction problems~\cite[Theorem~5]{InoueKozhemiachenko2025}. $\Sigma^\Pmsf_2$-membership can be obtained from Proposition~\ref{prop:FPLukNPcoNP} using a~‘guess-and-verify’ procedure.
\begin{restatable}{theorem}{sufficientsolutionexistence}\label{theorem:sufficientsolutionexistence}
Given an $\FP$ AP $\Pmbb=\langle\Gamma,\delta,\Hmsf,\Vmc\rangle$, it is $\Sigma^\Pmsf_2$-complete to determine whether it has sufficient solutions.
\end{restatable}

Full solution existence for an $\FP$ AP $\Pmbb=\langle\Gamma,\delta,\Hmsf,\Vmc\rangle$ can be seen as $\FP$-satisfiability of $\bigodot_{\gamma\in\Gamma}\gamma\odot\delta$ provided that the model satisfying $\Gamma\cup\{\delta\}$ should evaluate all probabilistic atoms over $\Vmc$, not~$[0,1]$. It is known from~\cite{HajekTulipani2001} that $\alpha\in\LProbint$ is $\FP$-satisfiable iff it is satisfiable in a~probabilistic model $\Mfrak=\langle2^{\Var(\alpha)},\mu\rangle$ s.t.\ $\mu(\{w\})>0$ holds for at most $|\Emc(\alpha)|+1$ states. Thus, we will consider full solutions $\theta=\bigodot^n_{i=1}\lambda_i$ s.t.\ $n\leq2\cdot|\Emc[\Pmbb]|+2$ (recall that $\Prob(\pi){\approx}\cvalue$ is a~shorthand for $(\Prob(\pi){\geq}\cvalue)\odot(\Prob(\pi)\leq\cvalue)$). We will call such solutions \emph{concise} to differentiate them from full solutions of arbitrary length.
\begin{definition}\label{def:concisesolution}
Given an $\FP$ AP~$\Pmbb$ and its full solution~$\theta=\bigodot^n_{i=1}\lambda_i$, we say that $\theta$~is \emph{concise} iff $n\leq2\cdot|\Emc[\Pmbb]|+2$.
\end{definition}

$\np$-membership of concise full solution existence for $\FP$ APs now follows from Theorem~\ref{theorem:fullsolutionrecognitionpolynomial}. For $\np$-hardness, we can provide a~polynomial reduction from $\CPL$-satisfiability.
\begin{restatable}{theorem}{fullsolutionexistenceNP}\label{theorem:fullsolutionexistenceNP}
Given an $\FP$ AP $\Pmbb=\langle\Gamma,\delta,\Hmsf,\Vmc\rangle$, it is $\np$-complete to determine whether it has concise full solutions.
\end{restatable}
\subsection{Prioritisation of Solutions\label{ssec:prioritisation}}
Up to now, we have been considering \emph{arbitrary} full and sufficient solutions. Note, however, that usually we have some preference relation between different solutions. In particular, a~usual requirement is that a~solution should be as simple as possible. One of the most straightforward choices to formalise simplicity in logic-based abduction is to utilise entailment. If $\zeta_1$ and $\zeta_2$ are two solutions s.t.\ $\zeta_1\models\zeta_2$ and $\zeta_2\not\models\zeta_1$, then $\zeta_2$ is simpler because it makes fewer assumptions and thus $\zeta_1$ is more preferable than~$\zeta_1$. In classical logic, this is usually formulated in terms of~\emph{$\subseteq$-minimality}~\cite{EiterGottlob1995} (given $\LCPL$-terms $\sigma$ and~$\tau$, $\sigma\models_\CPL\tau$ iff all literals in~$\tau$ belong to~$\sigma$) and interpreted as preferring solutions without \emph{redundant hypotheses}: if both $\tau\wedge h$ and $\tau$ are solutions, then~$h$ is redundant~\cite{PengReggia1990}.

An advantage of \emph{entailment-minimal} solutions is that they can be easily generalised to abduction in non-classical logics (cf., e.g.,~works by Bienvenu et al.~\shortcite{BienvenuInoueKozhemiachenko2024KR,BienvenuInoueKozhemiachenko2026} and Inoue and Kozhemiachenko~\shortcite{InoueKozhemiachenko2025}). In the framework of~$\FP$-abduction, we can employ $\models_\FP$-minimality to prioritise \emph{sufficient} solutions. On the other hand, different \emph{full} solutions are either equivalent or incompatible. Hence, we need another approach to prioritisation. Since full solutions correspond to probability distributions on~$2^\Vmbf$, it makes sense to employ the \emph{principle of maximum entropy} by Janes~\shortcite{Jaynes1957-1,Jaynes1957-2} (cf.~a~discussion by Williamson~\shortcite{Williamson2010}) to establish a~preference relation between them. The principle states that the probability distributions that best represent our knowledge are those with the highest entropy. Thus, given an $\FP$ AP, we prefer full solutions with higher entropy. Entropy was used to prioritise probabilistic explanations by Valkovsky et al.~\shortcite{ValkovskySavvinGerasimov1999} in the context of constraint programming, and Dubois et al.~\shortcite{DuboisGilioKern-Isberner2008}. There, however, the explanation was not a~probability distribution itself but an event.

We will thus consider two classes of preferred solutions to $\FP$ APs: \emph{entropy-maximal} and \emph{$\models_\FP$-minimal} (entailment-minimal). We begin with entropy-maximal full solutions.
\begin{definition}\label{def:entropy}
Let $\Pmbb=\langle\Gamma,\delta,\Hmsf,\Vmc\rangle$ be an $\FP$ AP with $\Var[\Pmbb]=\{p_1,\ldots,p_m\}$ and $\theta=\bigodot^n_{i=1}\Prob(\tau_i){\approx}\cvalue_i$ its full solution s.t.\ $c_i>0$ for all~$i\in\{1,\ldots,n\}$. The \emph{entropy of~$\theta$} ($\Hmbb(\theta)$) is defined as follows: $\Hmbb(\theta)=-\sum^n_{i=1}c_i\cdot\log_2 c_i$.
\end{definition}
\begin{definition}\label{def:CEMsolutions}
Let $\Pmbb=\langle\Gamma,\delta,\Hmsf,\Vmc\rangle$ be an $\FP$ AP, that contains~$k$ unique probabilistic atoms. We say that $\theta$ is a~\emph{concise entropy-maximal} (CEM) solution to~$\Pmbb$ iff $\theta$ is a~full solution, and there is no full solution $\theta'$ containing at most $2{\cdot}k{+}2$ PITs s.t.\ $\Hmbb(\theta){<}\Hmbb(\theta')$.
\end{definition}

Note that CEM solutions have bounded length. This is done for the following reason. Intuitively, the higher the entropy of a~full solution~$\theta$, the closer it is to a~uniform probability distribution. In other words, the higher $\Hmbb(\theta)$ is, \emph{the more~$\theta$ contains terms $\Prob(\tau){\approx}\cvalue$ with $c\neq0$}. Thus, solutions with higher entropy could be exponentially long w.r.t.\ the size of the problem. For example, if $\phi$ is a~tautology and $\Var(\phi)=\Vmbf=\{p_1,\ldots,p_n\}$, then the (non-concise) entropy-maximal explanation of $\Prob(\phi){\geq}\one$ is the uniform probability distribution on~$2^\Vmbf$ expressed as $\bigodot^{2^n}_{i=1}\Prob(\tau_i){\approx}\overline{\tfrac{1}{2^n}}$. To circumvent this problem, we will only consider ‘concise’ full solutions (recall Definition~\ref{def:concisesolution}).
\begin{example}\label{example:probabilities5}
Consider $\theta_\mathsf{rain}$. It contains~$8$ PITs (note that $\Prob(\tau){\approx}\cvalue$'s are shorthands for $\Prob(\tau){\leq}\cvalue\odot\Prob(\tau){\geq}\cvalue$) and is concise. By Definition~\ref{def:entropy}, $\Hmbb(\theta_\mathsf{rain})\approx1.685$. Recall from Example~\ref{example:probabilities2} that $\Pmbb_\mathsf{rain}$ contains~$8$ unique probabilistic atoms. Thus, we can consider full solutions to~$\Pmbb_\mathsf{rain}$ that contain up to~$18$ terms (i.e., up to~$9$ events).

Now, to produce a~probability distribution with a~higher entropy than~$\theta_\mathsf{rain}$, we ‘spread’ the probability of $\neg s\wedge\neg w\wedge\neg x\wedge c\wedge r$ to other events. Consider:
\begin{align*}
\theta^\Uparrow_\mathsf{rain}=&\Prob(s\wedge w\wedge\neg x\wedge\neg c\wedge\neg r){\approx}\overline{0.3}\odot\\
&\Prob(\neg s\wedge\neg w\wedge\neg x\wedge c\wedge r){\approx}\overline{0.35}\odot\\
&\Prob(s\wedge\neg w\wedge\neg x\wedge c\wedge r){\approx}\overline{0.15}\odot\\
&\Prob(s\wedge\neg w\wedge x\wedge c\wedge\neg r){\approx}\overline{0.1}{\odot}\\
&\Prob(\neg s\wedge w\wedge x\wedge\neg c\wedge r){\approx}\overline{0.1}
\end{align*}
Using Definition~\ref{def:entropy}, we obtain that $\Hmbb(\theta^\Uparrow_\mathsf{rain})\approx2.126$, i.e., $\Hmbb(\theta^\Uparrow_\mathsf{rain})>\Hmbb(\theta_\mathsf{rain})$. Thus, $\eta^{\Uparrow}_\mathsf{rain}$ is a~more preferable full solution to~$\Pmbb_\mathsf{rain}$ than~$\theta_\mathsf{rain}$.
\end{example}

We can now show that recognising concise entropy-maximal solutions is $\conp$-complete. Indeed, we can solve the complementary problem in~$\np$ time: given a~$\langle\Vmc,\Vmbf\rangle$-complete PIT $\theta$, we guess another (polynomially large) solution candidate $\theta'$ and check whether $\Hmbb(\theta)<\Hmbb(\theta')$ in polynomial time. If the check succeeds, $\theta$ \emph{is not a~CEM solution}. To show $\conp$-hardness, we can construct a~polynomial-time reduction from $\CPL$ validity to the recognition of concise entropy-maximal solutions.
\begin{restatable}{theorem}{maxentropyconp}\label{theorem:maxentropyconp}
Given an $\FP$ AP $\Pmbb=\langle\Gamma,\delta,\Hmsf,\Vmc\rangle$ with $\Vmbf=\Var[\Pmbb]$ and a~$\langle\Vmbf,\Vmc\rangle$-complete PIT~$\theta$, it is $\conp$-complete to determine whether $\theta$~is a~CEM solution.
\end{restatable}

Let us now consider entailment-minimal solutions. We adapt the analogous notion introduced for $\Luk$-abduction problems in~\cite{InoueKozhemiachenko2025}.
\begin{definition}\label{def:entailmentminimalsolution}
Let $\Pmbb=\langle\Gamma,\delta,\Hmsf,\Vmc\rangle$ be an $\FP$ AP. We say that $\eta$ is an \emph{$\models_\FP$-minimal} (entailment-minimal) solution iff it is a~sufficient solution, and there is no other sufficient solution~$\zeta$ s.t.\ $\eta\models_\FP\zeta$ but $\zeta\not\models_\FP\eta$.
\end{definition}

To illustrate the notion of $\models_\FP$-minimal solutions, we come back to Examples~\ref{example:probabilities2} and~\ref{example:probabilities3}. We show how to construct an $\models_\FP$-minimal solution from a~given sufficient solution.
\begin{example}\label{example:probabilities4}
Recall first a~sufficient solution $\eta_\mathsf{rain}$ to~$\Pmbb_\mathsf{rain}$. One can easily check that it is \emph{not} $\models_\FP$-minimal. We can, however, transform it into an entailment-minimal solution by increasing intervals of permitted values for each literal.

In fact, it turns out that $\Prob(s\wedge w){\approx}\overline{0.3}$ is \emph{redundant} and $\eta_0{=}\Prob(w){\geq}\overline{0.4}{\odot}\Prob(s){\approx}\overline{0.4}$ is a~solution. Clearly, \mbox{$\Gamma_\mathsf{rain},\!\eta_0\not\!\models_\FP\!\bot$.} More\-over, by Definition~\ref{def:FPLuksemantics}, we obtain from $\Gamma_\mathsf{rain}$ (cf.~Example~\ref{example:probabilities2}) that $\Imc_\Mfrak(\Prob(w)){\geq}0.4$ entails $\Imc_\Mfrak(\Prob(c)){\geq}0.6$ and $\Imc_\Mfrak(\Prob(s)){=}0.4$ entails $\Imc_\Mfrak(\Prob(r)){>}0.4$. Additionally, $\Imc_\Mfrak(\Prob(s\wedge w))\leq0.4$ because $\Imc_\Mfrak(\Prob(s)){=}0.4$, and hence $\Imc_\Mfrak(\Prob(c\vee r))\in[0.6;0.8]$. Thus, $\Imc_\Mfrak(\Prob(c\wedge r)){>}0.2$, i.e., $\Gamma_\mathsf{rain},\eta_0\models_\FP\Prob(c\wedge r){\geq}\overline{\tfrac{1}{5}}$, as required. Now observe that neither of the following PITs solve~$\Pmbb_\mathsf{rain}$:
\begin{align*}
\eta_1&=\Prob(w){\geq}\overline{0.39}\odot\Prob(s){\approx}\overline{0.4}\\
\eta_2&=\Prob(w){\geq}\overline{0.4}\odot\Prob(s){\geq}\overline{0.4}\odot\Prob(s){\leq}\overline{0.41}\\
\eta_3&=\Prob(w){\geq}\overline{0.4}\odot\Prob(s){\geq}\overline{0.39}\odot\Prob(s){\leq}\overline{0.4}
\end{align*}
Thus, we cannot increment the sets of permitted values of PILs in~$\eta_0$. Note, however, that $\Prob(s){\approx}\overline{0.4}$ is a~shorthand for $\Prob(s){\geq}\overline{0.4}\odot\Prob(s){\leq}\overline{0.4}$. Hence, we can weaken $\Prob(s){\leq}\overline{0.4}$ to $\Prob(s\wedge w){\leq}\overline{0.4}$. One can check that the resulting PIT
\begin{align*}
\eta^{\min}_\mathsf{rain}&=\Prob(w){\geq}\overline{0.4}\odot\Prob(s){\geq}\overline{0.4}\odot\Prob(s\wedge w){\leq}\overline{0.4}
\end{align*}
is a~sufficient solution to~$\Pmbb_\mathsf{rain}$. Moreover, it is indeed an $\models_\FP$-minimal solution since we can neither increase intervals nor weaken the remaining events. 
\end{example}

Now notice that, in contrast to classical and Łukasiewicz logics, the satisfiability and entailment of PITs are not polynomially decidable. It is easy to see that $\FP$-satisfiability of arbitrary probabilistic atoms is $\np$-complete. By Definition~\ref{def:FPLuksemantics}, $\phi$~is $\CPL$-satisfiable iff $\Prob(\phi)$ is $\FP$-satisfiable. Even the $\FP$-satisfiability of PITs is $\np$-complete (and entailment is $\conp$-complete) because we can encode DNF-non-validity in~$\CPL$ using PITs.
\begin{restatable}{proposition}{PITcomplexity}\label{prop:PITcomplexity}~
\begin{enumerate}[noitemsep,topsep=0pt]
\item Given a~PIT $\zeta=\bigodot^n_{i=1}(\Prob(\tau_i)\diamond\cvalue_i)$, it is $\np$-complete to check whether $\zeta$ is $\FP$-satisfiable.
\item Given two PITs $\zeta_1$ and $\zeta_2$ it is $\conp$-complete to check whether $\zeta_1\models_\FP\zeta_2$.
\end{enumerate}
\end{restatable}

Still, even though entailment between PITs is not polynomially tractable, recognising $\models_\FP$-minimal solutions is $\DP$-complete. For $\DP$-hardness, we construct a~polynomial reduction from the recognition of entailment-minimal solutions of \emph{$\Luk$-abduction problems}, which is $\DP$-complete~\cite[Theorem~3]{InoueKozhemiachenko2025} to the recognition of $\models_\FP$-minimal solutions. For $\DP$-membership, we show that given a~PIT $\eta$ and an $\FP$ AP~$\Pmbb$, there are only polynomially many ‘next weakest’ terms that we need to check. If none of them is a~sufficient solution to~$\Pmbb$, but $\eta$~is, then $\eta$~is a~$\models_\FP$-minimal solution. These ‘next weakest’ PITs can be built as shown in Example~\ref{example:probabilities4} where we transformed~$\eta_\mathsf{rain}$ into~$\eta^{\min}_\mathsf{rain}$.
\begin{restatable}{theorem}{FPminimalrecognitionDP}\label{theorem:FPminimalrecognitionDP}
Given an $\FP$ AP $\Pmbb=\langle\Gamma,\delta,\Hmsf,\Vmc\rangle$ and a~PIT~$\eta$, it is $\DP$-complete to check whether it is an~$\models_\FP$-minimal solution to~$\Pmbb$.
\end{restatable}
\section{Abduction in Fragments of~$\FP$\label{sec:fragments}}
As we saw in the previous section, abduction in $\FP$ is intractable in the general case (especially when it comes to \emph{sufficient solutions}). Thus, it makes sense to consider fragments of~$\FP$ where abductive reasoning is polynomially decidable or at least simpler than in the general case.

In this section, we will consider two fragments of $\FP$ that generalise the \emph{simple clause fragment} of Łukasiewicz logic~\cite{BofillManyaVidalVillaret2019} and the \emph{cover-free fragment} of~$\Lukinttriangle$~\cite{InoueKozhemiachenko2025}. Note, however, that we need to restrict not only the shape of outer formulas, but also the shape of inner formulas. Indeed, as we noted in Section~\ref{sec:FPabduction}, satisfiability of a~single formula $\Prob(\phi)$ or $\Prob(\phi)\diamond\cvalue$ is $\np$-hard if $\phi$ is arbitrary. Moreover, the events should be syntactically simple, yet expressive enough to reason about probabilities of non-trivial events. The next definition presents such a~restriction.
\begin{definition}[Chained inner-positive theories]\label{def:positivechains}~
\begin{itemize}[noitemsep,topsep=0pt]
\item $\Phi\subseteq\LCPL$ is \emph{chained positive} (CP) iff
\begin{itemize}[noitemsep,topsep=0pt]
\item $\pi=\bigwedge^m_{i=1}p_i$ or $\pi=\bigvee^n_{i=1}q_i$ for every $\pi\in\Phi$;
\item for every $\pi,\sigma\in\Phi$, it holds that either $\pi\models_\CPL\sigma$, or $\sigma\models_\CPL\pi$, or $\Var(\pi)\cap\Var(\sigma)=\varnothing$.
\end{itemize}
\item $\Gamma\subseteq\LProbint$ is \emph{chained inner-positive} (CIP) if $\Emc[\Gamma]$ is chained positive.
\end{itemize}
\end{definition}
\begin{example}\label{example:positivechains}
Recall the theory $\Gamma_\mathsf{rain}$ from Example~\ref{example:probabilities2}. Observe that it is \emph{not CIP}. Indeed, $\Emc[\Gamma]$ contains $c\vee r$, $c$, and~$r$. Clearly, $c\models_\CPL c\vee r$, $r\models_\CPL c\vee r$, but $c\not\models_\CPL r$ and $r\not\models_\CPL c$. However $\Gamma'_\mathsf{rain}\cup\{\delta_\mathsf{rain}\}$ is CIP for $\Gamma'_\mathsf{rain}=\Gamma_\mathsf{rain}\setminus\{\neg\triangle(\Prob(r)\rightarrow\Prob(s))\}$. Indeed, we have two ‘chains’ of events in~$\Gamma'_\mathsf{rain}$: $s\wedge w\models_\CPL w$ and $c\wedge r\models_\CPL c\models_\CPL c\vee r$. Note, however, that $\eta_\mathsf{rain}$ (cf.~Example~\ref{example:probabilities3}) is \emph{not} a~solution to $\Pmbb'_\mathsf{rain}=\langle\Gamma'_\mathsf{rain},\delta_\mathsf{rain},\Hmsf_\mathsf{rain},\Vmc_\mathsf{rain}\rangle$. In fact, since $r\notin\Emc[\Gamma'_\mathsf{rain}]$, $\Pmbb'_\mathsf{rain}$ \emph{has no sufficient solutions}. Still, $\theta_\mathsf{rain}$ is a~full solution to~$\Pmbb'_\mathsf{rain}$.
\end{example}

An important property of CIP theories is that they can be translated into~$\LLukint$ in polynomial time only. In fact, given a~CIP theory~$\Gamma$, it suffices to construct an $\LLukint$ theory $\Psi_\Gamma$ that encodes (classical) entailment between events in~$\Gamma$. In this case, an $\Luk$-valuation satisfying~$\Psi_\Gamma$ will induce a~map coherent with a~probability measure on $2^{2^{\Var[\Gamma]}}$.

The next definition is adapted from~\cite{deFinetti2017}.
\begin{definition}\label{def:coherence}
Let $\Embb\subseteq\LCPL$ be finite, $\pfrak:\Embb\rightarrow[0,1]$, and $\Mfrak=\langle2^{\Var[\Embb]},\mu\rangle$ be a~probabilistic model. We say that \emph{$\mu$~is coherent with~$\pfrak$} iff $\mu(\|\phi\|_\Mfrak)=\pfrak(\phi)$ for all $\phi\in\Embb$. The map $\pfrak:\Embb\rightarrow[0,1]$ is \emph{coherent} if there exists a~probabilistic model $\Mfrak=\langle2^{\Var[\Embb]},\mu\rangle$ such that $\mu$ is coherent with~$\pfrak$.
\end{definition}
\begin{restatable}{proposition}{simpletranslation}\label{prop:simpletranslation}
For every finite CIP theory $\Gamma\subseteq\LProbint$, there is an $\LLukint$-theory $\Psi_\Gamma$ s.t.\
\begin{enumerate}[noitemsep,topsep=0pt]
\item $\Psi_\Gamma$ is $\Luk$-satisfiable iff $\Gamma$ is $\FP$-satisfiable;
\item $\Psi_\Gamma$ is constructible in polynomial time w.r.t.\ the size of~$\Gamma$.
\end{enumerate}
\end{restatable}
\subsection{Simple Clause Fragment\label{ssec:PSC}}
Let us now consider the probabilistic simple clause fragment of~$\FP$. Recall from~\cite{BofillManyaVidalVillaret2019} that (propositional) simple clauses are formulas of the form $\bigoplus^n_{i=1}l_i$ with $l_i$'s being propositional variables and their negations. As one can check from Definition~\ref{def:Luksemantics}, $\bigoplus^n_{i=1}l_i$ can be equivalently represented as $\bigodot^{k}_{i=1}\neg l_i\rightarrow\bigoplus^n_{j=n-k}l_j$ for every $k\leq n$. Thus, simple clauses can be used to express fuzzy (disjunctive) rules. Moreover, in contrast to classical disjunctive clauses, $\Luk$-satisfiability of sets of simple clauses is \emph{polynomially decidable}, whence, simple-clause abduction turns out to be $\np$-complete~\cite{InoueKozhemiachenko2025}, while classical abduction in theories represented as sets of disjunctive clauses is $\Sigma^\Pmsf_2$-hard in general~\cite{EiterGottlob1995}.

Thus, in the context of \emph{probabilistic} abduction, it makes sense to consider an $\FP$-counterpart of simple clauses. The following definitions present \emph{probabilistic} simple clauses (PSCs) and PSC abduction problems.
\begin{definition}[Probabilistic simple clauses]\label{def:PSC}~
\begin{itemize}[noitemsep,topsep=0pt]
\item A~\emph{probabilistic simple clause} (PSC) is a~formula of the form $\bigoplus^m_{i=1}\Prob(\sigma_i)\oplus\bigoplus^n_{i=1}\neg\Prob(\tau_i)$ with $\sigma$ and $\tau$ being conjunctions or disjunctions of variables.
\item A~\emph{PSC theory} is a~CIP theory $\Gamma$ s.t.\ every $\gamma\in\Gamma$ is either a~PSC or has a~form $\Prob(\tau)\diamond\cvalue$ with $\tau$ being a~conjunction or a~disjunction of variables.
\end{itemize}
\end{definition}

Let us illustrate the expressivity of PSC theories with an example. Note, in particular, that we can formalise some statements about probabilities of events containing~$\neg$. Moreover, formulas $\Prob(\tau)\diamond\cvalue$ \emph{do not need to be PILs}.
\begin{example}\label{example:PSCtheory}
Consider the following statements:
\begin{description}[noitemsep,topsep=0pt]
\item[$\alpha$:] the probability of precipitation (rain, snow, or hailstorm) tomorrow is $70\%$;
\item[$\beta$:] the probability of the rain tomorrow is at least as high as the probability of neither snow nor rain happening tomorrow.
\end{description}
We set $\alpha=\alpha_1\odot\alpha_2$ with $\alpha_1=\Prob(r\vee s\vee h){\geq}\overline{0.7}$ and $\alpha_2=\Prob(r\vee s\vee h){\leq}\overline{0.7}$ and $\beta=\Prob(\neg r\wedge\neg s)\rightarrow\Prob(r)$. In its current form, $\{\alpha,\beta\}$ is not a~PSC theory because $\Emc[\{\alpha,\beta\}]$ is not chained positive. Note, however, that by Definition~\ref{def:FPLuksemantics}, $\Prob(\neg\phi)$ and $\neg\Prob(\phi)$ are equivalent. Thus, $\beta$ is equivalent to $\beta'=\Prob(r\vee s)\oplus\Prob(r)$ and $\{\alpha_1,\alpha_2,\beta'\}$ is a~PSC theory.
\end{example}
\begin{definition}\label{def:PSCabduction}
A~\emph{PSC abduction problem} is an $\FP$ AP $\Pmbb=\langle\Gamma,\delta,\Hmsf,\Vmc\rangle$ s.t.\ $\Gamma\cup\{\delta\}$ is a~PSC theory and $\Emc[\Pmbb]$ is CP.
\end{definition}
Note that the observation of a~PSC AP can have various types, depending on the constraints we want to explain using the theory. Furthermore, we need to list $\LCPL$-terms in~$\Hmsf$ and cannot allow arbitrary conjunctions of hypotheses: while $\{a,b\}$ is CP, $\{a,b,a\wedge b\}$ is not. Defining~$\Hmsf$ as a~set of terms (not literals) and demanding that $\Emc[\Pmbb]$ be chained positive ensures that given a~\emph{sufficient} solution $\eta$ for $\Pmbb$, $\Gamma\cup\{\eta,\delta\}$ is a~CIP theory. We can use this together with Propositions~\ref{prop:LukisfragmentofFPLuk} and~\ref{prop:simpletranslation} to show the following statement.
\begin{restatable}{theorem}{PSCrecognitionPcomplete}\label{theorem:PSCrecognitionPcomplete}
Given a~PSC AP $\Pmbb=\langle\Gamma,\delta,\Hmsf,\Vmc\rangle$ and a~PIT~$\eta$, it is $\Pmsf$-complete under logspace reductions to check if $\eta$ is an ($\FP$-minimal) sufficient solution of~$\Pmbb$.
\end{restatable}
As expected from the previous statement, determining the existence of sufficient solutions is $\np$-complete. Somewhat surprisingly, the existence of concise \emph{full solutions} is also $\np$-complete as in the general case. This is because we can reduce classical satisfiability of CNFs to the existence of full solutions to PSC APs.
\begin{restatable}{theorem}{PSCsolutionexistence}\label{theorem:PSCsolutionexistence}
Given a~PSC AP $\Pmbb=\langle\Gamma,\delta,\Hmsf,\Vmc\rangle$, it is:
\begin{enumerate}[noitemsep,topsep=0pt]
\item $\np$-complete to check if it has \emph{sufficient} solutions;
\item $\np$-complete to check if it has \emph{concise full} solutions.
\end{enumerate}
\end{restatable}
\subsection{Interval Clause Fragment\label{ssec:PIC}}
Let us now consider the probabilistic interval clause fragment of~$\FP$. Propositional interval clauses (i.e., formulas of the form $\pi_1\oplus\ldots\oplus\pi_n$ with $\pi$'s being \emph{interval literals}) were introduced by Inoue and Kozhemiachenko~\shortcite{InoueKozhemiachenko2025} to formalise statements such as ‘if the value of~$p$ is at least~$0.5$, then the value of~$q$ is at most~$0.6$’~--- $p{\geq}\overline{0.5}\rightarrow q{\leq}\overline{0.6}$ or, equivalently $p{<}\overline{0.5}\oplus q{\leq}\overline{0.6}$. In this section, we introduce the $\FP$-counterparts of interval clauses and investigate the complexity of abductive reasoning with \emph{probabilistic interval clauses}.
\begin{definition}\label{def:PIC}
A~\emph{probabilistic interval clause} (PIC) is a~formula $\bigoplus^n_{i=1}\lambda_i$ with every $\lambda$ having the form $\Prob(\phi)\diamond\cvalue$ with $\phi\in\LCPL$. A~\emph{PIC theory} is a~set of PICs.
\end{definition}

It is known from~\cite{InoueKozhemiachenko2025} that the complexity of abductive reasoning with interval-clause theories is the same as with arbitrary $\LLukint$-theories. Thus, by Proposition~\ref{prop:LukisfragmentofFPLuk}, the same would hold w.r.t.\ abductive reasoning with PIC theories. Still, the interval clause fragment of~$\Lukinttriangle$ can be restricted to the so-called \emph{cover-free} fragment, over which solution recognition was polynomial and solution existence $\np$-complete.
\begin{definition}\label{def:coverfreeintervalclauses}
Let $\Phi\subseteq\LLukint$ be a~set of interval clauses represented as $\bigodot^m_{i=1}(p_i\diamond\cvalue_i)\!\rightarrow\!(q\diamond\dvalue)$ and $\bigodot^m_{i=1}(p_i\diamond\cvalue_i)\!\rightarrow\!\bot$. $\Phi$ is \emph{cover-free} (CF) if no pair of interval literals $\pi$ and $\pi'$ over the same variable occurs in~$\Phi$ s.t.\ $\Luk\models\pi\oplus\pi'$.
\end{definition}

By Proposition~\ref{prop:LukisfragmentofFPLuk}, it would follow from~\cite[Theorem~16]{InoueKozhemiachenko2025} that solution existence is $\np$-hard for $\FP$ APs $\Pmbb=\langle\Phi^\Prob,\lambda^\Prob,\Hmsf,\Vmc\rangle$ with $\Phi$ being a~CF set of \emph{interval clauses} and $\lambda$ an interval literal. It is thus instructive to see whether we can find a~subclass of CF theories where solution existence can be determined in polynomial time. The definition below introduces \emph{short} probabilistic cover-free (SPCF) theories. We use this name to indicate that $n=1$ in the implications of the form $\bigodot^n_{i=1}\lambda_i\rightarrow\lambda$ (i.e., clauses with a~non-empty head must be \emph{binary}).
\begin{definition}\label{def:SPCF}
An SPCF theory is a~CIP theory $\Gamma$ s.t.
\begin{itemize}[noitemsep,topsep=0pt]
\item $\Gamma$ is a~set of PICs represented as $\bigodot^n_{i=1}\Prob(\pi_i)\diamond\cvalue_i\rightarrow\bot$ or $\Prob(\pi)\diamond\cvalue\rightarrow\Prob(\varrho)\diamond\dvalue$;
\item for every pair of probabilistic atoms $\lambda$ and $\lambda'$ occurring in~$\Gamma$, it holds that $\FP\not\models\lambda\oplus\lambda'$.
\end{itemize}
\end{definition}
Consider now the following example.
\begin{example}\label{example:SPCF}
Assume that $\Gamma$ contains the following relations between the frequencies of different events. If black ice~($b$) or thick fog~($f$) occurs at least $40\%$ of the time, then traffic accidents~($a$) will happen at least every other day. If traffic accidents or snowdrifts~($s$) happen at least twice a~week, then major traffic jams~($j$) will occur at least once a~week.

We can formalise~$\Gamma$ as follows: $\Prob(b{\vee}f){\geq}\overline{\tfrac{2}{5}}\rightarrow\Prob(a){\geq}\overline{\tfrac{1}{2}}$, $\Prob(a{\vee}s){\geq}\overline{\tfrac{2}{7}}\rightarrow\Prob(j){\geq}\overline{\tfrac{1}{7}}$. Now, to explain why traffic jams happen every week, we can, for example, assume that thick fog occurred at least $40\%$ of the time: $\eta=\Prob(f)\geq\overline{\tfrac{2}{5}}$. Note that $\Gamma,\eta\consvDashFPLuk\Prob(j){\geq}\overline{\tfrac{1}{7}}$ and $\Gamma\cup\{\eta,\Prob(j){\geq}\overline{\tfrac{1}{7}}\}$ is SPCF.
\end{example}
The next definition introduces the class of $\FP$ abduction problems whose theories are SPCF. Note that we demand that the set of \emph{all events} occurring in the problem be~CP.
\begin{definition}\label{def:SPCFabduction}
An~\emph{SPCF abduction problem} is an $\FP$ AP $\Pmbb=\langle\Gamma,\Prob(\chi)\diamond\dvalue,\Hmsf,\Vmc\rangle$ s.t.\ $\Gamma\cup\{\Prob(\chi)\diamond\dvalue\}$ is an SPCF theory and $\Emc[\Pmbb]$ is chained positive.
\end{definition}

The polynomial decidability of the sufficient solution existence follows from the fact that the problem of determining the existence of solutions to SPCF APs is reducible to determining the existence of solutions to classical abduction problems whose theories contain so-called $\ihsbminus$-clauses (literals~--- $p$ and $\neg q$~--- and implications of the form $p\rightarrow q$ and $\bigwedge^n_{i=1}p_i\rightarrow\bot$) proposed by Khanna et al.~\shortcite{KhannaSudanTrevisan2002}. As determining the existence of solutions to $\ihsbminus$ abduction problems is \emph{polynomially tractable}~\cite[Proposition~12]{CreignouZanuttini2006}, the existence of sufficient solutions to SPCF APs can also be decided in polynomial time. On the other hand, determining the existence of \emph{full} solutions to SPCF APs is $\np$-hard because we can construct a~reduction from $\CPL$-satisfiability of CNFs. Thus, SPCF theories manifest an unusual behaviour: it is \emph{harder} to check whether a~full solution exists than whether a~sufficient solution exists.

The intuitive reason behind this is that solutions to $\CPL$ abduction problems with $\ihsbminus$ theories and observations in the form of a~single variable can be w.l.o.g.\ assumed to contain one literal only. Thus (after translating the original problem into its $\CPL$ counterpart), it suffices to check hypotheses one by one. Moreover, for any candidate~$\eta$ for a~sufficient solution to~$\Pmbb{=}\langle\Gamma,\Prob(\chi)\diamond\dvalue,\Hmsf,\Vmc\rangle$, $\Gamma\cup\{\eta,\Prob(\chi)\diamond\dvalue\}$ is a~CIP theory. However, to check if a~full solution exists, one needs to verify whether there is a~suitable $\langle\Vmbf,\Vmc\rangle$-complete term, which is independent from~$\Hmsf$ and is not necessarily CP.
\begin{restatable}{theorem}{SPCFsolutionexistence}\label{theorem:SPCFsolutionexistence}
For a SPCF AP $\Pmbb=\langle\Gamma,\!\Prob(\chi)\diamond\dvalue,\!\Hmsf,\!\Vmc\rangle$, it is
\begin{enumerate}[noitemsep,topsep=0pt]
\item decidable in polynomial time if it has sufficient solutions;
\item $\np$-complete to check if it has concise full solutions.
\end{enumerate}
\end{restatable}
\section{Modelling Most Likely Explanations in~$\FP$\label{sec:probabduction}}
Until now, we have interpreted probabilistic explanations as statements \emph{about probabilities} that explain the observed frequency of events, and showed how one can use~$\FP$ to model such explanations. A~more traditional perspective on probabilistic explanations is to consider \emph{the most likely explanations} of events. More formally, one can represent this task as follows: given a~theory~$\Phi$, an observation~$\chi$, and probability assignments of (some) events, one has to come up with an explanation $\tau$ s.t.\ there is no other explanation $\sigma$ whose probability is higher than that of~$\tau$. In this section, we will demonstrate how to model this reasoning task in~$\FP$.
\begin{definition}\label{def:probabilisticAP}
A~\emph{probabilistic abduction problem} (PrAP) is a~tuple $\Pmbb_\pfrak=\langle\Phi,\chi,\Hmsf,\Embb,\pfrak\rangle$ s.t.\ 
\begin{itemize}[noitemsep,topsep=0pt]
\item $\Phi\cup\{\chi\}\cup\Embb\subseteq\LCPL$ is finite and $\Var[\Embb]\subseteq\Var[\Phi\cup\{\chi\}]$;
\item $\Hmsf$ is a~finite set of literals s.t.\ $\Var[\Hmsf]\subseteq\Var[\Phi\cup\{\chi\}]$;
\item $\pfrak:\Embb\rightarrow[0,1]_\Qmbb$.
\end{itemize}
We call $\Phi$, $\chi$, $\Hmsf$, $\Embb$, and $\pfrak$ \emph{theory}, \emph{observation}, \emph{hypotheses}, \emph{events}, and \emph{value assignment}, respectively.
\end{definition}
Note that PrAPs extend the (propositional fragments of the) frameworks in~\cite{Poole1993,Sato1995}, as we allow arbitrary (not just Horn) theories. Besides, PrAPs are similar to abduction problems with preferences on hypotheses~\cite[\S4.4]{EiterGottlob1995}. The difference is that $\pfrak$~can be undefined on some hypotheses or their conjunctions and that we use numerical values, not a~qualitative preference.

To determine solutions to PrAPs, we will consider probabilistic models whose measures are coherent with~$\pfrak$. For simplicity, we will represent $\Phi$ as $\bigwedge_{\phi\in\Phi}\phi$ and w.l.o.g.\ assume that theories of PrAPs contain only one formula.
\begin{definition}[Solutions to PrAPs]\label{def:PrAPsolution}
For a~PrAP $\Pmbb_\pfrak=\langle\{\phi\},\chi,\Hmsf,\Embb,\pfrak\rangle$ and $\Vmbf=\Var[\Pmbb_\pfrak]$, we define the following.
\begin{itemize}[noitemsep,topsep=0pt]
\item A~\emph{solution} to~$\Pmbb_\pfrak$ is an $\LCPL$-term $\tau$ built from~$\Hmsf$ s.t.\ $\phi,\tau\models_\CPL\chi$ and there is a~probabilistic model $\Mfrak=\langle2^\Vmbf,\mu\rangle$ s.t.\ $\mu$ is coherent with~$\pfrak$ and $\mu\left(\left\|\phi\wedge\tau\right\|_\Mfrak\right){>}0$.
\item A~solution $\tau$ is \emph{preferred} iff for every other solution~$\sigma$, there is a~probabilistic model~$\Mfrak=\langle2^\Vmbf,\mu\rangle$ s.t.\ $\mu$ is coherent with~$\pfrak$ and $\mu(\|\tau\|_\Mfrak)\geq\mu(\|\sigma\|_\Mfrak)$.%
\end{itemize}
\end{definition}
Note that $\CPL$ abduction problems can be seen as PrAPs with $\Embb=\{\phi\}$ and $\pfrak(\phi)=1$ (i.e., every statement in the theory has probability~$1$). Indeed, there is $\Mfrak=\langle2^\Vmbf,\mu\rangle$ s.t.\ $\mu\left(\left\|\phi\wedge\tau\right\|_\Mfrak\right)>0$ iff $\phi\wedge\tau$ is $\CPL$-sa\-tis\-fi\-able. Observe also, the notion of preferred solutions refines that of $\subseteq$-minimal solutions: if $\sigma$ and $\tau$ are $\subseteq$-minimal, we can still prefer $\tau$ if there is a~probabilistic model $\Mfrak=\langle2^\Vmbf,\mu\rangle$ s.t.\ $\mu(\|\tau\|_\Mfrak)\geq\mu(\|\sigma\|_\Mfrak)$. Solutions to PrAPs can also be characterised in terms of \emph{conditional probability} as follows.
\begin{definition}\label{def:conditionalprobability}
Let $\Mfrak{=}\langle2^\Vmbf,\mu\rangle$ be a~probabilisitic model, $\Var[\{\phi,\chi\}]{\subseteq}\Vmbf$, and $\mu(\|\chi\|_\Mfrak){>}0$. The \emph{conditional probability of~$\phi$ given~$\chi$} is defined as $\Probfrak_\mu(\phi{\mid}\chi)=\tfrac{\mu(\|\phi\wedge\chi\|_\Mfrak)}{\mu(\|\chi\|_\Mfrak)}$.
\end{definition}
\begin{restatable}{proposition}{PrAPsolutionsconditionalprobability}\label{prop:PrAPsolutionsconditionalprobability}
Let $\Pmbb_\pfrak=\langle\{\phi\},\chi,\Hmsf,\Embb,\pfrak\rangle$ be a~PrAP, $\Vmbf=\Var[\Pmbb_\pfrak]$, and $\tau$ an $\LCPL$-term. Then $\tau$ is a~solution to~$\Pmbb_\pfrak$ iff the following conditions hold:
\begin{enumerate}[noitemsep,topsep=0pt]
\item $\mu\left(\left\|\phi\wedge\tau\right\|_\Mfrak\right)\leq\mu(\|\chi\|_\Mfrak)$ in all probabilistic models;
\item there is a~probabilistic model $\Mfrak'=\langle2^\Vmbf,\mu'\rangle$ s.t.\ $\mu'$ is coherent with~$\pfrak$ and $\Probfrak_{\mu'}\left(\chi\mid\phi\wedge\tau\right)=1$.
\end{enumerate}
\end{restatable}

Let us now establish a~translation of classical probabilistic abduction into~$\FP$. As solutions to PrAPs are \emph{events} and not \emph{assertions about their probabilities}, we will show the correspondence between pairs of PrAPs and their solutions and $\FP$-theories that express the probabilities coherent with~$\pfrak$.
\begin{definition}\label{def:probabilisticAPFPLukcounterpart}
Let $\Embb\subseteq\LCPL$ be finite and $\pfrak:\Embb\rightarrow[0,1]_\Qmbb$. The \emph{$\FP$-counterpart of~$\pfrak$} is the following set of $\LProbint$-formulas:
\begin{align*}
\Xi_\pfrak&=\{\Prob(\phi){\approx}\cvalue\mid\phi\in\Embb,\pfrak(\phi)=c\}
\end{align*}
\end{definition}
\begin{restatable}{theorem}{PrAPtoFPLuk}\label{theorem:PrAPtoFPLuk}
Let $\Pmbb_\pfrak=\langle\{\phi\},\chi,\Hmsf,\Embb,\pfrak\rangle$ be a~PrAP, $\tau$ an $\LCPL$-term, and $\Vmbf=\Var[\Pmbb_\pfrak]$. Then it holds that:
\begin{enumerate}[noitemsep,topsep=0pt]
\item $\tau$ is a~solution to~$\Pmbb_\pfrak$ iff $\Xi_\pfrak,\neg\triangle\neg\Prob(\phi\wedge\tau)\not\models_\FP\bot$ and $\FP\models\Prob(\phi\wedge\tau)\rightarrow\Prob(\chi)$;
\item $\tau$ is a~preferred solution to~$\Pmbb_\pfrak$ iff, in addition to that, $\Xi_\pfrak,\Prob(\sigma)\rightarrow\Prob(\tau)\not\models_\FP\bot$ for every solution~$\sigma$ to~$\Pmbb_\pfrak$.
\end{enumerate}
\end{restatable}
Let us illustrate the correspondence from Theorem~\ref{theorem:PrAPtoFPLuk}.
\begin{example}\label{example:PrAPtranslation}
Let $\Pmbb_\pfrak=\langle\Phi,s,\Hmsf,\Embb,\pfrak\rangle$ and, additionally, $\Phi=\{p\rightarrow s,(q\wedge r)\rightarrow s\}$, $\Hmsf{=}\{p,q,r\}$, $\Embb=\Hmsf$, $\pfrak(p)=\tfrac{1}{2}$, $\pfrak(q)=\tfrac{1}{3}$, and $\pfrak(r)=\tfrac{1}{4}$. One can see that $p$~and~$q\wedge r$ are solutions to~$\Pmbb_\pfrak$. Moreover, $p$~is the preferred solution, and $q\wedge r$ \emph{is not} because $\mu(\|p\|)\geq\mu(\|q\wedge r\|)$ in all probabilistic models coherent with~$\pfrak$.

Now, set $\Xi_\pfrak=\{\Prob(p){\approx}\overline{\tfrac{1}{2}},\Prob(q){\approx}\overline{\tfrac{1}{3}},\Prob(r){\approx}\overline{\tfrac{1}{4}}\}$ and let $\phi_\Pmbb=\bigwedge_{\phi\in\Phi}$. This gives us $\Xi_\pfrak,\neg\triangle\neg\Prob(\phi_\Pmbb\wedge p)\not\models_\FP\bot$, $\FP\models\Prob(\phi_\Pmbb{\wedge} p){\rightarrow}\Prob(s)$, and $\Xi_\pfrak,\Prob(q{\wedge}r){\rightarrow}\Prob(p)\not\models_\FP\bot$, as required.
\end{example}

Using Theorem~\ref{theorem:PrAPtoFPLuk}, we can obtain the following complexity evaluations of classical probabilistic abduction.
\begin{restatable}{theorem}{PrAPcomplexity}\label{theorem:PrAPcomplexity}
Given a~PrAP $\Pmbb_\pfrak$ and an $\LCPL$-term~$\tau$, it is:
\begin{enumerate}[noitemsep,topsep=0pt]
\item $\DP$-complete to check if $\tau$ is a~solution to~$\Pmbb_\pfrak$;
\item in $\Pi^\Pmsf_2$ to check if $\tau$ is a~preferred solution to~$\Pmbb_\pfrak$;
\item $\Sigma^\Pmsf_2$-complete to check if $\Pmbb_\pfrak$~has solutions.
\end{enumerate}
\end{restatable}

Notice that \emph{cardinality-minimal} ($\leq$-minimal) solutions (solutions with the fewest hypotheses) to classical abduction problems are exactly the preferred solutions w.r.t.\ a~uniform probability distribution. If $\pfrak$~is a~uniform probability distribution, then $\pfrak(\sigma)\leq\pfrak(\tau)$ iff $|\sigma|\leq|\tau|$. Recognition of $\leq$-minimal solutions is known to be \emph{$\Pi^\Pmsf_2$-complete}~\cite{EiterGottlob1995}. Unfortunately, there seems to be no way to construct a~polynomial reduction from the recognition of $\leq$-minimal solutions to preferred solutions because a~uniform probability distribution cannot be concisely represented.

Furthermore, note from Definition~\ref{def:probabilisticAP} that our PrAPs are very general: a~value assignment in~$\Pmbb_\pfrak$ is not necessarily a~probability distribution on $2^{\Var[\Pmbb_\pfrak]}$. In most approaches to probabilistic abduction, however, it is assumed that $\pfrak$ is a~probability distribution. Formally, this means that $\Embb=\{\tau_1,\ldots,\tau_n\}$ is a~set of $\LCPL$-terms s.t.\ $\Var(\tau_i)=\Var[\Pmbb_\pfrak]$ for all $i\in\{1,\ldots,n\}$ and $\sum^n_{i=1}\pfrak(\tau_i)=1$.

Under this assumption, solution recognition becomes $\conp$-complete. For $\conp$-membership, observe that verifying $\phi,\tau\models_\CPL\chi$ is $\conp$-complete, but checking that $\Xi_\pfrak,\neg\triangle\neg\Prob(\phi\wedge\tau)\not\models_\FP\bot$ takes polynomial time w.r.t.\ the size of~$\pfrak$. To show $\conp$-hardness, we construct the following polynomial reduction from $\CPL$-validity. Given $\chi\in\LCPL$, $p\notin\Var(\chi)$, $\Hmsf{=}\{p\}$, $\Embb{=}\{p\wedge\bigwedge_{q\in\Var(\chi)}q\}$, and $\pfrak(p\wedge\bigwedge_{q\in\Var(\chi)}q){=}1$, $p$~is a~solution to $\Pmbb_\pfrak=\langle\{p\vee\neg p\},\chi,\Hmsf,\Embb,\pfrak\rangle$ iff $\chi$~is $\CPL$-valid.
\begin{restatable}{theorem}{conprecognitionPrAP}\label{theorem:conprecognitionPrAP}
Given a~PrAP $\Pmbb_\pfrak=\langle\{\phi\},\chi,\Hmsf,\Embb,\pfrak\rangle$ s.t.\
\begin{itemize}[noitemsep,topsep=0pt]
\item $\Embb\subseteq\{\bigwedge_{p\in X}p\wedge\bigwedge_{q\notin X}\neg q{\mid} X\subseteq\Var[\{\phi,\chi\}]\}$, and
\item $\sum_{\sigma\in\Embb}\pfrak(\sigma){=}1$
\end{itemize}
and an $\LCPL$-term $\tau$, it is $\conp$-complete to check if $\tau$ is a~solution to~$\Pmbb_\pfrak$.
\end{restatable}
\section{Conclusion and Future Work\label{sec:conclusion}}
We considered probabilistic abduction formalised in a~fuzzy probabilistic logic $\FP$. Our study provides a~comprehensive picture of the complexity of full and sufficient solution recognition and existence. We also formalised \emph{classical} probabilistic abduction in~$\FP$, thus connecting two interpretations of probabilistic explanations: (1)~the explanation \emph{why the given event has the observed probability} and (2)~a~most likely explanation \emph{why the given event occurs}.

We also observe that the results in Section~\ref{sec:complexity} can be straightforwardly extended to abductive reasoning with \emph{conditional} probabilities. For that, one can consider abduction in the logic $\FPRLtriangle$~\cite{Flaminio2007} that extends~$\FP$ with multiplication by a~rational number from~$[0,1]$. Since $\FPRLtriangle$ is $\np$-complete, the complexity is preserved. Even more expressive propositional ground logics can be employed to deal, for instance, with product operation. In this case, we should move to the probability logic $\FP(\Luk\Pi)$ based on the combination of Łukasiewicz and Product logics by Esteva et al.~\shortcite{EstevaGodoMontagna2001}. Besides the increased expressive power, it is known from~\cite{HajekTulipani2001} that the satisfiability of $\FP(\Luk\Pi)$ is in $\pspace$, which would make abductive reasoning in it highly intractable.

Several questions, however, remain open. First, we did not consider determining the relevance and necessity of hypotheses. That is, given an abduction problem~$\Pmbb$ and a~hypothesis~$h$, determine if $h$~occurs in some (respectively, all) solutions of~$\Pmbb$. We conjecture that the complexity of relevance and necessity of hypotheses to solutions of $\FP$ APs would coincide with that of $\Luk$-abduction problems. It also makes sense to consider relevance and necessity of PILs w.r.t.\ \emph{full solutions}. Additionally, we plan to establish the exact complexity of recognising preferred solutions to PrAPs.

Another important task would be to provide a~translation of $\FP$-abduction into $\Luk$-abduction. Flaminio~\shortcite{Flaminio2020RSL} constructed an embedding of \emph{deductive} reasoning in fuzzy probabilistic into deductive reasoning in fuzzy propositional logic. Thus, it makes sense to extend their results and obtain a~similar correspondence for abduction.

A~further connection to investigate would be the one between abduction in~$\FP$ and alternative notions of explanations. In particular, \emph{contrastive explanations}~\cite{IgnatievNarodytskaSherMarques-Silva2021} and \emph{extended abduction} studied by Inoue and Sakama~\shortcite{InoueSakama1995,InoueSakama1999,Inoue2000,SakamaInoue2003}. Intuitively, a~contrastive explanation of~$\phi$ is telling \emph{why $\neg\phi$ did not happen}. In extended abduction, an explanation entails the observation but not some undesired consequence. In a~similar fashion, in $\FP$-abduction, an explanation why $\phi$ has probability~$20\%$ also tells why $\phi$ \emph{does not have} probability~$25\%$.

Another question is that of \emph{prioritisation} of full solutions. We considered preference based on the \emph{principle of maximal entropy}, but there are other alternatives discussed by Dubois et al.~\shortcite{DuboisGilioKern-Isberner2008}. One of them is the principle of \emph{maximal likelihood} studied by Inoue et al.~\shortcite{InoueSatoIshihataKameyaNabeshima2009} and Ishihata et al.~\shortcite{IshihataKameyaSato2012}. It would be instructive to consider the prioritisation of full solutions according to the maximal likelihood principle.

Finally, we plan to consider other notions of uncertainty. First, more expressive logics that can reason about independent events and conditional probabilities~\cite{HajekTulipani2001} or belief functions~\cite{GodoHajekEesteva2003}. Second, fuzzy logics for \emph{qualitative} probabilities~\cite{BilkovaFrittellaKozhemiachenkoMajer2023}. For the latter, it will make sense to establish a~correspondence between qualitative probabilistic abduction and preference abduction~\cite{InoueSakama1999IJCAI,KonczakVogel2005} where the explanation is a~\emph{comparative} statement about likelihoods of events.
\vfill\eject
\section*{Acknowledgements}
Tommaso Flaminio was supported by the Spanish project SHORE (PID2022-141529NB-C22) funded by the MCIN/AEI/10.13039/501100011033 and the H2020-MSCA-RISE-2020 project MOSAIC (Grant Agreement number 101007627). Katsumi Inoue was supported by JSPS KAKENHI Grant Number JP25K03190 and JST CREST Grant Number JPMJCR22D3, Japan. We would also like to thank the reviewers for their constructive suggestions that improved the quality of the paper.
\section*{AI Declaration}
The authors have not employed any Generative AI tools.
\bibliographystyle{kr}
\bibliography{kr-sample}

@InProceedings{HajekGodoEsteva1995,
author = {H\'{a}jek, Petr and Godo, Llu\'{\i}s and Esteva, Francesc},
  booktitle = {UAI'95: Proceedings of the Eleventh Conference on Uncertainty in Artificial Intelligence},
  title     = {Fuzzy logic and probability},
  year      = {1995},
  address   = {San Francisco},
  editor    = {Besnard, P. and Hanks, S.},
  pages     = {237--244},
  publisher = {Morgan Kaufmann},
doi={10.5555/2074158.2074185}
}

@INCOLLECTION{IgnatievNarodytskaSherMarques-Silva2021,
  title     = "From contrastive to abductive explanations and back again",
  booktitle = "{AIxIA 2020 -- Advances in Artificial Intelligence}",
  author    = "Ignatiev, Alexey and Narodytska, Nina and Asher, Nicholas and Marques-Silva, Joao",
  publisher = "Springer International Publishing",
  pages     = "335--355",
  series    = "Lecture Notes in Artificial Intelligence",
volume={12414},
doi={10.1007/978-3-030-77091-4_21},
editor={Matteo Baldoni and Stefania Baldini},
  year      =  2021,
  address   = "Cham"
}

@inproceedings{InoueSakama1995,
author = {Inoue, Katsumi and Sakama, Chiaki},
title = {Abductive framework for nonmonotonic theory change},
year = {1995},
isbn = {1558603638},
publisher = {Morgan Kaufmann Publishers Inc.},
address = {San Francisco, CA, USA},
booktitle = {Proceedings of the 14th International Joint Conference on Artificial Intelligence --- Volume 1},
pages = {204--210},
location = {Montreal, Quebec, Canada},
series = {IJCAI'95}
}

@INCOLLECTION{KonczakVogel2005,
  title     = "Abduction and Preferences in Linguistics",
  booktitle = "Logic Programming and Nonmonotonic Reasoning",
  author    = "Konczak, Kathrin and Vogel, Ralf",
  publisher = "Springer Berlin Heidelberg",
  pages     = "384--388",
  series    = "Lecture Notes in Computer Science",
  year      =  2005,
  address   = "Berlin, Heidelberg",
doi={10.1007/11546207_32},
editor={Chitta Baral and Gianluigi Greco and Nicola Leone and Giorgio Terracina}
}

@inproceedings{InoueSakama1999IJCAI,
  author       = {Katsumi Inoue and Chiaki Sakama},
  editor       = {Thomas Dean},
  title        = {Abducing Priorities to Derive Intended Conclusions},
  booktitle    = {Proceedings of the Sixteenth International Joint Conference on Artificial Intelligence},
  pages        = {44--49},
  publisher    = {Morgan Kaufmann},
  year         = {1999},
}

@ARTICLE{SakamaInoue2003,
  title     = "An abductive framework for computing knowledge base updates",
  author    = "Sakama, Chiaki and Inoue, Katsumi",
  journal   = "Theory and Practice of Logic Programming",
  publisher = "Cambridge University Press (CUP)",
  volume    =  3,
  number    =  06,
  pages     = "671--715",
doi={10.1017/S1471068403001716},
year      =  2003
}

@article{GodoHajekEesteva2003,
  author       = {Llu{\'{\i}}s Godo and
                  Petr H{\'{a}}jek and
                  Francesc Esteva},
  title        = {{A Fuzzy Modal Logic for Belief Functions}},
  journal      = {Fundamenta Informaticae},
  volume       = {57},
  number       = {2-4},
  pages        = {127--146},
  year         = {2003},
}

@INCOLLECTION{Inoue2000,
  title     = "A simple characterization of extended abduction",
  booktitle = "{Computational Logic --- CL 2000}",
  author    = "Inoue, Katsumi",
  publisher = "Springer Berlin Heidelberg",
  pages     = "718--732",
  series    = "Lecture Notes in Artificial Intelligence",
volume={1861},
doi={10.1007/3-540-44957-4_48},
editor={John Lloyd and Veronica Dahl and Ulrich Furbach and Manfred Kerber and Kung-Kiu Lau and Catuscia Palamidessi and Lu\'{\i}s {Moniz Pereira} and Yehoshua Sagiv and Peter J. Stuckey},
  year      =  2000,
  address   = "Berlin, Heidelberg"
}

@ARTICLE{InoueSakama1999,
  title     = "Computing extended abduction through transaction programs",
  author    = "Inoue, Katsumi and Sakama, Chiaki",
  journal   = "Ann. Math. Artif. Intell.",
  publisher = "Springer Science and Business Media LLC",
  volume    =  25,
  number    = "3--4",
  pages     = "339--367",
  year      =  1999,
 doi={10.1023/A:1018926021566}
}

@ARTICLE{BilkovaFrittellaKozhemiachenkoMajer2023,
  title     = "Qualitative reasoning in a two-layered framework",
  author    = "B{\'\i}lkov{\'a}, Marta and Frittella, Sabine and Kozhemiachenko, Daniil and Majer, Ondrej",
  journal   = "International Journal of Approximate Reasoning",
  publisher = "Elsevier BV",
  volume    =  154,
  pages     = "84--108",
  month     =  mar,
  year      =  2023,
doi={10.1016/j.ijar.2022.12.011}
}

@inproceedings{Poole1985,
  title={{On the Comparison of Theories: Preferring the Most Specific Explanation}},
  author={Poole, David},
  booktitle={Proceedings of the Ninth International Joint Conference on Artificial Intelligence},
editor       = {Aravind K. Joshi},
  pages={144--147},
publisher    = {Morgan Kaufmann},
  year={1985}
}

@book{Williamson2010,
  title={{In defence of objective Bayesianism}},
  author={Williamson, Jon},
  year={2010},
  publisher={OUP},
pagenumber={185+vi.}
}

@InProceedings{CorsiFlaminioGodoHosni2023,
  title = {A~modal logic for uncertainty: a completeness theorem},
  author = {Corsi, E.A. and Flaminio, T. and Godo, L. and Hosni, H.},
  booktitle = {Proceedings of the Thirteenth International Symposium on Imprecise Probability: Theories and Applications},
  pages = 	 {119--129},
  year = 	 {2023},
  editor = 	 {Miranda, E. and Montes, I. and Quaeghebeur, E. and Vantaggi, B.},
  volume = 	 {215},
  series = 	 {Proceedings of Machine Learning Research},
  publisher =    {PMLR}
}

@inproceedings{InoueSatoIshihataKameyaNabeshima2009,
  author       = {Inoue, Katsumi and Sato, Taisuke and Ishihata, Masakazu and Kameya, Yoshitaka and Nabeshima, Hidetomo},
  editor       = {Boutilier, Craig},
  title        = {{Evaluating Abductive Hypotheses using an EM Algorithm on BDDs}},
  booktitle    = {Proceedings of the 21st International Joint Conference on Artificial Intelligence},
  pages        = {810--815},
  year         = {2009},
}

@inproceedings{IgnatievNarodytskaMarques-Silva2019,
  author       = {Ignatiev, Alexey and Narodytska, Nina and Marques{-}Silva, Jo{\~{a}}o},
  title        = {Abduction-Based Explanations for Machine Learning Models},
  booktitle    = {The Thirty-Third {AAAI} Conference on Artificial Intelligence},
  pages        = {1511--1519},
  publisher    = {{AAAI} Press},
  year         = {2019},
  doi          = {10.1609/AAAI.V33I01.33011511},
}

@inproceedings{ConsoleTorasso1990,
  author       = {Console, Luca and Torasso, Pietro},
  title        = {Integrating Models of the Correct Behavior into Abductive Diagnosis},
  booktitle    = {9th European Conference on Artificial Intelligence},
  pages        = {160--166},
  year         = {1990},
}

@article{Nilsson1986,
  author       = {Nilsson, Nils J.},
  title        = {Probabilistic Logic},
  journal      = {Artificial Intelligence},
  volume       = {28},
  number       = {1},
  pages        = {71--87},
  year         = {1986},
  doi          = {10.1016/0004-3702(86)90031-7},
}

@inproceedings{Pople1973,
  author       = {Pople, Harry E.},
  editor       = {Nilsson, Nils J.},
  title        = {On the Mechanization of Abductive Logic},
  booktitle    = {Proceedings of the 3rd International Joint Conference on Artificial Intelligence},
  pages        = {147--152},
  publisher    = {William Kaufmann},
  year         = {1973},
}

@inproceedings{IshihataKameyaSato2012,
  author       = {Ishihata, Masakazu and Kameya, Yoshitaka and Sato, Taisuke},
  editor       = {Muggleton, Stephen H. and Tamaddoni-Nezhad, Alireza and Lisi, Francesca A.},
  title        = {{Variational Bayes Inference for Logic-Based Probabilistic Models on BDDs}},
  booktitle    = {Inductive Logic Programming~--- 21st International Conference},
  series       = {Lecture Notes in Artificial Intelligence},
  volume       = {7207},
  pages        = {189--203},
  publisher    = {Springer},
  year         = {2011},
  doi          = {10.1007/978-3-642-31951-8\_19},
}

@BOOK{PengReggia1990,
  title     = "Abductive inference models for diagnostic problem-solving",
  author    = "Peng, Yun and Reggia, James A.",
  publisher = "Springer",
  series    = "Artificial Intelligence",
  edition   =  {1st (reprint)},
  month     =  dec,
  year      =  2012,
  address   = "New York, NY",
  language  = "en"
}

@INCOLLECTION{ValkovskySavvinGerasimov1999,
  title     = "Abduction problem in probabilistic constraint logic programming",
  booktitle = "Advances in Soft Computing",
  author    = "Valkovsky, Vladislav B. and Savvin, Konstantin O. and Gerasimov, Michael B.",
  publisher = "Springer London",
  pages     = "85--98",
  year      =  1999,
  address   = "London"
}

@ARTICLE{Jaynes1957-1,
  title     = "Information theory and statistical mechanics",
  author    = "Jaynes, Edwin Thompson",
  journal   = "Physical Review",
  publisher = "American Physical Society (APS)",
  volume    =  106,
  number    =  4,
  pages     = "620--630",
  month     =  may,
  year      =  1957,
}

@ARTICLE{Jaynes1957-2,
  title     = "Information theory and statistical mechanics. {II}",
  author    = "Jaynes, Edwin Thompson",
  journal   = "Physical Review",
  publisher = "American Physical Society (APS)",
  volume    =  108,
  number    =  2,
  pages     = "171--190",
  month     =  oct,
  year      =  1957,
}

@inproceedings{KozhemiachenkoSedlar2025,
    author = {Kozhemiachenko, Daniil and Sedlár, Igor},
    title = {Complexity of {{\L}ukasiewicz} Modal Probabilistic Logics},
    booktitle = {Proceedings Twentieth Conference on Theoretical Aspects of Rationality and Knowledge},
editor={Bjorndahl, Adam},
series={Electronic Proceedings in Theoretical Computer Science},
volume={437},
pages={350--364},
    year = {2025}
}

@Article{HajekGodoEsteva2000,
  author  = {H{\'a}jek, P. and Godo, L. and Esteva, F.},
  journal = {Neural Network World},
  title   = {Reasoning about probability using fuzzy logic},
  year    = {2000},
  number  = {5},
  pages   = {811--824},
  volume  = {10},
}

@INPROCEEDINGS{BienvenuInoueKozhemiachenko2024KR,
  title           = "{Abductive Reasoning in a Paraconsistent Framework}",
  booktitle       = "Proceedings of the Twenty-First International Conference on Principles of Knowledge Representation and Reasoning",
  author          = "Bienvenu, Meghyn and Inoue, Katsumi and Kozhemiachenko, Daniil",
  publisher       = "International Joint Conferences on Artificial Intelligence Organization",
  pages           = "134--144",
  month           =  nov,
  year            =  2024,
  address         = "California",
  conference      = "21st International Conference on Principles of Knowledge Representation and Reasoning \{KR-2024\}",
  location        = "Hanoi, Vietnam",
doi={10.24963/kr.2024/13}
}

@book{deFinetti2017,
  title={Theory of Probability},
volume={1},
  author={de Finetti, Bruno},
  publisher={Wiley},
  year={1974},
city={New York}
}

@ARTICLE{Flaminio2020RSL,
  title     = "Three characterizations of strict coherence on infinite-valued events",
  author    = "Flaminio, Tommaso",
  journal   = "Review of Symbolic Logic",
  publisher = "Cambridge University Press (CUP)",
  volume    =  13,
  number    =  3,
  pages     = "593--610",
  month     =  sep,
  year      =  2020,
  language  = "en",
doi={10.1017/S1755020319000546}
}

@Article{Flaminio2007,
  author    = {Flaminio, T.},
  journal   = {Archive for Mathematical Logic},
  title     = {{NP-containment for the coherence test of assessments of conditional probability: a fuzzy logical approach}},
  year      = {2007},
  issn      = {1432-0665},
  month     = feb,
  number    = {3–4},
  pages     = {301--319},
  volume    = {46},
  doi       = {10.1007/s00153-007-0045-3},
  publisher = {Springer Science and Business Media LLC},
}

@Article{FlaminioUgolini2024,
  author       = {Flaminio, T. and Ugolini, S.},
  title        = {Encoding de {F}inetti's coherence within {{\L}}ukasiewicz logic and {MV}-algebras},
  journal      = {Annals of Pure and Applied Log.},
  volume       = {175},
  number       = {9},
  pages        = {103337},
  year         = {2024},
  doi          = {10.1016/J.APAL.2023.103337},
}

@misc{FlaminioInoueKozhemiachenko2026arxiv,
    author = {Flaminio, Tommaso and Inoue, Katsumi and Kozhemiachenko, Daniil},
    year = {2026},
    title = {{Probabilistic Abduction in a~Fuzzy Logic Framework}},
    note = {Available at \href{}{arXiv: XXX}},
}

@INCOLLECTION{Fenstad1967,
  title     = "Representations of probabilities defined on first order languages",
  booktitle = "Studies in Logic and the Foundations of Mathematics",
  author    = "Fenstad, J.E.",
  publisher = "Elsevier",
  pages     = "156--172",
  series    = "Studies in logic and the foundations of mathematics",
  year      =  1967,
  language  = "en"
}

@inproceedings{DeRaedtKimmigToivonen2007,
  title={ProbLog: A probabilistic Prolog and its application in link discovery},
  author={Raedt, Luc de and Kimmig, Angelika and Toivonen, Hannu},
  booktitle={Proceedings of the 20th international joint conference on artificial intelligence},
  pages={2462--2467},
  year={2007},
}

@inproceedings{InoueKozhemiachenko2025,
    title     = {{Complexity of Abduction in Łukasiewicz Logic}},
    author    = {Inoue, Katsumi and Kozhemiachenko, Daniil},
    booktitle = {{Proceedings of the 22nd International Conference on Principles of Knowledge Representation and Reasoning}},
    pages     = {417--428},
    year      = {2025},
    month     = {10},
    doi       = {10.24963/kr.2025/41},
  }

@inproceedings{FlaminioPretoUgolini2023,
  author       = {Flaminio, T. and Preto, S. and Ugolini, S.},
  editor       = {Marquis, P. and Son, T.C. and Kern{-}Isberner, G.},
  title        = {Reasoning about Probability via Continuous Functions},
  booktitle    = {Proceedings of the 20th International Conference on Principles of Knowledge Representation and Reasoning},
  pages        = {282--290},
  year         = {2023},
  doi          = {10.24963/KR.2023/28},
}

@ARTICLE{SatoIshihataInoue2011,
  title     = "Constraint-based probabilistic modeling for statistical abduction",
author    = "Sato, Taisuke and Ishihata, Masakazu and Inoue, Katsumi",
journal   = "Machine Learning",
  publisher = "Springer Science and Business Media LLC",
  volume    =  83,
  number    =  2,
  pages     = "241--264",
  month     =  may,
  year      =  2011,
  language  = "en"
}

@Article{EstevaGodoMontagna2001,
  author    = {Esteva, Francesc and Godo, Lluís and Montagna, Franco},
  journal   = {Archive for Mathematical Logic},
  title     = {The {$L\Pi$} and {$L\Pi\frac{1}{2}$} logics: two complete fuzzy systems joining Łukasiewicz and Product Logics},
  year      = {2001},
  number    = {1},
  pages     = {39--67},
  volume    = {40},
  doi       = {10.1007/s001530050173},
  publisher = {Springer Science and Business Media LLC},
}

@inproceedings{KateMooney2009,
  title={{Probabilistic abduction using Markov logic networks}},
  author={Kate, Rohit and Mooney, Raymond J.},
  booktitle={IJCAI-09 Workshop on Plan, Activity, and Intent Recognition},
  pages={22--28},
  year={2009}
}

@ARTICLE{Poole1993,
  title     = "{Probabilistic Horn abduction and Bayesian networks}",
  author    = "Poole, D.",
  journal   = "Artificial Intelligence",
  publisher = "Elsevier BV",
  volume    =  64,
  number    =  1,
  pages     = "81--129",
  month     =  nov,
  year      =  1993,
  language  = "en"
}

@ARTICLE{DuboisGilioKern-Isberner2008,
  title     = "{Probabilistic Abduction Without Priors}",
  author    = "Dubois, D. and Gilio, A. and Kern-Isberner, G.",
  journal   = "International Journal of Approximate Reasoning",
  publisher = "Elsevier BV",
  volume    =  47,
  number    =  3,
  pages     = "333--351",
  month     =  mar,
  year      =  2008,
  language  = "en"
}

@article{BaldiCintulaNoguera2020,
  title={{Classical and Fuzzy Two-Layered Modal Logics for Uncertainty: Translations and Proof-Theory}},
  author={Baldi, P. and Cintula, P. and Noguera, C.},
  journal={International Journal of Computational Intelligence Systems},
  volume={13},
  number={1},
  pages={988--1001},
  year={2020},
  publisher={Atlantis Press}
}

@inproceedings{BaldiCintulaNoguera2019,
  author       = {Baldi, Paolo and Cintula, Petr and Noguera, Carles},
  editor       = {Nov{\'{a}}k, Vil{\'{e}}m and Mar{\'{\i}}k, Vladim{\'{\i}}r and Stepnicka, Martin and Navara, Mirko and Hurt{\'{\i}}k, Petr},
  title        = {Translating Classical Probability Logics into Modal Fuzzy Logics},
  booktitle    = {Proceedings of the 11th Conference of the European Society for Fuzzy Logic and Technology, {EUSFLAT} 2019},
  series       = {Atlantis Studies in Uncertainty Modelling},
  volume       = {1},
  publisher    = {Atlantis Press},
  year         = {2019},
  doi          = {10.2991/EUSFLAT-19.2019.49},
}

@article{HajekTulipani2001,
  title={Complexity of fuzzy probability logics},
  author={H{\'a}jek, P. and Tulipani, S.},
  journal={Fundamenta Informaticae},
  volume={45},
  number={3},
  pages={207--213},
  year={2001},
  publisher={IOS Press}
}

@inproceedings{Vojtas1999,
  author       = {Vojt{\'{a}}\v{s}, P.},
  editor       = {Mayor, G. and Su{\~{n}}er, J.},
  title        = {{Fuzzy Logic Abduction}},
  booktitle    = {Proceedings of the {EUSFLAT-ESTYLF} Joint Conference},
  pages        = {319--322},
  publisher    = {Universitat de les Illes Balears, Palma de Mallorca, Spain},
  year         = {1999},
}

@inproceedings{YamadaMukaidono1995,
  title={{Fuzzy Abduction Based on \L{}ukasiewicz Infinite-Valued Logic and Its Approximate Solutions}},
  author={Yamada, K. and Mukaidono, M.},
  booktitle={Proceedings of 1995 IEEE International Conference on Fuzzy Systems.},
  volume={1},
  pages={343--350},
  year={1995},
  organization={IEEE}
}

@ARTICLE{ChakrabortyKonarPalJain2013,
  title     = "Extending the Contraposition Property of Propositional Logic for
               Fuzzy Abduction",
  author    = "Chakraborty, A. and Konar, A. and Pal, N.R. and Jain, L.C.",
  journal   = "IEEE Transactions in Fuzzy Systems",
  publisher = "Institute of Electrical and Electronics Engineers (IEEE)",
  volume    =  21,
  number    =  4,
  pages     = "719--734",
  month     =  aug,
  year      =  2013
}

@article{FaginHalpernMegiddo1990,
  title={A logic for reasoning about probabilities},
  author={Fagin, R. and Halpern, J.Y. and Megiddo, N.},
  journal={Information and computation},
  volume={87},
  number={1--2},
  pages={78--128},
  year={1990},
  publisher={Academic Press},
doi={10.1016/0890-5401(90)90060-U}
}

@incollection{SakamaInoue1995,
    author = {Sakama, C. and Inoue, K.},
    isbn = {9780262291439},
    title = "{The Effect of Partial Deduction in Abductive Reasoning}",
    booktitle = "{Logic Programming: The 12th International Conference}",
    publisher = {The MIT Press},
    year = {1995},
    month = {06},
    doi = {10.7551/mitpress/4298.003.0042},
}

@inproceedings{Stickel1990,
  title={{Rationale and Methods for Abductive Reasoning in Natural-Language Interpretation}},
  author={Stickel, M.E.},
  booktitle={Natural Language and Logic},
series={Lecture Notes in Artificial Intelligence},
volume={459},
  pages={233--252},
  year={1990},
  organization={Springer},
doi={10.1007/3-540-53082-7_26},
}

@ARTICLE{CreignouZanuttini2006,
  title     = "{A Complete Classification of the Complexity of Propositional Abduction}",
  author    = "Creignou, N. and Zanuttini, B.",
  journal   = "SIAM Journal of Computing",
  publisher = "Society for Industrial \& Applied Mathematics (SIAM)",
  volume    =  36,
  number    =  1,
  pages     = "207--229",
  month     =  jan,
  year      =  2006
}

@ARTICLE{BofillManyaVidalVillaret2019,
  title     = "{New complexity results for {\L}ukasiewicz logic}",
  author    = "Bofill, M. and Many{\`a}, F. and Vidal, A. and Villaret, M.",
  journal   = "Soft Comput.",
  publisher = "Springer Science and Business Media LLC",
  volume    =  23,
  number    =  7,
  pages     = "2187--2197",
  month     =  apr,
  year      =  2019,
  language  = "en"
}

@ARTICLE{Vojtas2001,
  title     = "Fuzzy logic programming",
  author    = "Vojt{\'a}{\v s}, P.",
  journal   = "Fuzzy Sets And Systems",
  publisher = "Elsevier BV",
  volume    =  124,
  number    =  3,
  pages     = "361--370",
  month     =  dec,
  year      =  2001,
  language  = "en"
}

@ARTICLE{Ebrahim2001,
  title     = "Fuzzy logic programming",
  author    = "Ebrahim, R.",
  journal   = "Fuzzy Sets And Systems",
  publisher = "Elsevier BV",
  volume    =  117,
  number    =  2,
  pages     = "215--230",
  month     =  jan,
  year      =  2001,
  language  = "en"
}

@incollection{Baaz1996,
  title={{Infinite-valued G{\"o}del logics with $0 $-$1 $-projections and relativizations}},
  author={Baaz, M.},
  booktitle={G{\"o}del'96: Logical foundations of mathematics, computer science and physics---Kurt G{\"o}del's legacy, Brno, Czech Republic, August 1996, proceedings},
  pages={23--33},
  year={1996},
  publisher={Association for Symbolic Logic}
}

@ARTICLE{ElAyebMarquisRusinowitch1993,
  title     = "Preferring diagnoses by abduction",
  author    = "El Ayeb, B. and Marquis, P. and Rusinowitch, M.",
  journal   = "IEEE Transactions on Systems, Man, and Cybernetics",
  publisher = "Institute of Electrical and Electronics Engineers (IEEE)",
  volume    =  23,
  number    =  3,
  pages     = "792--808",
  year      =  1993,
}

@inproceedings{Sato1995,
  author       = {Sato, Taisuke},
  editor       = {Sterling, Leon},
  title        = {A Statistical Learning Method for Logic Programs with Distribution Semantics},
  booktitle    = {Proceedings of the Twelfth International Conference on Logic Programming},
  pages        = {715--729},
  publisher    = {{MIT} Press},
  year         = {1995},
}

@ARTICLE{Paul1993,
  title     = "Approaches to abductive reasoning: an overview",
  author    = "Paul, G.",
  journal   = "Artificial Intelligence Review",
  publisher = "Springer Science and Business Media LLC",
  volume    =  7,
  number    =  2,
  pages     = "109--152",
  month     =  apr,
  year      =  1993,
  language  = "en"
}

@BOOK{FlachKakas2000,
  title     = "Abduction and Induction",
  editor    = "Flach, P.A. and Kakas, A.C.",
  publisher = "Springer",
  series    = "Applied Logic Series",
volume={18},
  month     =  dec,
  year      =  2000,
  address   = "Dordrecht, Netherlands",
  language  = "en"
}

@BOOK{Magnani2001,
  title     = "Abduction, reason and science",
  author    = "Magnani, L.",
  publisher = "Springer",
  edition   =  2001,
  month     =  jun,
  year      =  2011,
  address   = "New York, NY",
  language  = "en"
}

@INPROCEEDINGS{KhannaSudanTrevisan2002,
  title      = "Constraint satisfaction: the approximability of minimization problems",
  booktitle  = "Proceedings of Computational Complexity. Twelfth Annual {IEEE} Conference",
  author     = "Khanna, Sanjeev and Sudan, Madhu and Trevisan, Luca",
  publisher  = "IEEE Computing Society",
  year       =  2002,
  conference = "Computational Complexity. Twelfth Annual IEEE Conference",
  location   = "Ulm, Germany",
doi={10.1109/CCC.1997.612323}
}

@BOOK{JosephsonJosephson1994,
  title     = "Abductive Inference: Computation, Philosophy, Technology",
  editor    = "Josephson, J.R. and Josephson, S.G.",
  publisher = "Cambridge University Press",
  month     =  oct,
  year      =  1994,
  address   = "Cambridge, England"
}

@ARTICLE{Koitz-HristovWotawa2018,
  title     = "Applying algorithm selection to abductive diagnostic reasoning",
  author    = "Koitz-Hristov, R. and Wotawa, F.",
  journal   = "Applied Intelligence",
  publisher = "Springer Science and Business Media LLC",
  volume    =  48,
  number    =  11,
  pages     = "3976--3994",
  month     =  nov,
  year      =  2018,
  language  = "en"
}

@article{DaiXuYuZhou2019,
  title={Bridging machine learning and logical reasoning by abductive learning},
  author={Dai, W.-Z. and Xu, Q. and Yu, Y. and Zhou, Z.-H.},
  journal={Advances in Neural Information Processing Systems},
  volume={32},
  year={2019}
}

@inproceedings{BhagavatulaLeBrasMalaviyaSakaguchiHoltzmanRashkinDowneyYihChoi2020,
  author       = {Bhagavatula, C. and Le Bras, R. and Malaviya, C. and Sakaguchi, K. and Holtzman, A. and Rashkin, H. and Downey, D. and Yih, W.-t. and Choi, Y.},
  title        = {{Abductive Commonsense Reasoning}},
  booktitle    = {8th International Conference on Learning Representations, {ICLR} 2020},
  publisher    = {OpenReview.net},
  year         = {2020},
}

@ARTICLE{EiterGottlob1995,
  title     = "{The Complexity of Logic-Based Abduction}",
  author    = "Eiter, T. and Gottlob, G.",
  journal   = "Journal of the Association for Computing Machinery",
  publisher = "Association for Computing Machinery (ACM)",
  volume    =  42,
  number    =  1,
  pages     = "3--42",
  month     =  jan,
  year      =  1995,
  language  = "en"
}

@ARTICLE{MellouliBouchon-Meunier2003,
  title     = "Abductive reasoning and measures of similitude in the presence of fuzzy rules",
  author    = "Mellouli, N. and Bouchon-Meunier, B.",
  journal   = "Fuzzy Sets And Systems",
  publisher = "Elsevier BV",
  volume    =  137,
  number    =  1,
  pages     = "177--188",
  month     =  jul,
  year      =  2003,
  language  = "en"
}

@ARTICLE{MiyataFuruhashiUchikawa1998,
  title     = "A proposal of abductive inference with degrees of manifestations",
  author    = "Miyata, Y. and Furuhashi, T. and Uchikawa, Y.",
  journal   = "International Journal of Intelligent Systems",
  publisher = "Hindawi Limited",
  volume    =  13,
  number    =  6,
  pages     = "467--481",
  month     =  jun,
  year      =  1998,
  language  = "en"
}

@ARTICLE{AzzoliniBellodiferilliRiguzziZese2022,
  title     = "Abduction with probabilistic logic programming under the distribution semantics",
  author    = "Azzolini, Damiano and Bellodi, Elena and Ferilli, Stefano and Riguzzi, Fabrizio and Zese, Riccardo",
  journal   = "International Journal of Approximate Reasoning",
  publisher = "Elsevier BV",
  volume    =  142,
  pages     = "41--63",
  month     =  mar,
  year      =  2022,
  language  = "en",
doi={10.1016/j.ijar.2021.11.003}
}

@article{BienvenuInoueKozhemiachenko2026,
    author = {Meghyn Bienvenu and Katsumi Inoue and Daniil Kozhemiachenko},
    title = {{Abductive Reasoning in Expansions of Belnap--Dunn Logic}},
    journal = {Journal of Artificial Intelligence Research},
volume={85},
doi={10.1613/jair.1.18874},
    year = {2026}
}

@ARTICLE{Shanahan2005,
  title     = "Perception as abduction: turning sensor data into meaningful representation",
  author    = "Shanahan, M.",
  journal   = "Cognitive Science",
  publisher = "Wiley",
  volume    =  29,
  number    =  1,
  pages     = "103--134",
  month     =  jan,
  year      =  2005,
}

@INCOLLECTION{Tsypyschev2017,
  title     = "Application of risk theory approach to fuzzy abduction",
  series = "Advances in Intelligent Systems and Computing",
booktitle={Cybernetics and Mathematics Applications in Intelligent Systems},
volume={574},
  author    = "Tsypyschev, V.N.",
  publisher = "Springer International Publishing",
  pages     = "13--19",
  year      =  2017,
  address   = "Cham"
}

@INPROCEEDINGS{dAllonnesAkdagBouchon-Meunier2007,
  title           = "Selecting implications in fuzzy abductive problems",
  booktitle       = "2007 {IEEE} Symposium on Foundations of Computational Intelligence",
  author          = "d'Allonnes, A.R. and Akdag, H. and Bouchon-Meunier, B.",
  publisher       = "IEEE",
  month           =  apr,
  year            =  2007,
  location        = "Honolulu, HI, USA"
}

@INCOLLECTION{BergadanoCutelloGunetti2000,
  title     = "Abduction in Machine Learning",
  booktitle = "Abductive Reasoning and Learning",
  author    = "Bergadano, F. and Cutello, V. and Gunetti, D.",
  publisher = "Springer Netherlands",
volume={4},
series={Handbook of Defeasible Reasoning and Uncertainty Management Systems},
editor={Gabbay, D.M. and Smets, P.},
  pages     = "197--229",
  year      =  2000,
  address   = "Dordrecht"
}
\clearpage
\onecolumn
\appendix
\section{Proofs of Section~\ref{sec:complexity}}
\FPLuktoLukembedding*
\begin{proof}
It is clear from Definitions~\ref{def:Luksemantics} and~\ref{def:FPLuksemantics} that if $\phi$ is $\Luk$-valid, then it is $\FP$-valid. For the converse direction, let $v(\phi)<1$ for some $\Luk$-valuation~$v$. As all variables in~$\phi$ are pairwise independent, it is clear that there is a~measure $\mu$ on $2^{2^{\Var(\phi)}}$ s.t.\ $\mu(\|p\|)=v(p)$ for every $p\in\Var(\phi)$. From here, it follows that $\Imc_\Mfrak(\phi^\Prob)=v(\phi)<1$ with $\Mfrak=\langle{2^{\Var(\phi)}},\mu\rangle$.
\end{proof}
\subsection{Proofs of Section~\ref{ssec:solutionrecognition}}
We begin with a~technical definition.
\begin{definition}\label{def:outercounterpart}
Let $\alpha\!\in\!\LProbint$, $\sharp{\in}\{\neg,\triangle\}$, and $\circledast{\in}\{\odot,\oplus,\rightarrow\}$. For every $\phi\in\Emc(\alpha)$, we let $p_\phi$ to be a~fresh propositional variable and define the \emph{outer counterpart} $\alpha^\uparrow$ of~$\alpha$ as follows:
\begin{align*}
(\Prob(\phi))^\uparrow&=p_\phi&(\Prob(\phi)\diamond\cvalue)^\uparrow&=p_\phi\diamond\cvalue\\
(\sharp\alpha)^\uparrow&=\sharp\alpha^\uparrow&(\alpha\circledast\beta)^\uparrow&=\alpha^\uparrow\circledast\beta^\uparrow
\end{align*}
For $\Gamma\subseteq\LProbint$, we set $\Gamma^\uparrow=\{\alpha^\uparrow\mid\alpha\in\Gamma\}$.
\end{definition}
Let us briefly consider how $(\cdot)^\uparrow$ works. E.g., if $\alpha=\Prob(q\wedge r)\rightarrow\neg\Prob(s)$, we need two fresh variables $p_{q\wedge r}$ and $p_s$. This gives us $\alpha^\uparrow=p_{q\wedge r}\rightarrow\neg p_s$. Note that the lengths of~$\alpha^\uparrow$ and~$\alpha$ differ by a~constant factor.
\sufficientsolutionrecognitionDP*
\begin{proof}
$\DP$-membersip follows immediately from Proposition~\ref{prop:FPLukNPcoNP} because to recognise a~sufficient solution, we need to check that $\Gamma,\eta\not\models_{\FP}\bot$ (an $\np$ task) and $\Gamma,\eta\models_{\FP}\delta$ (a~$\conp$ task) and these two checks can be run independently of one another.

For $\DP$-hardness, we construct a~polynomial reduction from the solution recognition in Łukasiewicz logic (recall Definition~\ref{def:Lukabduction} for the notions of $\Luk$-abduction problems and their solutions), which is shown to be $\DP$-complete by Inoue and Kozhemiachenko~\shortcite{InoueKozhemiachenko2025}. In particular, solution recognition for $\Luk$-abduction problems $\Pmbb=\langle\Phi,\chi,\Hmsf\rangle$ s.t.\ $\{p\vee\neg p\mid p\in\Var[\Pmbb]\}\subseteq\Phi$\footnote{Here, $p\vee\neg p$ is a~shorthand for $(p\rightarrow\neg p)\rightarrow\neg p$.} and $\Hmsf$ contains only interval literals $p{\geq}\one$ and $q{\leq}\zero$ is shown to be $\DP$-hard (cf.~\cite[Theorem~1]{InoueKozhemiachenko2025}).

Given such a~problem, we define an $\FP$ abduction problem $\Pmbb^\Prob=\langle\Phi^\Prob,\chi^\Prob,\Hmsf',\Vmc'\rangle$ with $\Hmsf'=\Var[\Hmsf]$ and $\Vmc'=\{0,1\}$ (recall Definition~\ref{def:FPLuktoLukembedding} for $\Phi^\Prob$ and~$\chi^\Prob$). Now let $\tau$ be an interval term. We show that $\tau$ is an $\Luk$-solution to~$\Pmbb$ iff~$\tau^\Prob$ is a~sufficient solution to~$\Pmbb^\Prob$. If $\tau$ solves~$\Pmbb$, then $\Phi,\tau\models_\Luk\chi$ and $\Phi,\tau\not\models_\Luk\bot$. By Proposition~\ref{prop:LukisfragmentofFPLuk}, this is equivalent to $\Phi^\Prob,\tau^\Prob\models_{\FP}\chi$ and $\Phi,\tau^\Prob\not\models_{\FP}\bot$.

Conversely, let $\eta$ be as follows:
\begin{align*}
\eta&=\bigodot^k_{i=1}\Prob(p_i){\geq}\one\odot\bigodot^{k'}_{i=1}\Prob(p'_i){>}\zero\odot\bigodot^m_{i=1}\Prob(q_i){\leq}0\odot\bigodot^{m'}_{i=1}\Prob(q'_i){<}\one
\end{align*}
We define (recall Definition~\ref{def:outercounterpart}):
\begin{align}\label{equ:etauparrow}
\eta^\uparrow&=\bigodot_{\Prob(p){\geq}\one\in\eta}p{\geq}\one\odot\bigodot_{\Prob(p'){>}\zero\in\eta}p'{\geq}\one\odot\bigodot_{\Prob(q){\leq}\zero\in\eta}q{\leq}\zero\odot\bigodot_{\Prob(q'){<}\one\in\eta}q'{\leq}\zero
\end{align}
We show that $\eta$ is a~sufficient solution to~$\Pmbb^\Prob$ iff $\eta^\uparrow$ is a~solution to the $\Luk$-abduction problem~$\Pmbb$. Since $\Prob(r)\vee\neg\Prob(r)\in\Phi^\Prob$ for every $r\in\Var[\Pmbb]$ and since $\Imc_\Mfrak(\alpha\vee\neg\alpha)=1$ iff $\Imc_\Mfrak(\alpha)\in\{0,1\}$ for every $\alpha\in\LProbint$, it follows that
\begin{align}\label{equ:LEMequivalence}
\forall i\in\{1,\ldots,k'\}:\Phi^\Prob\models_{\FP}\Prob(p'_i){>}\zero\leftrightarrow\Prob(p'_i){\geq}\one
&&
\forall j\in\{1,\ldots,m'\}:\Phi^\Prob\models_{\FP}\Prob(q'_j){<}\one\leftrightarrow\Prob(q'_j){\leq}\zero
\end{align}
Thus, we can replace PILs $\Prob(p'_i){>}\zero$ and $\Prob(q'_j){<}\one$ in~$\eta$ with $\Prob(p'_i){\geq}\one$ and $\Prob(q'_j){\leq}\zero$:
\begin{align}\label{equ:etapitchfork}
\eta^\pitchfork&=\bigodot^k_{i=1}\Prob(p_i){\geq}\one\odot\bigodot^{k'}_{i=1}\Prob(p'_i){\geq}\one\odot\bigodot^m_{i=1}\Prob(q_i){\leq}0\odot\bigodot^{m'}_{i=1}\Prob(q'_i){\leq}\zero
\end{align}
It is clear that $\eta^\pitchfork$~is a~sufficient solution to~$\Pmbb^\Prob$ iff $\eta$ is a~sufficient solution to~$\Pmbb^\Prob$. Now, observe from Definition~\ref{def:FPLuktoLukembedding}, \eqref{equ:etauparrow}, and~\eqref{equ:etapitchfork} that $\eta^\pitchfork=(\eta^\uparrow)^\Prob$. Thus, we have
\begin{align*}
\Phi^\Prob,\eta\consvDashFPLuk\chi^\Prob&\text{ iff }\Phi^\Prob,\eta^\pitchfork\consvDashFPLuk\chi^\Prob\tag{from~\eqref{equ:LEMequivalence}}\\
&\text{ iff }\Phi,\eta^\uparrow\consvDashLuk\chi\tag{by Proposition~\ref{prop:LukisfragmentofFPLuk}}
\end{align*}
The result now follows.
\end{proof}
\fullsolutionrecognitionpolynomial*
\begin{proof}
Observe from Definition~\ref{def:probabilisticintervalliterals} that $\langle\Vmbf,\Vmc\rangle$-complete PITs correspond to a~unique probability distribution on~$2^\Vmbf$. This means that for every $\theta=\bigodot^k_{i=1}\Prob(\tau_i){\leq}\cvalue_i\odot\bigodot^k_{i=1}\Prob(\tau_i){\geq}\cvalue_i$ there is exactly one probabilistic model $\Mfrak_\theta=\langle2^\Vmbf,\mu_\theta\rangle$ s.t.\ $\Imc_{\Mfrak_\theta}(\theta)=1$. Thus, to determine whether $\Gamma,\theta\consvDashFPLuk\delta$, it suffices to check that $\Imc_{\Mfrak_\theta}(\gamma)=1$ for every $\gamma\in\Gamma$ and $\Imc_{\Mfrak_\theta}(\delta)=1$.

Now, we need to compute the values of probabilistic atoms occurring in~$\Gamma\cup\{\delta\}$. To do that, observe that since $\theta$ defines a~probability distribution, $c_i$'s correspond to the assignment of probabilities to the \emph{atoms} of $2^{2^\Vmbf}$. Thus, $\mu_\theta(\|\phi\|_{\Mfrak_\theta})=\sum_{\tau_i\models_\CPL\phi}c_i$. Furthermore, since $\tau_i$'s contain all variables in~$\Vmbf$, it takes polynomial time to check that $\tau_i\models_\CPL\phi$. Thus, it takes polynomial time (w.r.t.\ the length of~$\theta$) to determine the value $\mu_\theta(\|\phi\|_{\Mfrak_\theta})$ for every $\phi\in\Emc[\Pmbb]$. It follows that the values of $\Imc_{\Mfrak_\theta}(\Prob(\phi))$ can also be computed in polynomial time. Hence, it takes polynomial time to determine the values of $\Imc_{\Mfrak_\theta}(\gamma)$ for $\gamma\in\Gamma\cup\{\delta\}$, as required.
\end{proof}
\sufficientsolutionexistence*
\begin{proof}
For $\Sigma^\Pmsf_2$-hardness, we reuse the reduction from Theorem~\ref{theorem:sufficientsolutionrecognitionDP}. Namely, given an $\Luk$-abduction problem $\Pmbb=\langle\Phi,\chi,\Hmsf\rangle$ s.t.\ $\{p\vee\neg p\mid p\in\Var[\Pmbb]\}\subseteq\Phi$ and $\Hmsf$ contains only interval literals $p{\geq}\one$ and $q{\leq}\zero$,\footnote{Observe from~\cite[Theorem~5]{InoueKozhemiachenko2025} that solution existence for such $\Luk$-abduction problems is $\Sigma^\Pmsf_2$-hard. Moreover, observe that since $p\vee\neg p$ is a~shorthand for $(p\rightarrow\neg p)\rightarrow\neg p$, it follows from Definition~\ref{def:Luksemantics} that for every $\Luk$-valuation $v$, it holds that $v(p\vee\neg p)=1$ iff $v(p)\in\{0,1\}$. Thus, adding $p\vee\neg p$ for each $p\in\Var[\Pmbb]$ to the theory forces all propositional variables to have \emph{classical} values.} we construct $\Pmbb^\Prob=\langle\Phi^\Prob,\chi^\Prob,\Hmsf',\Vmc'\rangle$ as shown in the proof of Theorem~\ref{theorem:sufficientsolutionrecognitionDP}. Then, by Proposition~\ref{prop:LukisfragmentofFPLuk}, $\tau$ is a~solution to~$\Pmbb$, iff $\tau^\Prob$ is a~sufficient solution to~$\Pmbb^\Prob$. And conversely, $\eta$~is a~sufficient solution to~$\Pmbb^\Prob$ iff $\eta^\uparrow$ (cf.~\eqref{equ:etauparrow}) is a~solution to~$\Pmbb$.

For $\Sigma^\Pmsf_2$-membership, we observe from Theorem~\ref{theorem:sufficientsolutionrecognitionDP} that we can guess a~solution candidate~$\eta$ and then verify it in $\DP$ time. This gives us the desired upper bound. The result now follows.
\end{proof}

\fullsolutionexistenceNP*
\begin{proof}
We begin with $\np$-hardness. To show that the existence of concise full solutions is $\np$-hard, we provide a~polynomial reduction from $\CPL$-satisfiability. Namely, let $\phi\in\CPL$, $p\in\Var(\phi)$, and $\Vmbf=\Var(\phi)$. We show that $\phi\in\LCPL$ is $\CPL$-satisfiable iff $\Pmbb=\langle\varnothing,\Prob(\phi),\{p\},\{0,1\}\rangle$ has a~full solution.\footnote{Note that the theory of $\Pmbb$ is empty. Recall, furthermore, from Definition~\ref{def:solutions} that full solutions do not depend on the set of hypotheses.}

Let $\phi$ be $\CPL$-satisfiable and let $v$ be a~$\CPL$-valuation s.t.\ $v(\phi)=1$. We show that $\theta=\Prob\left(\bigwedge_{v(q)=1}q\wedge\bigwedge_{v(r)=0}\neg r\right){\geq}\one$ is a~full solution to~$\Pmbb$. First, $\theta$ is a~$\langle\Vmbf,\Vmc\rangle$-complete PIT, whence $\theta\not\models_{\FP}\bot$. Second, let $\Mfrak=\langle2^\Vmbf,\mu\rangle$ be a~probabilistic model s.t.\ $\Imc_\Mfrak(\theta)=1$. It follows that $\mu(\left\|\bigwedge_{v(q)=1}q\wedge\bigwedge_{v(r)=0}\neg r\right\|_\Mfrak)=1$, and thus $\mu(\|\phi\|_\Mfrak)=1$. Hence, $\theta\consvDashFPLuk\Prob(\phi)$ as required.

Conversely, let $\theta$ be a~full solution to~$\Pmbb$. Since $\theta$ is a~$\langle\Vmbf,\Vmc\rangle$-complete PIT, we can w.l.o.g.\ assume that $\theta=\Prob(\tau){\geq}\one$ with $\tau$~being an $\LCPL$-term s.t.\ $\Var(\tau)=\Vmbf$ and $\Prob(\tau){\geq}\one\consvDashFPLuk\Prob(\phi)$. Hence, by Definitions~\ref{def:FPLuksemantics} and~\ref{def:FPentailment} we have that $\tau\consvDashCPL\phi$. Thus, we can define a~classical valuation $v$ s.t.\ $v(q)=1$ iff $q\in\tau$. It is clear that $v(\phi)=1$, i.e., $\phi$ is $\CPL$-satisfiable.

For $\np$-membership observe from Theorem~\ref{theorem:fullsolutionrecognitionpolynomial} that it takes polynomial time w.r.t.\ the combined size of the given $\langle\Vmbf,\Vmc\rangle$-complete PIT~$\theta$ and an $\FP$ AP~$\Pmbb$ to check whether $\theta$ is a~full solution to~$\Pmbb$. Thus, given an $\FP$ AP, we can guess a~candidate for a~concise full solution and then verify it in polynomial time w.r.t.\ the size of~$\Pmbb$ (recall from Definition~\ref{def:concisesolution} that the size of a~concise solution to~$\Pmbb$ is bounded by the size of~$\Pmbb$).
\end{proof}
\subsection{Proofs of Section~\ref{ssec:prioritisation}}
\maxentropyconp*
\begin{proof}
We begin with $\conp$-hardness. We construct a~polynomial reduction from the validity of $\LCPL$-formulas, which is $\conp$-complete. Let $\phi\in\LCPL$, $q\notin\Var(\phi)$, and $\Vmbf=\Var(\phi)\cup\{q\}$, and consider the following $\FP$ abduction problem $\Pmbb_\mathsf{taut}=\langle\Gamma_\mathsf{taut},\delta_\mathsf{taut},\Hmsf_\mathsf{taut},\Vmc_\mathsf{taut}\rangle$ and PIT~$\theta_\mathsf{taut}$:
\begin{align*}
\Gamma_\mathsf{taut}&=\left\{\Prob\left(\phi\rightarrow\left(q\wedge\bigwedge_{p\in\Var(\phi)}p\right)\right)\right\}&\delta_\mathsf{taut}&=\Prob(q\vee\neg q)&
\Hmsf_\mathsf{taut}&=\{q\}&\Vmc_\mathsf{taut}&=\{0,\tfrac{1}{2},1\}\\
\theta_\mathsf{taut}&=\Prob\left(q\wedge\bigwedge_{p\in\Var(\phi)}p\right){\approx}\one
\end{align*}
Observe that $\theta_\mathsf{taut}$ is indeed a $\langle\Vmbf,\Vmc_\mathsf{taut}\rangle$-complete PIT and that $\Pmbb_\mathsf{taut}$ and $\theta_\mathsf{taut}$ are constructible in polynomial time. Moreover, it contains only four PILs, so it is concise ($\Pmbb_\mathsf{taut}$ contains two different probabilistic atoms, whence we can use up to six PILs in a~full solution for it to be concise). We show that $\phi$~is $\CPL$-valid iff $\theta_\mathsf{taut}$ is a~CEM solution to~$\Pmbb_\mathsf{taut}$.

First, assume that $\phi$ is $\CPL$-valid. Then we have that $\phi\rightarrow(q\wedge\bigwedge_{p\in\Var(\phi)}p)$ is (classically) equivalent to $q\wedge\bigwedge_{p\in\Var(\phi)}p$. Thus, $\Prob(\phi\rightarrow(q\wedge\bigwedge_{p\in\Var(\phi)}p))$ is $\FP$-equivalent to $\Prob(q\wedge\bigwedge_{p\in\Var(\phi)}p)$. It follows from Definitions~\ref{def:FPLuksemantics} and~\ref{def:FPLukabduction} that $\theta_\mathsf{taut}$ is a~full solution to~$\Pmbb_\mathsf{taut}$. In fact, it is \emph{the only} full solution to~$\Pmbb_\mathsf{taut}$ because any other probability distribution is inconsistent with~$\Gamma_\mathsf{taut}$. Hence, $\theta_\mathsf{taut}$ is (trivially) entropy-maximal. One can also compute that $\Hmbb(\theta_\mathsf{taut})=0$.

For the converse direction, assume that $\phi$ is not $\CPL$-valid. In this case, $\phi\rightarrow(q\wedge\bigwedge_{p\in\Var(\phi)}p)$ is \emph{not equivalent} to $q\wedge\bigwedge_{p\in\Var(\phi)}p$. Now let $v$~be a~classical valuation s.t.\ $v(\phi)=0$ and construct the following PIT:
\begin{align*}
\theta^+&=\left(\Prob\left(q\wedge\bigwedge_{p\in\Var(\phi)}p\right){=}\tfrac{1}{2}\right)\odot\left(\Prob\left(\neg q\wedge\bigwedge_{v(r)=1}r\wedge\bigwedge_{v(s)=0}\neg s\right){=}\tfrac{1}{2}\right)
\end{align*}
Observe that~$\theta^+$ is $\langle\Vmbf,\Vmc_\mathsf{taut}\rangle$-complete and concise (it contains four PILs). Note, moreover, that even if there is only one valuation that falsifies~$\phi$, $\theta^+$ is still correctly defined because $q\wedge\bigwedge_{p\in\Var(\phi)}p,\neg q\wedge\bigwedge_{v(r)=1}r\wedge\bigwedge_{v(s)=0}\neg s\models_\CPL\bot$ and $\Var(q\wedge\bigwedge_{p\in\Var(\phi)} p)=\Var(\neg q\wedge\bigwedge_{v(r)=1}r\wedge\bigwedge_{v(s)=0}\neg s)=\Vmbf$. One can also see from Definition~\ref{def:solutions} that~$\theta^+$ is a~complete solution to~$\Pmbb$. Finally, we can compute the entropy of~$\theta^+$ and obtain $\Hmbb(\theta^+)=1$. Hence, $\theta_\mathsf{taut}$ was not a~CEM solution.

For $\conp$-membership, we sketch an $\np$ procedure that solves the complementary problem. Given a~$\langle\Vmc,\Vmbf\rangle$-complete PIT~$\theta$, we guess another $\langle\Vmc,\Vmbf\rangle$-com\-plete PIT $\theta'$ with at most~$2\cdot k+2$ PILs, calculate $\Hmbb(\theta)$ and $\Hmbb(\theta')$, and verify that $\Hmbb(\theta)<\Hmbb(\theta')$ in polynomial time. If we succeed, $\theta$ \emph{is not a~CEM solution}.
\end{proof}
\PITcomplexity*
\begin{proof}
We begin with the Item~1. $\np$-membership follows from Proposition~\ref{prop:FPLukNPcoNP}. We provide a~reduction from the classical \emph{non-validity problem of DNFs}, which is $\np$-complete. Namely, let $\phi=\bigvee\limits^m_{i=1}\bigwedge\limits^n_{j=1}l^i_j$ be a~formula in DNF, and define 
\begin{align*}
\zeta_\phi&\coloneqq\bigodot\limits^m_{i=1}\left(\Prob\left(\bigwedge\limits^{n}_{j=1}l^i_j\right)\leq\zero\right)
\end{align*}
We show that there is some classical valuation $v$ s.t.\ $v(\phi)=0$ iff $\zeta_\phi$ is $\FP$-satisfiable. Let first $\phi$ be $\CPL$-valid. It is clear that $\mu(\|\phi\|_\Mfrak)\!=\!1$ for every probabilistic model $\Mfrak\!=\!\langle2^{\Var(\phi)},\mu\rangle$. This means that $\sum^m_{i=1}\mu(\|\bigwedge^{n}_{j=1}l^i_j\|)\geq1$, whence, $\Imc_\Mfrak(\Prob(\bigwedge^{n}_{j=1}\!l^i_j))>0$ for some $i\in\{1,\ldots,m\}$. This means that for every $\Mfrak$, there is at least one $i\in\{1,\ldots,m\}$ s.t.\ $\Imc_\Mfrak(\Prob(\bigwedge^{n}_{j=1}\!l^i_j)\leq\zero)=0$. Hence, $\zeta_\phi$ is $\FP$-unsatisfiable.

Conversely, let $\phi$ be classically \emph{non-valid}. Then there is some classical valuation~$v$ s.t.\ $v(\phi)=0$. Since $\phi$ is in DNF, $v(\bigwedge^n_{j=1}l^i_j)=0$ for every $i\in\{1,\ldots,m\}$. We need to construct a~probabilistic model $\Mfrak=\langle2^{\Var(\phi)},\mu\rangle$ s.t.\ $\Imc_\Mfrak(\zeta_\phi)=1$. This means that $\mu(\|\bigwedge^n_{j=1}l^i_j\|_\Mfrak)=0$ for every $i\in\{1,\ldots,m\}$. Since $\mu$ is a~probability measure, it suffices to define it on the atoms of $2^{2^{\Var(\phi)}}$, i.e., on subsets of variables of~$\phi$. Now, let $\Mfrak=\langle2^{\Var(\phi)},\mu\rangle$ be a~probabilistic model s.t.\ $\mu(\{X\})=1$ when $X=\{p\mid v(p)=1\}$ and $\mu(\{Y\})=0$ for every other $Y\subseteq\Var(\phi)$ s.t.\ $Y\neq X$. It is clear that $\mu(\|\bigwedge^n_{j=1}l^i_j\|_\Mfrak)=0$ for every $i\in\{1,\ldots,m\}$. Hence, $\Imc_\Mfrak(\zeta_\phi)=1$, as required.

In the second item, $\conp$-membership follows from Proposition~\ref{prop:FPLukNPcoNP}. For hardness, that we reduce $\FP$-\emph{unsatisfiability} of PITs to entailment as follows: $\zeta$ is $\FP$-unsatisfiable iff $\zeta\models_{\FP}\Prob(p)>\one$. (Note that $\Prob(p)>\one$ is $\FP$-unsatisfiable.)
\end{proof}

Before proceeding to the proof of Theorem~\ref{theorem:FPminimalrecognitionDP}, we first formalise the notion of the ‘next weakest’ PIL. The next definitions generalise a~similar notion from~\cite[cf.~\href{https://arxiv.org/abs/2507.13847}{arXiv:2507.13847}]{InoueKozhemiachenko2025}.

\begin{definition}\label{def:hypothesesordering}
Let $\Hmsf$ be a~finite set of $\LCPL$-terms and $\sigma\in\Hmsf$. We say that
\begin{itemize}[noitemsep,topsep=0pt]
\item $\varrho$ is a~\emph{next strongest term in~$\Hmsf$} w.r.t.\ $\sigma$ ($\varrho\Rrightarrow_\Hmsf\sigma$) iff $\varrho\models_\CPL\sigma$, $\sigma\not\models_\CPL\varrho$, and there is no $\sigma'\in\Hmsf$ s.t.\ $\sigma'\models_\CPL\sigma$, $\varrho\models_\CPL\sigma$, and $\sigma'\not\models_\CPL\varrho$;
\item $\tau$ is a~\emph{next weakest term in~$\Hmsf$} w.r.t.\ $\sigma$ ($\tau\Lleftarrow_\Hmsf\sigma$) iff $\sigma\models_\CPL\tau$, $\tau\not\models_\CPL\sigma$, and there is no $\sigma''\in\Hmsf$ in s.t.\ $\sigma\models_\CPL\sigma''$, $\sigma''\models_\CPL\tau$, and $\tau\not\models_\CPL\sigma''$.
\end{itemize}
\end{definition}

The following statement is evident.
\begin{lemma}\label{lemma:hypothesesorderingpolynomial}
Given a~set of $\LCPL$-terms $\Hmsf$ and $\varrho,\sigma,\tau\in\Hmsf$, it takes polynomial time w.r.t.\ the size of~$\Hmsf$ to determine whether~$\varrho$ is a~next strongest term in~$\Hmsf$ and whether~$\tau$ is a~next weakest term in~$\Hmsf$ w.r.t.~$\sigma$.
\end{lemma}
\begin{definition}\label{def:nextweakestPIL}
Let $\Vmc=\{0,\tfrac{1}{n},\ldots,\tfrac{n-1}{1},1\}$, $\diamond\in\{\leq,<,\geq,>\}$, and $\Hmsf\subseteq\LCPL$ be a~finite set of $\LCPL$-terms. Let further $\lambda=\Prob(\sigma)\diamond\overline{\tfrac{k}{n}}$ with $\sigma\in\Hmsf$. We define $\lambda^\flat_\Vmc$ as follows.
\begin{align*}
(\Prob(\sigma){\geq}\overline{\tfrac{k}{n}})^\flat_\Vmc&=\Prob(\sigma){>}\overline{\tfrac{k-1}{n}}&(\Prob(\sigma){>}\overline{\tfrac{k}{n}})^\flat_\Vmc&=\Prob(\sigma){\geq}\overline{\tfrac{k}{n}}&(\Prob(\sigma){\leq}\overline{\tfrac{k}{n}})^\flat_\Vmc&=\Prob(\sigma){<}\overline{\tfrac{k+1}{n}}&(\Prob(\sigma){<}\overline{\tfrac{k}{n}})^\flat_\Vmc&=\Prob(\sigma){\leq}\overline{\tfrac{k}{n}}
\end{align*}
Furthermore, let $\lambda_1=\Prob(\sigma){\geq}\overline{\tfrac{k}{n}}$, $\lambda_2=\Prob(\sigma){>}\overline{\tfrac{k}{n}}$, $\lambda_3=\Prob(\sigma){\leq}\overline{\tfrac{k}{n}}$, and $\lambda_4=\Prob(\sigma){<}\overline{\tfrac{k}{n}}$ be PILs. We define:
\begin{align*}
\Lambda^\flat_{\lambda_1,\Hmsf}&=\{\Prob(\tau){\geq}\overline{\tfrac{k}{n}}\mid\tau\Lleftarrow_\Hmsf\sigma\}&
\Lambda^\flat_{\lambda_2,\Hmsf}&=\{\Prob(\tau){>}\overline{\tfrac{k}{n}}\mid\tau\Lleftarrow_\Hmsf\sigma\}\\
\Lambda^\flat_{\lambda_3,\Hmsf}&=\{\Prob(\varrho){\leq}\overline{\tfrac{k}{n}}\mid\varrho\Rrightarrow_\Hmsf\sigma\}&
\Lambda^\flat_{\lambda_4,\Hmsf}&=\{\Prob(\varrho){<}\overline{\tfrac{k}{n}}\mid\varrho\Rrightarrow_\Hmsf\sigma\}
\end{align*}
\end{definition}

The next lemma can be straightforwardly obtained from Definition~\ref{def:nextweakestPIL}.
\begin{lemma}\label{lemma:nextweakestPIL}
Let $\Vmc=\{0,\tfrac{1}{n},\ldots,\tfrac{n-1}{1},1\}$ and $\Hmsf$ be a~finite set of $\LCPL$-terms. Let further, $\lambda$ be a~PIL s.t.\ $\Emc(\lambda)\subseteq\Hmsf$, $\lambda^\flat_\Vmc$ exist and $\lambda^*\in\Lambda^\flat_{\lambda,\Hmsf}$. Then the following statements hold.
\begin{enumerate}[noitemsep,topsep=0pt]
\item There is no PIL $\lambda'$ s.t.\ $\Emc(\lambda')\subseteq\Hmsf$, $\lambda\models_\FP\lambda'$, $\lambda'\models_\FP\lambda^\flat_\Vmc$, $\lambda'\not\models_\FP\lambda$, and $\lambda^\flat_\Vmc\not\models_\FP\lambda'$.
\item There is no PIL $\lambda''$ s.t.\ $\Emc(\lambda'')\subseteq\Hmsf$, $\lambda\models_\FP\lambda''$, $\lambda''\models_\FP\lambda^*$, $\lambda''\not\models_\FP\lambda$, and $\lambda^*\not\models_\FP\lambda''$.
\end{enumerate}
\end{lemma}
\begin{proof}[Proof sketch]
We begin with Item~1 and consider only the case of $\lambda=\Prob(\sigma){\geq}\cvalue$ with $c=\tfrac{k}{n}$ for brevity. Other cases can be dealt with similarly. Now, assume for contradiction that there is some $\lambda'$ s.t.\ $\lambda\models_\FP\lambda'$, $\lambda'\models_\FP\lambda^\flat_\Vmc$, $\lambda'\not\models_\FP\lambda$, and $\lambda^\flat_\Vmc\not\models_\FP\lambda'$. We consider the following cases of~$\lambda'$: (i)~$\lambda'=\Prob(\sigma'){\geq}\cvalue'$; (ii)~$\lambda'=\Prob(\sigma'){>}\cvalue'$; (iii)~$\lambda'=\Prob(\sigma'){\leq}\cvalue'$; (iv)~$\lambda'=\Prob(\sigma'){<}\cvalue'$.

In the first case, we have $c'=c$. If $c'>c$, then $\Prob(\sigma){\geq}\cvalue\not\models_\FP\Prob(\sigma'){\geq}\cvalue'$. If $c'<c$, then $\Prob(\sigma'){\geq}\cvalue'\not\models_\FP\Prob(\sigma){\geq}\overline{\tfrac{k-1}{n}}$. Now, one can see that $\Prob(\sigma){\geq}\cvalue\models_\FP\Prob(\sigma'){\geq}\cvalue$ entails $\sigma\models_\CPL\sigma'$. If $\sigma'\not\models_\CPL\sigma$, then $\Prob(\sigma'){\geq}\overline{\tfrac{k}{n}}\not\models_\FP\Prob(\sigma){>}\overline{\tfrac{k-1}{n}}$, which contradicts the assumption that $\lambda'\models_\FP\lambda^\flat_\Vmc$. Otherwise, $\lambda'\models_\FP\lambda$, contradicting the assumption that $\lambda'\not\models_\FP\lambda$.

In the second case, we can show that $c'=\tfrac{k-1}{n}$ using a~similar reasoning. From here, it will follow that either $\lambda^\flat_\Vmc\models_\FP\lambda'$ or $\lambda\not\models_\FP\lambda'$. Again, the reasoning is similar to the previous case. In the third case, one can check that $\Prob(\sigma){\geq}\overline{\tfrac{k}{n}}\models_\FP\Prob(\sigma'){\leq}\cvalue'$ implies that $\sigma,\sigma'\models_\CPL\bot$ and $c'\leq\tfrac{n-k}{n}$. But in this case, we have $\Prob(\sigma'){\leq}\cvalue'\not\models_\FP\Prob(\sigma){>}\overline{\tfrac{k-1}{n}}$, contrary to the assumption. The reasoning in the fourth case is similar to the reasoning in the third case.

For Item~2, let $\lambda=\Prob(\sigma){\geq}\cvalue$ and $\lambda^*=\Prob(\tau){\geq}\cvalue$ with $\tau\Lleftarrow_\Hmsf\sigma$. We assume for contradiction that there is some $\lambda''$ s.t.\ $\lambda\models_\FP\lambda''$, $\lambda''\models_\FP\lambda^*$, $\lambda''\not\models_\FP\lambda$, and $\lambda^*\not\models_\FP\lambda''$. As in Item~1, we have four cases of~$\lambda''$: (i)~$\lambda''=\Prob(\sigma''){\geq}\cvalue''$; (ii)~$\lambda''=\Prob(\sigma''){>}\cvalue''$; (iii)~$\lambda''=\Prob(\sigma''){\leq}\cvalue''$; (iv)~$\lambda''=\Prob(\sigma''){<}\cvalue''$.

In case~(i), $\Prob(\sigma){\geq}\cvalue\models_\FP\Prob(\sigma''){\geq}\cvalue''$ entails that $c''=c$ (recall the proof of Item~1). Moreover, $\Prob(\sigma){\geq}\cvalue\models_\FP\Prob(\sigma''){\geq}\cvalue$ implies that $\sigma\models_\CPL\sigma''$. Now, observe from Definitions~\ref{def:hypothesesordering} and~\ref{def:nextweakestPIL}, that this entails that either (a)~$\sigma''\models_\CPL\sigma$ or (b)~$\tau\models_\CPL\sigma$. But then (a)~contradicts the assumption that $\lambda''\not\models_\FP\lambda$, and (b)~contradicts the assumption that $\lambda^*\not\models_\FP\lambda''$. Cases (ii)--(iv) can be considered similarly.
\end{proof}

Let us also recall from~\cite{InoueKozhemiachenko2025} the notion of $\models_\Luk$-minimal solutions to $\Luk$-abduction problems.
\begin{definition}\label{def:Lukminsolutions}
A~solution $\tau$ to an $\Luk$-abduction problem $\Pmbb=\langle\Phi,\chi,\Hmsf\rangle$ is called \emph{proper} iff $\tau\not\models_\Luk\chi$. A~proper solution~$\tau$ is \emph{$\models_\Luk$-minimal} iff there is no other proper solution $\sigma$ s.t.\ $\tau\models_\Luk\sigma$ and $\sigma\not\models_\Luk\tau$.
\end{definition}
\FPminimalrecognitionDP*
\begin{proof}
We begin with the proof of $\DP$-hardness. For this, we provide a~polynomial reduction from the recognition of $\models_\Luk$-minimal solutions to~$\Luk$-abduction problems, which is $\DP$-complete. Namely, we consider the following class of $\Luk$-abduction problems $\Pmbb=\langle\Phi\cup\{q\},\chi\odot q,\Hmsf\rangle$: $\Phi=\{p\vee\neg p\mid p\in\Var(\chi)\cup\{q\}\}$, $\chi\in\LLukint$, $q\notin\Var(\chi)$, $\Hmsf=\{p{\geq}\one\mid p\in\Var(\chi)\}\cup\{p{\leq}\zero\mid p\in\Var(\chi)\}$.\footnote{Observe from the proof of Theorem~3 in~\cite{InoueKozhemiachenko2025} that recognising $\models_\Luk$-minimal solutions to such problems is $\DP$-complete.} We construct an $\FP$ AP $\Pmbb^\Prob=\langle\Phi^\Prob,\chi^\Prob,\Hmsf',\Vmc\rangle$ as shown in the proof of Theorem~\ref{theorem:sufficientsolutionrecognitionDP}. Let~$\tau$ be an interval term. We show that $\tau$ is an $\models_\Luk$-minimal solution to~$\Pmbb$ iff $\tau^\Prob$ is an $\models_\FP$-minimal solution to~$\Pmbb^\Prob$.

Let $\tau$ be an $\models_\Luk$-minimal solution to~$\Pmbb$. It follows that $\Phi,q,\tau\consvDashLuk\chi\odot q$ and there is no interval term $\sigma$ s.t.\ $\Phi,q,\sigma\consvDashLuk\chi\odot q$, $\tau\models_\Luk\sigma$, and $\sigma\not\models_\Luk\tau$. By Proposition~\ref{prop:LukisfragmentofFPLuk}, we obtain that $\Phi^\Prob,\Prob(q),\tau^\Prob\consvDashFPLuk\chi^\Prob\odot\Prob(q)$. Thus, $\tau^\Prob$ is a~sufficient solution to~$\Pmbb^\Prob$. Now, assume for contradiction that there is some other sufficient solution~$\eta$ to~$\Pmbb^\Prob$ s.t.\ $\tau^\Prob\models_\FP\eta$ and $\eta\not\models_\FP\tau^\Prob$. Now define $\eta^\uparrow$ as in~\eqref{equ:etauparrow}. Using the same reasoning as in Theorem~\ref{theorem:sufficientsolutionrecognitionDP}, we obtain that $\eta^\uparrow$ is a~solution to~$\Pmbb$. Moreover, it is easy to see using Proposition~\ref{prop:LukisfragmentofFPLuk} that $\tau\models_\Luk\eta^\uparrow$ but $\eta^\uparrow\not\models_\FP\tau$. This contradicts the assumption that $\tau$ is an~$\models_\Luk$-solution to~$\Pmbb$.

Conversely, let $\eta$ be an $\models_\FP$-minimal solution to~$\Pmbb^\Prob$. Then $\Phi^\Prob,\Prob(q)\consvDashFPLuk\chi^\Prob\odot\Prob(q)$ and, additionally, there is no other sufficient solution $\zeta$ s.t.\ $\eta\models_\FP\zeta$ and $\zeta\not\models_\FP\eta$. Moreover, since $q\notin\Hmsf'$ (because $q\notin\Var[\Hmsf]$), we have $\eta\not\models_\FP\chi^\Prob\odot\Prob(q)$. Now, we can define~$\eta^\uparrow$ as shown in~\eqref{equ:etauparrow}. By the same reasoning as in Theorem~\ref{theorem:sufficientsolutionrecognitionDP}, we can show that $\eta^\uparrow$ is a~(proper) solution of~$\Pmbb$.

Now, assume for contradiction that there is some other proper solution $\tau$ of~$\Pmbb$ s.t.\ $\eta^\uparrow\models_\Luk\tau$ and $\tau\not\models_\Luk\eta^\uparrow$. Since $\eta^\uparrow$ and $\tau$ contain only interval literals of the form $p{\geq}\one$ and $q{\leq}\zero$, it follows that every literal occurring in~$\tau$ also occurs in~$\eta^\uparrow$, and there are some interval literals in~$\eta^\uparrow$ that do not occur in~$\tau$. By Proposition~\ref{prop:LukisfragmentofFPLuk}, this implies that $(\eta^\uparrow)^\Prob\models_\FP\tau^\Prob$, $\tau^\Prob\not\models_\FP(\eta^\uparrow)^\Prob$, and that $\tau^\Prob$ is a~sufficient solution to~$\Pmbb^\Prob$. Furthermore, observe that $(\eta^\uparrow)^\Prob=\eta^\pitchfork$ (recall~\eqref{equ:etapitchfork} for~$\eta^\pitchfork$). Now given a~PIT $\eta$, let $\eta^\ddag$ denote the result of replacing $\Prob(p){\geq}\one$ and $\Prob(q){\leq}\zero$ with, respectively, $\Prob(p){>}\zero$ and $\Prob(q){<}\one$ in~$\eta$. Observe from~\eqref{equ:LEMequivalence} that $(\tau^\Prob)^\ddag$ is also a~sufficient solution to~$\Pmbb^\Prob$.

Now, note that $\Prob(p){\geq}\one\models_\FP\Prob(p){>}\zero$ and $\Prob(q){\leq}\zero\models_\FP\Prob(q){<}\one$. Moreover, $\eta^\pitchfork\models_\FP\tau^\Prob$, $\tau^\Prob\models_\FP(\tau^\Prob)^\ddag$, and the only difference between $\eta^\pitchfork$ and $\eta$ is that some $\Prob(p){\geq}\one$'s and $\Prob(q){\leq}\zero$'s were replaced with $\Prob(p){>}\zero$'s and $\Prob(q){<}\one$. But observe that since all interval literals in~$\tau$ occur in~$\eta^\uparrow$, then for each PIT $\lambda$ in~$(\tau^\Prob)^\ddag$, it either occurs in~$\eta$, or its stronger counterpart ($\Prob(p){\geq}\one$ or $\Prob(q){\leq}\zero$) occurs in~$\eta$. Additionally, there are interval literals that occur only in~$\eta^\uparrow$ but not in~$\tau$. Hence, they (or their weaker counterparts) occur in~$\eta$, but neither they nor their weaker counterparts occur in~$(\tau^\Prob)^\ddag$. Hence, $\eta\models_\FP(\tau^\Prob)^\ddag$ and $(\tau^\Prob)^\ddag\not\models_\FP\eta$. Thus, $\eta$ is not an $\models_\FP$-minimal solution to~$\Pmbb^\Prob$, contrary to the assumption. This gives us $\DP$-hardness.

Now, consider $\DP$-membership. Let $\eta=\bigodot^n_{i=1}\lambda_i$ be a~PIT and $\Pmbb=\langle\Gamma,\delta,\Hmsf,\Vmc\rangle$ an~$\FP$-abduction problem. To check that $\eta$~is an $\models_\FP$-minimal sufficient solution to~$\Pmbb$, we need to check that it is a~sufficient solution and that there is no other sufficient solution $\eta$ s.t.\ $\eta\models_\FP\eta$ and $\eta\not\models_\FP\eta$. For every PIL~$\lambda$ in~$\eta$, we construct
\begin{align*}
\eta^\flat_{\lambda,\Vmc}&=
\begin{cases}
\bigodot\limits_{\myoverset{\lambda'\in\eta}{\lambda'\neq\lambda}}\lambda'\odot\lambda^\flat_\Vmc&\text{if }\lambda^\flat_\Vmc\text{ is defined}\\
\bigodot\limits_{\myoverset{\lambda'\in\eta}{\lambda'\neq\lambda}}\lambda'&\text{otherwise}
\end{cases}&
\Theta^\flat_{\lambda,\Hmsf}&=\left\{\bigodot\limits_{\myoverset{\lambda'\in\eta}{\lambda'\neq\lambda}}\lambda'\odot\lambda^*\mid\lambda^*\in\Lambda^\flat_{\lambda,\Hmsf}\right\}
\end{align*}
as shown in Definition~\ref{def:nextweakestPIL}. Note that the size of~$\eta^\flat_{\lambda,\Vmc}$ is polynomial in the size of~$\eta$ and the size of~$\Theta^\flat_{\lambda,\Hmsf}$ is polynomial in the size of~$\Pmbb$ because there are at most $|\Hmsf|$~next weakest or next strongest terms in~$\Hmsf$ w.r.t.\ a~given $\sigma\in\Hmsf$. It also follows from Definition~\ref{def:nextweakestPIL} that $\eta\models_\FP\eta^*$ and $\eta^*\not\models_\FP\eta$ for every $\eta^*\in\Theta^\flat_{\lambda,\Hmsf}$. Furthermore, by Lemma~\ref{lemma:hypothesesorderingpolynomial}, $\Theta^\flat_{\lambda,\Hmsf}$ can be constructed in polynomial time. From Lemma~\ref{lemma:nextweakestPIL}, it is clear that given~$\eta$, it suffices to check that none of $\eta^\flat_{\lambda,\Vmc}$'s and $\eta^*\in\Theta^\flat_{\lambda,\Hmsf}$ is a~solution to~$\Pmbb$.

Thus, the $\DP$ procedure is as follows. Given~$\eta=\bigodot^n_{i=1}\lambda_i$ and~$\Pmbb$, we use an $\np$-oracle to guess that (i)~$\Gamma,\eta\not\models_\FP\bot$, (ii)~$\Gamma,\eta\models_\FP\delta$, (iii)~$\Gamma,\eta^*\not\models_\FP\delta$ for every PIL~$\lambda$ in~$\eta$ and every $\eta^*\in\Theta^\flat_{\lambda,\Hmsf}$. Note that calls~(i) and~(iii) are $\np$, call~(ii) is $\conp$, all of them are independent, and there are $\Omc(n^2\cdot|\Hmsf|\cdot|\Vmc|)$ calls in total. This gives us the desired $\DP$-membership.
\end{proof}
\section{Proofs of Section~\ref{sec:fragments}}
We begin with the definition of \emph{entailment-monotone} functions.
\begin{definition}\label{def:entailmentmonotonicity}
Let $\Upsilon\subseteq\LCPL$. A~function $f:\Upsilon\rightarrow[0,1]$ is called \emph{$\models_\CPL$-monotone} iff the following condition holds:
\begin{align*}
\forall \pi,\varrho\in\Upsilon:\pi\models_\CPL\varrho\Longrightarrow f(\pi)\leq f(\varrho)
\end{align*} 
\end{definition}
\begin{lemma}\label{lemma:monotoneassignmentCIP}
Let $\Upsilon\subseteq\LCPL$ be a~finite chained positive set and $\Var[\Upsilon]=\{p_1,\ldots,p_k\}=\Vmbf$. Then for every $\models_\CPL$-monotone $f:\Upsilon\rightarrow[0,1]$, there is a~probabilistic model $\Mfrak_f=\langle2^\Vmbf,\mu\rangle$ s.t.\ $\mu(\|\pi\|_\Mfrak)=f(\pi)$ for every $\pi\in\Upsilon$.
\end{lemma}
\begin{proof}
Let $\Upsilon=\biguplus^n_{i=1}\Theta_i$ with each $\Theta_i$ being linearly ordered w.r.t.\ $\models_\CPL$. Observe that $|\Upsilon|\leq2\cdot k-n$. Indeed, notice that since $\Var[\Theta_i]\cap\Var[\Theta_j]=\varnothing$ when $i\neq j$, it follows that the largest possible chain has the following form
\begin{align*}
\bigwedge^k_{i=1}p_i\models_\CPL\bigwedge^k_{i=2}p_i\models_\CPL\ldots\models_\CPL p_k\models_\CPL p_{k-1}\vee p_k\models_\CPL\ldots\models_\CPL\bigvee^k_{i=1}p_i
\end{align*}
and contains $2\cdot k-1$ formulas. If there are $n$ chains, then each will contain at most $2\cdot k_j-1$ formulas, where $k_j$ is the number of variables in the largest formula in the chain~$j$.

Furthermore, one can see that $\Upsilon\not\models_\CPL\bot$ because all formulas are $\neg$-free. Now, for each $\pi\in\Upsilon$, let $[\pi]_f=\{\varrho\mid \varrho\in\Upsilon,f(\pi)=f(\varrho)\}$ and $\sharp_f(\pi)$ be the number of $[\varrho]_f$'s s.t.\ $f(\varrho)\leq f(\pi)$. It is clear that $\sharp_f(\pi)\leq|\Upsilon|$ for every $\pi\in\Upsilon$. Set further, $r=\max\{\sharp_f(\pi)\mid\pi\in\Upsilon\}$. Moreover, one can see that there are no $\Upsilon',\Upsilon''\subseteq\Upsilon$ s.t.\ $\Upsilon'\cap\Upsilon''=\varnothing$ and $\CPL\models\bigvee_{\pi\in\Upsilon'}\pi\leftrightarrow\bigwedge_{\varrho\in\Upsilon''}\varrho$.

Now, given $w\in2^\Vmbf$, define
\begin{align*}
\rank(w)&=
\begin{cases}
0&\text{if there is no }\pi\in\Upsilon\text{ s.t.\ }w\models_\CPL\pi\\
\min\{\sharp_f(\pi)\mid\forall\pi':\sharp_f(\pi)=\sharp_f(\pi')\Rightarrow w\models_\CPL\pi'\}&\text{otherwise}
\end{cases}
\end{align*}
It is clear that $\rank(\bigwedge^k_{i=1}p_i)=1$ because $w\models_\CPL\pi$ for all $\pi\in\Upsilon$ and $\rank(\bigwedge^k_{i=1}\neg p_i)=0$. Furthermore, for every $i\leq r$, there is some (maybe, several) $w\in2^\Vmbf$ s.t.\ $\rank(w)=i$. For each such $i$, we choose one $w\in2^\Vmbf$ s.t.\ $\rank(w)=i$ and denote it with $w_i$ (in particular, we set $w_1=\bigwedge^k_{i=1}p_i$ and $w_0=\bigwedge^k_{i=1}\neg p_i$). We can now define the probability distribution $\partial_f$ on $2^\Vmbf$ as follows.
\begin{enumerate}[noitemsep,topsep=0pt]
\item We set $\partial_f(w_1)=\min\{f(\pi)\mid\pi\in\Upsilon\}$.
\item Then, for each $i\in\{2,\ldots,r\}$, pick $\pi\in\Upsilon$ s.t.\ $\sharp_f(\pi)=i$ and set $\partial_f(w_i)=f(\pi)-\partial_f(w_{i-1})$.
\item If $\max\{f(\pi)\mid\pi\in\Upsilon\}<1$, set additionally, $\partial(w_0)=1-\partial(w_r)$.
\end{enumerate}
It is clear that the probability measure $\mu_{\partial_f}$ induced by~$\partial_f$ is coherent with~$f$. The result now follows.
\end{proof}

\simpletranslation*
\begin{proof}
We define $\Psi_\Gamma$ as follows:
\begin{align}\label{equ:PsiGamma}
\Psi_\Gamma&=\{\alpha^\uparrow\mid\alpha\in\Gamma\}\cup\{p_\phi\rightarrow p_\chi\mid\phi,\chi\in\Emc[\Gamma]\text{ and }\phi\models_\CPL\chi\}
\end{align}
Observe that $\Psi_\Gamma$ indeed has polynomial w.r.t.\ $\Gamma$ size. Moreover, since formulas in~$\Emc[\Gamma]$ are either conjunctions or disjunctions of variables, it takes polynomial time to verify whether $\phi\models_\CPL\chi$ for each pair of $\phi,\chi\in\Emc[\Gamma]$.

It is clear from Definition~\ref{def:FPLuksemantics} that if $\Gamma$ is $\FP$-satisfiable, then $\Psi_\Gamma$ is $\Luk$-satisfiable. Indeed, if $\Mfrak$ satisfies $\Gamma$, it suffices to set $v(p_\phi)=\Imc_\Mfrak(\Prob(\phi))$. For the converse direction, let $v$ be an $\Luk$-valuation s.t.\ $v(\psi)=1$ for every $\psi\in\Psi_\Gamma$. We need to define a~probabilistic model $\Mfrak_v$ that satisfies $\Phi_\Gamma$, i.e., a~measure $\mu$ on $2^{2^\Vmbf}$ s.t.\ $\Imc_{\Mfrak_v}(\alpha)=1$ for every $\alpha\in\Gamma$. We define $f(\phi)=v(p_\phi)$ for every $\phi\in\Emc[\Gamma]$. Note that $f$ is $\models_\CPL$-monotone because $v(p_\phi\rightarrow p_\chi)=1$ for every $\phi,\chi\in\Emc[\Gamma]$ s.t.\ $\phi\models_\CPL\chi$ and $v(p_\phi\rightarrow p_\chi)=1$ iff $v(p_\phi)\leq v(p_\chi)$. Thus, by Lemma~\ref{lemma:monotoneassignmentCIP}, there is a~probabilistic model $\Mfrak_f=\langle2^\Vmbf,\mu_f\rangle$ s.t.\ $\mu_f(\|\phi\|_\Mfrak)=f(\phi)=v(p_\phi)$ for all $\phi\in\Emc[\Gamma]$. It now follows from Definitions~\ref{def:Luksemantics} and~\ref{def:FPLuksemantics} that $\Mfrak_f$ satisfies~$\Gamma$.
\end{proof}
\subsection{Proofs of Section~\ref{ssec:PSC}}
Before proceeding to the proof of Theorem~\ref{theorem:PSCrecognitionPcomplete}, let us recall the notion of Łukasiewicz simple-clause abduction problems from~\cite{InoueKozhemiachenko2025}.
\begin{definition}\label{def:simpleclauseabduction}
A~\emph{simple clause abduction problem} (SC AP) is an $\Luk$-abduction problem $\Pmbb=\langle\Gamma,\chi,\Hmsf\rangle$ s.t.\ $\Gamma$ is a~set of simple clauses and interval terms and $\chi$ is an interval clause, simple clause, interval term, or a~simple term.
\end{definition}

\PSCrecognitionPcomplete*
\begin{proof}
We begin with $\Pmsf$-hardness of recognising arbitrary sufficient solutions. Using Proposition~\ref{prop:LukisfragmentofFPLuk}, we will provide a~logspace reduction from the solution recognition for SC APs to sufficient solution recognition for PSC APs. Namely, let $\Pmbb=\langle\Phi,p,\Hmsf\rangle$ be an SC AP s.t.\ $\Phi$ is a~set of simple clauses, $p\in\Var$, $p\notin\Var[\Hmsf]$, and $\Hmsf$~contains only interval literals of the form $q{\geq}\one$ and $r{\leq}\zero$. Note from~\cite[Theorem~8]{InoueKozhemiachenko2025} that solution recognition for such SC APs is $\Pmsf$-hard. We define $\Pmbb^\Prob=\langle\Phi^\Prob,\Prob(p),\Hmsf',\Vmc'\rangle$ as shown in the proof of Theorem~\ref{theorem:sufficientsolutionrecognitionDP} (recall Definition~\ref{def:FPLuktoLukembedding} for $\Phi^\Prob$ and observe that $p^\Prob=\Prob(p)$).

Note from Definition~\ref{def:FPLuktoLukembedding} that if $\Phi\subseteq\LLukint$ is a~set of simple clauses and interval terms, then $(\Phi\cup\{p\})^\Prob$ is a~set of PSCs and PITs. Moreover, since all events in~$\Phi^\Prob$ are \emph{propositional variables}, $\Emc[(\Phi\cup\{p\})^\Prob]$ is chained positive, whence $(\Phi\cup\{p\})^\Prob$ is a~PSC theory and $\Pmbb^\Prob$~is a~PSC abduction problem. Thus, we can reuse the reasoning from Theorem~\ref{theorem:sufficientsolutionrecognitionDP} and use Proposition~\ref{prop:LukisfragmentofFPLuk} to show that given an interval literal~$\tau$ and an SC AP as defined in the previous paragraph, if $\tau$ is a~solution to~$\Pmbb$, then $\tau^\Prob$ is a~solution to~$\Pmbb^\Prob$. And conversely, if $\eta$ is a~sufficient solution to~$\Pmbb^\Prob$, then $\eta^\uparrow$ (recall~\eqref{equ:etauparrow} for the definition of~$\eta^\uparrow$) is a~solution to~$\Pmbb$. Finally, recall that the transformation of~$\Pmbb$ into~$\Pmbb^\Prob$ proceeds as follows: we replace every variable~$p$ in the theory and observation with $\Prob(p)$, form the set of hypotheses by copying variables from the hypotheses of~$\Pmbb$, and set $\Vmc=\{0,1\}$. This procedure utilises logarithmic space only. Thus, we obtain $\Pmsf$-hardness, as required.

Let us now proceed to membership. For that, we will construct a~polynomial reduction to solution recognition for SC APs. Now let $\Pmbb=\langle\Gamma,\delta,\Hmsf,\Vmc\rangle$ be a~PSC AP. We define an SC AP $\Pmbb^\mathsf{SC}=\langle\Gamma^\mathsf{SC},\delta^\uparrow,\Hmsf^\mathsf{SC}\rangle$ (recall Definition~\ref{def:outercounterpart} for~$\delta^\uparrow$) as follows.
\begin{align*}
\Gamma^\mathsf{SC}&=\{\gamma^\uparrow\mid\gamma\in\Gamma\}\cup\{p_\phi\rightarrow p_\chi\mid\phi,\chi\in\Emc[\Pmbb]\text{ and }\phi\models_\CPL\chi\}&
\Hmsf^\mathsf{SC}&=\{p_\tau\diamond\cvalue\mid\tau\in\Hmsf,c\in\Vmc,\diamond\in\{\leq,<,\geq,>\}\}
\end{align*}
We observe briefly that $\Pmbb^\mathsf{SC}$ is indeed a~simple clause abduction problem. Since every $\alpha\in\Gamma\cup\{\delta\}$ is either a~PSC or has the form $\Prob(\tau)\diamond\cvalue$, we have from Definition~\ref{def:outercounterpart} that every $\alpha^\uparrow\in\Gamma^\mathsf{SC}\cup\{\delta^\uparrow\}$ is either a~simple clause or an interval literal. Moreover, $p_\phi\rightarrow p_\chi$ is also a~simple clause because it can be equivalently represented as $\neg p_\phi\oplus p_\chi$. Furthermore, by Proposition~\ref{prop:simpletranslation}, it takes polynomial time to construct~$\Gamma^\mathsf{SC}$.

Assume now that $\eta$ is a~PIT. We show that $\eta$ is a~sufficient solution to~$\Pmbb$ iff $\eta^\uparrow$ is a~solution to~$\Pmbb^\mathsf{SC}$. If $\eta$ is a~sufficient solution to~$\Pmbb$, we have that $\Gamma,\eta\consvDashFPLuk\delta$. We show that $\Gamma^\mathsf{SC},\eta^\uparrow\consvDashLuk\delta^\uparrow$. Observe that $\Gamma^\mathsf{SC}$ is defined as $\Psi_\Gamma$ in~\eqref{equ:PsiGamma} and $\Gamma\cup\{\eta\}$ is a~CIP theory (recall Definition~\ref{def:PSCabduction}), whence $\Gamma^\mathsf{SC},\eta^\uparrow\not\models_\Luk\bot$ iff $\Gamma,\eta\not\models_\FP\bot$. Similarly, by Definition~\ref{def:FPLuksemantics}, we have that $\Gamma,\eta\models_\FP\delta$ follows from $\Gamma,\eta,\neg\triangle\delta\models_\FP\bot$. Notice that $\Gamma\cup\{\eta,\neg\triangle\delta\}$ is also a~CIP theory. Whence, again, by Proposition~\ref{prop:simpletranslation}, we have that $\Gamma,\eta,\neg\triangle\delta\models_\FP\bot$ implies $\Gamma^\mathsf{SC},\eta^\uparrow,(\neg\triangle\delta)^\uparrow\models_\Luk\bot$.

Conversely, let $\tau$ be a~solution to~$\Pmbb^\mathsf{SC}$. Using Proposition~\ref{prop:simpletranslation} and the same reasoning as in the previous paragraph, one can show that $\tau^\Prob$ (recall Definition~\ref{def:FPLuktoLukembedding} and observe that $(\tau^\Prob)^\uparrow=\tau$) is a~solution to~$\Pmbb^\mathsf{SC}$. The $\Pmsf$-completeness of recognising (arbitrary) sufficient solutions now follows.

Let us now consider the recognition of entailment-minimal sufficient solutions. We construct a~logspace reduction from the recognition of $\models_\Luk$-minimal solutions to simple-clause abduction problems in Łukasiewicz logic, which is $\Pmsf$-complete. Namely, we take an SC AP $\Pmbb=\langle\Phi,p,\Hmsf\rangle$ s.t.\ $\Phi$ is a~set of simple clauses, $p\in\Var$, $p\notin\Var[\Hmsf]$, and $\Hmsf$~contains only interval literals of the form $q{\geq}\one$ and $r{\leq}\zero$. By~\cite[Theorem~8]{InoueKozhemiachenko2025}, recognition of $\models_\Luk$-minimal solutions to such problems is $\Pmsf$-complete. We use the same reduction as for arbitrary solutions and then reason as in the proof of Theorem~\ref{theorem:FPminimalrecognitionDP}.

For $\Pmsf$-membership, we observe that it takes polynomial time to recognise whether a~given PIT~$\eta$ is a~sufficient solution to~$\Pmbb$. Furthermore, by Lemmas~\ref{lemma:hypothesesorderingpolynomial} and~\ref{lemma:nextweakestPIL}, it takes polynomial time to generate all ‘next weakest’ solution candidates, and there are polynomially many of them. If \emph{none of them} is a~sufficient solution, it follows that~$\eta$ is an $\models_\FP$-minimal sufficient solution.
\end{proof}

\PSCsolutionexistence*
\begin{proof}
Consider Item~1. $\np$-membership follows immediately from Theorem~\ref{theorem:PSCrecognitionPcomplete}. For $\np$-hardness, we construct a~polynomial reduction from solution existence to simple-clause abduction problems in Łukasiewicz logic, which is $\np$-complete~\cite[Theorem~10]{InoueKozhemiachenko2025}. We reuse the reduction from Theorem~\ref{theorem:PSCrecognitionPcomplete}. Namely, given an SC AP, we define a~PSC AP $\Pmbb^\Prob=\langle\Phi^\Prob,\Prob(p),\Hmsf',\Vmc'\rangle$ as shown in the proof of Theorem~\ref{theorem:sufficientsolutionrecognitionDP} (recall Definition~\ref{def:FPLuktoLukembedding} for $\Phi^\Prob$ and observe that $p^\Prob=\Prob(p)$). Note again that since $\Phi$~is a~set of simple clauses, $\Phi^\Prob$ is a~set of \emph{probabilistic} simple clauses.

Let us now proceed to Item~2. $\np$-membership follows from Theorem~\ref{theorem:fullsolutionexistenceNP}. For $\np$-hardness, we provide a~polynomial reduction from the $\CPL$-satisfiability of propositional formulas in the conjunctive normal form. Let $\phi=\bigwedge^n_{j=1}\left(\bigvee^m_{i=1}p^j_i\vee\bigvee^l_{i=1}\neg q_i\right)$ be an $\LCPL$-formula in CNF and $r\notin\Var(\phi)$ and set $\Vmbf=\Var(\phi)\cup\{r\}$. We define 
\begin{align*}
\Gamma_\Smsf&=\left\{\bigoplus^m_{i=1}\Prob(p^j_i)\oplus\bigoplus^l_{i=1}\neg\Prob(q^j_i)\mid j\in\{1,\ldots,n\}\right\}\cup\{\Prob(r)\}&\delta_\Smsf&=\Prob(r)&\Hmsf_\Smsf&=\{p^1_1\}&\Vmc_\Smsf&=\{0,1\}\\
\Pmbb_\Smsf&=\langle\Gamma_\Smsf,\delta_\Smsf,\Hmsf_\Smsf,\Vmc_\Smsf\rangle
\end{align*}
Observe that $\Pmbb_\Smsf$ can be constructed from $\phi$ in polynomial time. Moreover, since we are considering the existence of full solutions, the set of hypotheses in~$\Pmbb_\Smsf$ is irrelevant (whence we can set it as~$\{p^1_1\}$). Now we show that $\phi$ is $\CPL$-satisfiable iff there is a~full solution to~$\Pmbb_\Smsf$.

Assume that $\phi$ is $\CPL$-satisfiable and, in particular, $v$~is a~classical valuation s.t.\ $v(\phi)=1$. First of all, we observe that $\Gamma_\Smsf\models_\FP\delta_\Smsf$. It remains to construct a~$\langle\Vmbf,\Vmc\rangle$-complete term $\theta_v$ s.t.\ $\theta_v\models_\FP\bigodot_{\gamma\in\Gamma_\Smsf}\gamma\odot\delta_\Smsf$. (Recall from Definition~\ref{def:probabilisticintervalliterals} that all $\langle\Vmbf,\Vmc\rangle$-complete terms are $\FP$-satisfiable, thus if such a~term entails some formula~$\alpha$, it \emph{consistently entails}~$\alpha$.)

Now let 
\begin{align}\label{equ:completePITfromvaluation}
\tau_v&=r\wedge\bigwedge_{v(s)=1}s\wedge\bigwedge_{v(t)=0}\neg t&\theta_v&=\Prob(\tau_v){\approx}\one
\end{align}
We show that $\theta_v\models_\FP\bigodot_{\gamma\in\Gamma_\Smsf}\gamma\odot\delta_\Smsf$. Observe, first of all, that $\theta_v$ is $\langle\Vmbf,\Vmc\rangle$-complete. Let further, $\Mfrak_{\theta_v}=\langle2^\Vmbf,\mu_{\theta_v}\rangle$ be a~probabilistic model s.t.\ $\mu(\|\tau\|_{\Mfrak_{\theta_v}})=1$. By Definition~\ref{def:FPLuksemantics}, we have $\Imc_{\Mfrak_{\theta_v}}(\theta_v)=1$.

Let us show that $\Imc_{\Mfrak_{\theta_v}}\left(\bigodot_{\gamma\in\Gamma_\Smsf}\gamma\odot\delta_\Smsf\right)=1$. Observe that $v(h)=\Imc_{\Mfrak}(\Prob(h))$ for every $h\in\Var(\phi)$. Indeed, if $v(h)=1$, then $h\in\tau$, whence $\mu(\|h\|_{\Mfrak_{\theta_v}})=1$ and $\Imc_{\Mfrak_{\theta_v}}(\Prob(h))=1$. Likewise, if $v(h)=0$, then $\neg h\in\tau$, whence \mbox{$\mu(\|\neg h\|_{\Mfrak_{\theta_v}})=1$}, i.e., $\mu(\|h\|_{\Mfrak_{\theta_v}})=0$ and $\Imc_{\Mfrak_{\theta_v}}(\Prob(h))=0$. Now, since $\vee$ and $\oplus$ behave in the same way on $\{0,1\}$, it follows that $\Imc_{\Mfrak_{\theta_v}}\left(\bigodot_{\gamma\in\Gamma_\Smsf}\gamma\odot\delta_\Smsf\right)=v(\phi)=1$, as required. As there is \emph{only one} probabilistic model~$\Mfrak$ over~$2^\Vmbf$ s.t.\ $\mu(\|\tau\|_\Mfrak)=1$ (namely, $\Mfrak_{\theta_v}$), it follows that $\theta\models_\FP\bigodot_{\gamma\in\Gamma_\Smsf}\gamma\odot\delta_\Smsf$, as required.

For the converse direction, let $\theta$ be a~full solution to~$\Pmbb_\Smsf$ and let further, $\Mfrak=\langle2^\Vmbf,\mu\rangle$ be the probabilistic model s.t.\ $\Imc_{\Mfrak}(\theta)=1$ and $\Imc_\Mfrak\left(\bigodot_{\gamma\in\Gamma_\Smsf}\gamma\odot\delta_\Smsf\right)=1$. By Definition~\ref{def:solutions}, it follows that $\theta$~has the form~$\Prob(\tau){\approx}\one$ for some $\LCPL$-term $\tau$ s.t.\ $\Var(\tau)=\Vmbf$. Now let $v_\theta$ be a~classical valuation s.t.\
\begin{align*}
v_\theta(h)&=
\begin{cases}
1&\text{if }h\in\tau\\
0&\text{if }\neg h\in\tau
\end{cases}
\end{align*}
We show that $v_\theta(\phi)=1$. Observe that for every $\psi\in\LCPL$ s.t.\ $\Var(\psi)\subseteq\Vmbf$, it holds that $\mu(\|\psi\|_\Mfrak)=\sum_{\myoverset{\tau\models_\CPL\psi}{\Var(\tau)=\Vmbf}}\mu(\|\tau'\|_\Mfrak)$. Thus, for every $h\in\Vmbf$, $\Imc_\Mfrak(\Prob(h))\in\{0,1\}$. Moreover, one can see that if $h\in\tau$, then $\Imc_\Mfrak(\Prob(h))=1$ and if $\neg h\in\tau$, then $\Imc_\Mfrak(\Prob(h))=0$. Thus, if $\Imc_\Mfrak(\Prob(h))=1$, then $v_\theta(h)=1$, and if $\Imc_\Mfrak(\Prob(h))=0$, then $v_\theta(h)=0$. Furthermore, $\Imc_\Mfrak\left(\bigodot_{\gamma\in\Gamma_\Smsf}\gamma\odot\delta_\Smsf\right)=1$ implies that $\Imc_\Mfrak\left(\bigoplus^m_{i=1}\Prob(p^j_i)\oplus\bigoplus^l_{i=1}\neg\Prob(q^j_i)\right)=1$ for every $j\in\{1,\ldots,n\}$. But then, since all probabilistic atoms have values in~$\{0,1\}$ and since $\vee$ and $\oplus$ behave in the same way on~$\{0,1\}$, it follows that $v_\theta\left(\bigvee^m_{i=1}p^j_i\vee\bigvee^l_{i=1}\neg q_i\right)=1$ for all $j\in\{1,\ldots,n\}$. Hence, $v(\phi)=1$, and $\phi$ is $\CPL$-satisfiable, as required.
\end{proof}
\subsection{Proofs of Section~\ref{ssec:PIC}}
To establish Theorem~\ref{theorem:SPCFsolutionexistence}, we first need several technical statements. In particular, we need to show that entailment and satisfiability of formulas of the form $\Prob(\pi)\diamond\cvalue$ with $\pi\in\LCPL$ being a~conjunction or a~disjunction of propositional variables are decidable in polynomial time.
\begin{lemma}\label{lemma:CIPliteralpolynomialvalidity}
Given $\Prob(\pi)\diamond\cvalue$ s.t.\ $\pi$ is a~conjunction or a~disjunction of variables, it takes polynomial time to determine whether $\FP\models\Prob(\pi)\diamond\cvalue$ and whether $\Prob(\pi)\diamond\cvalue\models_\FP\bot$.
\end{lemma}
\begin{proof}
We prove that $\FP\models\Prob(\pi)\diamond\cvalue$ iff one of the following conditions hold:
\begin{enumerate}[noitemsep,topsep=0pt]
\item[(i)] $\CPL\models\pi$, $\diamond\in\{\geq,>\}$, and $\Vmsf_\Prob(\Prob(\pi)\diamond\cvalue)\neq\varnothing$ (recall Definition~\ref{def:probabilisticintervalliterals} for $\Vmsf_{\Prob}(\Prob(\pi)\diamond\cvalue)$), or
\item[(ii)] $\CPL\models\neg\pi$, $\diamond\in\{\leq,<\}$, and $\Vmsf_\Prob(\Prob(\pi)\diamond\cvalue)\neq\varnothing$, or
\item[(iii)] $\Prob(\pi)\diamond\cvalue$ has the form $\Prob(\pi)\leq\one$ or $\Prob(\pi)\geq\zero$.
\end{enumerate}
Observe that it takes polynomial time to verify that $\Vmsf_\Prob(\Prob(\pi)\diamond\cvalue)\neq\varnothing$. It is also clear that if any condition (i)--(iii) holds, then $\FP\models\Prob(\pi)\diamond\cvalue$. Indeed, let (i) hold and $\Mfrak=\langle2^\Vmbf,\mu\rangle$ be any probabilistic model with $\Var(\pi)\subseteq\Vmbf$, then $\|\pi\|_\Mfrak=2^\Vmbf$, whence $\mu(\|\pi\|_\Mfrak)=1$. It is clear that $\Imc_\Mfrak(\Prob(\pi)\diamond\cvalue)=1$ provided that $\Vmsf_\Prob(\Prob(\pi)\diamond\cvalue)\neq\varnothing$. We can show that (ii) also implies $\FP\models\Prob(\pi)\diamond\cvalue$ in a~dual manner. Finally, observe that if (iii) holds, then $\Vmsf_\Prob(\Prob(\pi)\diamond\cvalue)=[0,1]$, whence, again, $\FP\models\Prob(\pi)\diamond\cvalue$.

For the converse direction, let $\FP\models\Prob(\pi)\diamond\cvalue$ and assume for the sake of contradiction that neither of conditions (i)--(iii) holds. This means that $\pi$ is $\CPL$-satisfiable but not $\CPL$-valid and $\Vmsf_\Prob(\Prob(\pi)\diamond\cvalue)\neq[0,1]$. It is clear that for each $x\notin\Vmsf_\Prob(\Prob(\pi)\diamond\cvalue)$ there is an probabilistic model $\Mfrak_x=\langle2^{\Var(\pi)},\mu_x\rangle$ s.t.\ $\mu(\|\pi\|_\Mfrak)=x$. By Definition~\ref{def:FPLuksemantics}, we obtain that $\Imc_{\Mfrak_x}(\Prob(\pi)\diamond\cvalue)=0$, as required. As conditions (i)--(iii) can be verified in polynomial time, it follows that $\FP$-validity of $\Prob(\pi)\diamond\cvalue$ can be established in polynomial time too.

Finally, to check the $\FP$-unsatisfiability of $\Prob(\pi)\diamond\cvalue$, it suffices to check whether $\FP\models\neg(\Prob(\pi)\diamond\cvalue)$, i.e., whether $\Prob(\pi)\blackdiamond\cvalue$ is $\FP$-valid. By the previous paragraph, this is decidable in polynomial time.
\end{proof}
\begin{lemma}\label{lemma:CIPliteralpolynomialentailment}
Given $\Prob(\pi)\diamond\cvalue$ and $\Prob(\varrho)\diamond'\dvalue$ s.t.\ $\pi$ and $\varrho$ conjunctions or disjunctions of variables, it takes polynomial time to determine whether $\Prob(\pi)\diamond\cvalue\models_\FP\Prob(\varrho)\diamond'\dvalue$ and whether $\Prob(\pi)\diamond\cvalue,\Prob(\varrho)\diamond'\dvalue\models_\FP\bot$.
\end{lemma}
\begin{proof}
Observe that if $\pi$ and $\varrho$ have forms $\bigwedge^m_{i=1}p_i$ or $\bigvee^n_{j=1}q_j$, then one of the following statements holds:
\begin{enumerate}[noitemsep,topsep=0pt]
\item $\pi\models_\CPL\varrho$ but $\varrho\not\models_\CPL\pi$, or
\item $\varrho\models_\CPL\pi$ but $\pi\not\models_\CPL\varrho$, or
\item $\pi$ and $\varrho$ are $\CPL$-equivalent, or
\item $\pi\not\models_\CPL\varrho$, $\varrho\not\models_\CPL\pi$, and $\pi,\varrho\not\models_\FP\bot$.
\end{enumerate}
Furthermore, given $\pi$ and $\varrho$, it takes polynomial time to determine whether case~1, 2, 3, or~4 holds. We show how to determine entailment in polynomial time in each of these three cases when $\Prob(\pi)\diamond\cvalue$ and $\Prob(\varrho)\diamond'\dvalue$ are $\FP$-satisfiable and not $\FP$-valid.

First, let $\pi\models_\CPL\varrho$ and $\varrho\not\models_\CPL\pi$. We show that $\Prob(\pi)\diamond\cvalue\models_\FP\Prob(\varrho)\diamond'\dvalue$ iff (a) $\diamond,\diamond'\in\{\geq,>\}$ and (b) $\Vmsf_\Prob(\Prob(\pi)\diamond\cvalue)\subseteq\Vmsf_\Prob(\Prob(\varrho)\diamond'\dvalue)$. Assume that (a) and (b) hold, and assume w.l.o.g.\ that $\diamond,\diamond'\in\{\geq\}$. By (b), it follows that $c\geq d$. Now, as $\|\pi\|_\Mfrak\subseteq\|\varrho\|_\Mfrak$, it follows that $\mu(\|\pi\|_\Mfrak)\leq\mu(\|\varrho\|_\Mfrak)$, for every probabilistic model~$\Mfrak$. Thus, $\mu(\|\pi\|_\Mfrak)\geq c$ entails $\mu(\|\pi\|_\Mfrak)\geq d$. It follows that $\Prob(\pi)\geq\cvalue\models_\FP\Prob(\varrho)\geq\dvalue$. Other combinations of $\geq$ and $>$ in $\Prob(\pi)\geq\cvalue$ and $\Prob(\varrho)\geq\dvalue$ can be considered similarly.

Conversely, let $\Prob(\pi)\diamond\cvalue\models_\FP\Prob(\varrho)\diamond'\dvalue$. Assume for the sake of contradiction that $\diamond$ is $\leq$ (the case of $\diamond$ being $<$ can be dealt with similarly). Let $\Mfrak=\langle2^{\Var(\pi\wedge\varrho)},\mu\rangle$ be s.t.\ $\mu(\|\pi\|_\Mfrak)=0$. Clearly, $\Imc_\Mfrak(\Prob(\pi)\leq\cvalue)=1$. As $\Prob(\varrho)\diamond'\dvalue$ is not $\FP$-valid and $\|\pi\|_\Mfrak\subsetneq\|\varrho\|_\Mfrak$ (recall that $\pi$ and $\varrho$ are not $\CPL$-equivalent), it is clear that we can extend $\mu$ in such a~way that $\mu(\|\varrho\|_\Mfrak)\notin\Vmsf(\Prob(\varrho)\diamond'\dvalue)$, whence $\Imc_\Mfrak(\Prob(\varrho)\diamond'\dvalue)=0$. Thus, $\Mfrak$ witnesses $\Prob(\pi)\leq\cvalue\not\models_\FP\Prob(\varrho)\diamond'\dvalue$. Similarly, if $\diamond'$ is $<$ (the case of $\diamond'$ being~$\leq$ can be shown in the same way), we define $\Mfrak=\langle2^{\Var(\pi\wedge\varrho)},\mu\rangle$ s.t.\ $\mu(\|\varrho\|_\Mfrak)=1$ (whence $\Imc(\Prob(\varrho)<\dvalue)=0$) and $\mu(\|\pi\|)\in\Vmsf(\Prob(\pi)\diamond\cvalue)$, which gives us $\Imc_\Mfrak(\Prob(\pi)\diamond\cvalue)=1$, i.e., $\Mfrak$ witnesses $\Prob(\pi)\diamond\cvalue\not\models_\FP\Prob(\varrho)<\dvalue$. Now let $\Vmsf_\Prob(\Prob(\pi)\diamond\cvalue)\not\subseteq\Vmsf_\Prob(\Prob(\varrho)\diamond'\dvalue)$. We further assume that $\diamond,\diamond'\in\{\geq\}$ (other combinations of $\geq$ and~$>$ can be dealt with in a~similar way). This means that $c<d$. Now let $\Mfrak=\langle2^{\Var(\pi,\varrho)},\mu\rangle$ be s.t.\ $\mu(\|\pi\|_\Mfrak)=\mu(\|\varrho\|_\Mfrak)=c$. Clearly, $\Imc_\Mfrak(\Prob(\pi)\geq\cvalue)=1$ but $\Imc_\Mfrak(\Prob(\pi)\geq\dvalue)=0$. Thus, $\Prob(\pi)\geq\cvalue\not\models_\FP\Prob(\varrho)\geq\dvalue$.

Now assume that $\varrho\models_\CPL\pi$ but $\pi\not\models_\CPL\varrho$. By an argument similar to the previous case, we can show that $\Prob(\pi)\diamond\cvalue\models_\FP\Prob(\varrho)\diamond'\dvalue$ iff (a)~$\diamond,\diamond'\in\{\leq,<\}$ and $\Vmsf_\Prob(\Prob(\pi)\diamond\cvalue)\subseteq\Vmsf_\Prob(\Prob(\varrho)\diamond'\dvalue)$.

Assume that $\pi$ and $\varrho$ are $\CPL$-equivalent. We show that $\Prob(\pi)\diamond\cvalue\models_\FP\Prob(\varrho)\diamond'\dvalue$ iff $\Vmsf_\Prob(\Prob(\pi)\diamond\cvalue)\subseteq\Vmsf_\Prob(\Prob(\varrho)\diamond'\dvalue)$. Indeed, observe that $\|\pi\|_\Mfrak=\|\varrho\|_\Mfrak$, whence $\mu(\|\pi\|_\Mfrak)=\mu(\|\varrho\|_\Mfrak)$, in every probabilistic model~$\Mfrak$. Thus, if $\mu(\|\pi\|_\Mfrak)\in\Vmsf(\Prob(\pi)\diamond\cvalue)$, then $\mu(\|\varrho\|_\Mfrak)\in\Vmsf(\Prob(\varrho)\diamond'\dvalue)$. The converse direction can be obtained similarly.

Finally, assume that $\pi\not\models_\CPL\varrho$, $\varrho\not\models_\CPL\pi$, and $\pi,\varrho\not\models_\FP\bot$. We show that $\Prob(\pi)\diamond\cvalue\not\models_\FP\Prob(\varrho)\diamond'\dvalue$. Indeed, as $\Prob(\pi)\diamond\cvalue$ and $\Prob(\varrho)\diamond'\dvalue$ are not valid nor unsatisfiable, there are some $x,y\in\{0,1\}$ s.t.\ $x\in\Vmsf_\Prob(\Prob(\pi)\diamond\cvalue)$ and $y\notin\Vmsf_\Prob(\Prob(\varrho)\diamond'\dvalue)$. Assume, in particular, that $x=0$ and $y=1$ (other combinations can be dealt with similarly). Clearly, there is a~probabilistic model $\Mfrak=\langle2^{\Var(\pi\wedge\varrho)},\mu\rangle$ s.t.\ $\mu(\|\pi\|_\Mfrak)=0$ and $\mu(\|\varrho\|_\Mfrak)=1$. But in this model, $\Imc_\Mfrak(\Prob(\pi)\diamond\cvalue)=1$ and $\Imc_\Mfrak(\Prob(\varrho)\diamond'\dvalue)=0$, whence $\Prob(\pi)\diamond\cvalue\not\models_\FP\Prob(\varrho)\diamond'\dvalue$.

It is clear that in each case 1--4, it takes polynomial time to determine whether $\Prob(\pi)\diamond\cvalue\models_\FP\Prob(\varrho)\diamond'\dvalue$. Let us now present the polynomial procedure that verifies whether $\Prob(\pi)\diamond\cvalue\models_\FP\Prob(\varrho)\diamond'\dvalue$ in the general case. Given $\Prob(\pi)\diamond\cvalue$ and $\Prob(\varrho)\diamond'\dvalue$, we first check whether $\Prob(\pi)\diamond\cvalue$ or $\Prob(\varrho)\diamond'\dvalue$ are $\FP$-unsatisfiable or $\FP$-valid (this takes us polynomial time by Lemma~\ref{lemma:CIPliteralpolynomialvalidity}). If $\Prob(\pi)\diamond\cvalue$ is unsatisfiable or $\Prob(\varrho)\diamond'\dvalue$ is valid, then $\Prob(\pi)\diamond\cvalue\models_\FP\Prob(\varrho)\diamond'\dvalue$. If $\Prob(\pi)\diamond\cvalue$ is satisfiable and $\Prob(\varrho)\diamond'\dvalue$ is unsatisfiable or $\Prob(\pi)\diamond\cvalue$ is valid and $\Prob(\varrho)\diamond'\dvalue$ is non-valid, then $\Prob(\pi)\diamond\cvalue\not\models_\FP\Prob(\varrho)\diamond'\dvalue$. Otherwise, we $\Prob(\pi)\diamond\cvalue$ and $\Prob(\varrho)\diamond'\dvalue$ are satisfiable and not valid. In this case, we verify which case 1--4 holds for $\pi$ and~$\varrho$, and check whether $\Prob(\pi)\diamond\cvalue\models_\FP\Prob(\varrho)\diamond'\dvalue$ in polynomial time, as we showed above.

To check that $\Prob(\pi)\diamond\cvalue,\Prob(\varrho)\diamond'\dvalue\not\models_\FP\bot$, it suffices to check that $\Prob(\pi)\diamond\cvalue\models_\FP\neg(\Prob(\varrho)\diamond'\dvalue)$, i.e., $\Prob(\pi)\diamond\cvalue\models_\FP\Prob(\varrho)\blackdiamond'\dvalue$. The result now follows.
\end{proof}

We are now ready to prove Theorem~\ref{theorem:SPCFsolutionexistence}.
\SPCFsolutionexistence*
\begin{proof}[Proof of Item~1]
We adapt the proof of~\cite[Theorem~15]{InoueKozhemiachenko2025}. To show that the existence of sufficient solutions is decidable in polynomial time, we construct a~polynomial embedding of SPCF abduction problems into classical abduction problems whose theory is a~set of $\ihsbminus$-clauses, i.e., formulas of the following form: $p$, $\neg q$, $p{\rightarrow}q$, or $\bigwedge^n_{i=1}p_i{\rightarrow}\bot$.

Let $\Pmbb=\langle\Gamma,\delta,\Hmsf,\Vmc\rangle$ be an SPCF abduction problem, $\delta=\Prob(\chi)\diamond\dvalue$, and $\eta$ a~probabilistic interval term. Let further $\Gamma=\{\kappa_1,\ldots,\kappa_m\}\cup\{\kappa^+_1,\ldots,\kappa^+_{m^+}\}$ be s.t.\ $\kappa$'s and $\kappa^+$'s have the following forms:
\begin{align*}
\kappa_i&=\bigodot^{n_i}_{j=1}\lambda_i\rightarrow\bot&\kappa^+_i&=\lambda'_i\rightarrow\lambda''_i
\end{align*}

Now, for each PIL $\lambda$ occurring in~$\Pmbb$, let $r_\lambda$ be a~fresh propositional variable. For $\kappa,\kappa^+\in\Gamma$, we define 
\begin{align*}
\kappa_\Hmc&\coloneqq\bigwedge_{\lambda\in\kappa}r_\lambda\rightarrow\bot&\kappa^+_\Hmc&\coloneqq r_{\lambda'}\rightarrow r_{\lambda''}
\end{align*}
We consider the following \emph{classical} abduction problem: $\Pmbb_\Hmc=\langle\Gamma_\Hmc,\delta_\Hmc,\Hmsf_\Hmc\rangle$ defined as described below.
\begin{align}\label{equ:Hornreductionproblem}
\Gamma_\Hmc=&\left\{\kappa_\Hmc\mid\kappa\in\Gamma\right\}\cup\{\kappa^+_\Hmc\mid\kappa^+\in\Gamma\}\cup\left\{r_{\Prob(\sigma)\diamond\cvalue}\!\wedge\!r_{\Prob(\sigma')\diamond\cvalue'}\!\rightarrow\!\bot\left|\begin{matrix}\Prob(\sigma)\diamond\cvalue,\Prob(\sigma')\diamond\cvalue'\models_\FP\bot\\\sigma,\sigma'\in\Emc[\Pmbb];~c,c'\in\Vmc\end{matrix}\right.\right\}\cup\nonumber\\&\left\{r_{\Prob(\sigma)\diamond\cvalue}\rightarrow r_{\Prob(\sigma')\diamond\cvalue'}\left|\begin{matrix}\Prob(\sigma)\diamond\cvalue\models_\FP\Prob(\sigma')\diamond\cvalue'\\\sigma,\sigma'\in\Emc[\Pmbb];~c,c'\in\Vmc
\end{matrix}\right.\right\}\nonumber\\
\delta_\Hmc=&r_\delta\nonumber\\
\Hmsf_\Hmc=&\{r_{\Prob(\pi)\diamond\evalue}\mid\pi\in\Hmsf\text{ and }e\in\Vmc\}
\end{align}
Note from Lemma~\ref{lemma:CIPliteralpolynomialentailment} that it takes polynomial time to determine whether $\lambda\models_\FP\lambda'$ and $\lambda,\lambda'\models_\FP\bot$ for any two PILs occurring in~$\Pmbb$. Furthermore, the size of~$\Pmbb_\Hmc$ is polynomial w.r.t.\ the size of~$\Pmbb$. Hence, it takes polynomial time to construct $\Pmbb_\Hmc$. Moreover, $\Gamma_\Hmc$ is indeed a~set of $\ihsbminus$-clauses. It remains to show that $\Pmbb$ has a~sufficient solution iff $\Pmbb_\Hmc$ has a~solution.

First, assume that $\eta=\bigodot^k_{i=1}\Prob(\pi_i)\diamond\evalue_i$ is a~sufficient solution to~$\Pmbb$. We show that $\eta_\Hmc=\bigwedge^k_{i=1}r_{\Prob(\pi_i)\diamond\evalue_i}$ is a~solution to $\Pmbb_\Hmc$. Recall from Definitions~\ref{def:Luksemantics} and~\ref{def:FPLuksemantics} that $\Imc(\Prob(\phi)\diamond\cvalue)\in\{0,1\}$ for every $\phi$ and $\Mfrak$. Now let $\Mfrak$ be a~probabilistic model witnessing $\Gamma,\eta\not\models_\FP\bot$. We set
\begin{align}\label{equ:Hornvaluation}
v_\Mfrak(p_\lambda)&=
\begin{cases}
1&\text{if }\Imc_\Mfrak(\lambda)=1\\
0&\text{otherwise}
\end{cases}
\end{align}
It is easy to check that $v_\Mfrak(\phi)=1$ for every $\phi\in\Gamma_\Hmc\cup\{\eta_\Hmc\}$. Thus, $\Gamma_\Hmc,\eta_\Hmc\not\models_\CPL\bot$. To see that $\Gamma_\Hmc,\eta_\Hmc\models_\CPL\delta_\Hmc$, assume for contradiction that there is some classical valuation $v$ s.t.\ $v(\phi)=1$ for every $\phi\in\Gamma_\Hmc\cup\{\tau_\Hmc\}$ but $v(\delta_\Hmc)=0$. Our goal now is to define a~probabilistic model $\Mfrak_v=\langle2^\Vmbf,\mu_v\rangle$ with $\Vmbf=\Var[\Pmbb]$ s.t.\ $\Imc_{\Mfrak_v}(\gamma)=1$ for every $\gamma\in\Gamma\cup\{\eta\}$ but $\Imc_{\Mfrak_v}(\delta)=0$ (observe that since $\delta$ is a~PIL, it can only have values~$0$ or~$1$). For every $\pi\in\Emc[\Pmbb]$, we define 
\begin{align}
\Vmc_{\Prob(\pi)}=&\bigcap\limits_{v(r_{\Prob(\pi)\geq\cvalue})=1}\!\!\!\!\!\!\!\![c,1]\cap\!\!\!\!\!\bigcap\limits_{v(r_{\Prob(\pi)>\cvalue'})=1}\!\!\!\!\!\!\!\!(c',1]\cap\bigcap\limits_{v(r_{\Prob(\pi)\leq\cvalue''})=0}\!\!\!\!\!\!\!\!(c'',1]\cap\!\!\!\!\!\bigcap\limits_{v(r_{\Prob(\pi)<\cvalue'''})=0}\!\!\!\!\!\!\!\![c''',1]\cap\nonumber\\
&\bigcap\limits_{v(r_{\Prob(\pi)\geq\dvalue})=0}\!\!\!\!\!\!\!\![0,d)\cap\!\!\!\!\!\bigcap\limits_{v(r_{\Prob(\pi)>\dvalue'})=0}\!\!\!\!\!\!\!\![0,d']\cap\bigcap\limits_{v(r_{\Prob(\pi)\leq\dvalue''})=1}\!\!\!\!\!\!\!\![0,d'']\cap\!\!\!\!\!\bigcap\limits_{v(r_{\Prob(\pi)<\dvalue'''})=1}\!\!\!\!\!\!\!\![0,d''')\label{equ:valuesets}
\end{align}

Let us show that $\Vmc_{\Prob(\pi)}\neq\varnothing$ for each $\pi\in\Emc[\Pmbb]$. Assume for contradiction that there is some $\pi\in\Emc[\Pmbb]$ s.t.\ $\Vmc_{\Prob(\pi)}=\varnothing$. This means that one of the following holds.
\begin{enumerate}[noitemsep,topsep=0pt]
\item There are $\triangleright,\triangleright'\in\{\geq,>\}$, $\triangleleft,\triangleleft'\in\{\leq,<\}$, and $c,d\in[0,1]_\Qmbb$ such that $c>d$ and
\begin{enumerate}[noitemsep,topsep=0pt]
\item $v(r_{\Prob(\pi)\triangleright\cvalue})=1$ and $v(r_{\Prob(\pi)\triangleleft\dvalue})=1$, or
\item $v(r_{\Prob(\pi)\triangleright\cvalue})=1$ and $v(r_{\Prob(\pi)\triangleright'\dvalue})=0$, or
\item $v(r_{\Prob(\pi)\triangleleft\cvalue})=0$ and $v(r_{\Prob(\pi)\triangleright\dvalue})=0$, or
\item $v(r_{\Prob(\pi)\triangleleft\cvalue})=0$ and $v(r_{\Prob(\pi)\triangleleft'\dvalue})=1$.
\end{enumerate}
\item There are $\triangleright,\triangleright'\in\{\geq,>\}$, $\triangleleft,\triangleleft'\in\{\leq,<\}$, and $c,d\in[0,1]_\Qmbb$ such that $c\geq d$, $\triangleright\neq\triangleright'$, $\triangleleft\neq\triangleleft'$, and
\begin{enumerate}[noitemsep,topsep=0pt]
\item $v(r_{\Prob(\pi)\triangleright\cvalue})=1$ and $v(r_{\Prob(\pi)\triangleleft'\dvalue})=1$, or
\item $v(r_{\Prob(\pi)\triangleright\cvalue})=1$ and $v(r_{\Prob(\pi)\triangleright'\dvalue})=0$, or
\item $v(r_{\Prob(\pi)\triangleleft\cvalue})=0$ and $v(r_{\Prob(\pi)\triangleright'\dvalue})=0$, or
\item $v(r_{\Prob(\pi)\triangleleft\cvalue})=0$ and $v(r_{\Prob(\pi)\triangleleft'\dvalue})=1$.
\end{enumerate}
\end{enumerate}

We consider only the first case since the second one can be dealt with in a~similar way. In the case (1.a), we have $\Prob(\pi)\triangleright\cvalue,\Prob(\pi)\triangleleft\dvalue\models_\FP\bot$. Thus, $(r_{\Prob(\pi)\triangleright\cvalue}\wedge r_{\Prob(\pi)\triangleleft\dvalue}){\rightarrow}\bot\in\Gamma_\Hmc$ and $v((r_{\Prob(\pi)\triangleright\cvalue}\wedge r_{\Prob(\pi)\triangleleft\dvalue})\rightarrow\bot)=0$, and $v$~does not satisfy~$\Gamma$, contrary to the assumption. In the case (1.b), note that $\Prob(\pi)\triangleright\cvalue\models_\FP\Prob(\pi)\triangleright'\dvalue$, whence $r_{\Prob(\pi)\triangleright\cvalue}\rightarrow r_{\Prob(\pi)\triangleright'\dvalue}\in\Gamma_\Hmc$. Again, $v$~does not satisfy~$\Gamma_\Hmc$. In the case (1.c), we have that $\FP\models(\Prob(\pi)\triangleleft\cvalue)\oplus(\Prob(\pi)\triangleright\dvalue)$, whence $\Gamma$ is not SPCF, contrary to the assumption. Finally, in the case (1.d), we have $\Prob(\pi)\triangleleft'\dvalue\models_\FP\Prob(\pi)\triangleleft\cvalue)$, whence $r_{\Prob(\pi)\triangleleft'\dvalue}\rightarrow r_{\Prob(\pi)\triangleleft\cvalue}\in\Gamma_\Hmc$. Just as in the case (1.b), we have that $v$~does not satisfy~$\Gamma_\Hmc$.

Moreover, using a~similar reasoning, we can show that for every $\pi,\varrho\in\Emc[\Pmbb]$ s.t.\ $\pi\models_\CPL\varrho$, there are some $x\in\Vmc_{\Prob(\pi)}$ and $x'\in\Vmc_{\Prob(\varrho)}$ with $x\leq x'$. Again, assume for the sake of contradiction that there are some $\pi,\varrho\in\Emc[\Pmbb]$ s.t.\ $\pi\models_\CPL\varrho$, but $x>x'$ for all $x\in\Vmc_{\Prob(\pi)}$ and $x'\in\Vmc_{\Prob(\varrho)}$. It follows that one of the following holds.
\begin{enumerate}[noitemsep,topsep=0pt]
\item There are $\triangleright,\triangleright'\in\{\geq,>\}$, $\triangleleft,\triangleleft'\in\{\leq,<\}$, and $c,d\in[0,1]_\Qmbb$ such that $c>d$ and
\begin{enumerate}[noitemsep,topsep=0pt]
\item $v(r_{\Prob(\pi)\triangleright\cvalue})=1$ and $v(r_{\Prob(\varrho)\triangleleft\dvalue})=1$, or
\item $v(r_{\Prob(\pi)\triangleright\cvalue})=1$ and $v(r_{\Prob(\varrho)\triangleright'\dvalue})=0$, or
\item $v(r_{\Prob(\pi)\triangleleft\cvalue})=0$ and $v(r_{\Prob(\varrho)\triangleright\dvalue})=0$, or
\item $v(r_{\Prob(\pi)\triangleleft\cvalue})=0$ and $v(r_{\Prob(\varrho)\triangleleft'\dvalue})=1$.
\end{enumerate}
\item There are $\triangleright,\triangleright'\in\{\geq,>\}$, $\triangleleft,\triangleleft'\in\{\leq,<\}$, and $c,d\in[0,1]_\Qmbb$ such that $c\geq d$, $\triangleright\neq\triangleright'$, $\triangleleft\neq\triangleleft'$, and
\begin{enumerate}[noitemsep,topsep=0pt]
\item $v(r_{\Prob(\pi)\triangleright\cvalue})=1$ and $v(r_{\Prob(\varrho)\triangleleft'\dvalue})=1$, or
\item $v(r_{\Prob(\pi)\triangleright\cvalue})=1$ and $v(r_{\Prob(\varrho)\triangleright'\dvalue})=0$, or
\item $v(r_{\Prob(\pi)\triangleleft\cvalue})=0$ and $v(r_{\Prob(\varrho)\triangleright'\dvalue})=0$ (with $\langle\triangleleft,\triangleright'\rangle\in\{\langle<,\geq\rangle,\langle\leq,>\rangle\}$), or
\item $v(r_{\Prob(\pi)\triangleleft\cvalue})=0$ and $v(r_{\Prob(\varrho)\triangleleft'\dvalue})=1$.
\end{enumerate}
\end{enumerate}

Just as before, we consider only case~1 because case~2 can be tackled in the same way. The reasoning is similar. In the case (1.a), we have $\Prob(\pi)\triangleright\cvalue,\Prob(\varrho)\triangleleft\dvalue\models_\FP\bot$. Hence, $v$~does not satisfy~$\Gamma$, contrary to the assumption. In the case (1.b), $\Prob(\pi)\triangleright\cvalue\models_\FP\Prob(\varrho)\triangleright'\dvalue$, whence again, $v$~does not satisfy~$\Gamma_\Hmc$. In the case (1.c), we have that $\FP\models(\Prob(\pi)\triangleleft\cvalue)\oplus(\Prob(\varrho)\triangleright\dvalue)$, whence $\Gamma$ is not SPCF, contrary to the assumption. Finally, in the case (1.d), we have $\Prob(\varrho)\triangleleft'\dvalue\models_\FP\Prob(\pi)\triangleleft\cvalue)$, whence $v$~does not satisfy~$\Gamma_\Hmc$ just as in the case (1.b).

Now, for every $\pi\in\Emc[\Pmbb]$, set $f(\pi)=x$ for some $x\in\Vmc_{\Prob(\pi)}$ s.t.\ $\pi\models_\CPL\pi'$ implies $f(\pi)\leq f(\pi')$ for every pair of $\pi,\pi'\in\Emc[\Pmbb]$. It follows from the above that such~$f$ exists. Moreover, one can see that $f$~is a~$\models_\CPL$-monotone function (recall Definition~\ref{def:entailmentmonotonicity}). By Lemma~\ref{lemma:monotoneassignmentCIP}, we can extend~$f$ to a~measure $\mu_v$ on~$2^\Vmbf$ s.t.\ $f(\pi)=\mu_v(\|\pi\|_\Mfrak)$ for all $\pi\in\Emc[\Pmbb]$. Now, using Definitions~\ref{def:Luksemantics} and~\ref{def:FPLuksemantics}, it is easy to check that $\Imc_{\Mfrak_v}(\gamma)=1$ for all $\gamma\in\Gamma$ and $\Imc_{\Mfrak_v}(\eta)=1$ but $\Imc_{\Mfrak_v}(\delta)=0$. Indeed, observe from~\eqref{equ:valuesets} that $v(r_\lambda)=1$ iff $\Imc_{\Mfrak_v}(\lambda)=1$. Thus, $\Gamma,\eta\not\models_\FP\delta$ which contradicts the assumption that $\delta$ solves~$\Pmbb$. Thus, $\Gamma_\Hmc,\eta_\Hmc\models_\CPL\delta_\Hmc$. It follows now that $\delta_\Hmc$ is a~solution to~$\Pmbb_\Hmc$.

For the converse direction, let $\zeta$ be a~solution to $\Pmbb_\Hmc$. That is, $\Gamma_\Hmc,\zeta\consvDashCPL\delta_\Hmc$. Assume that $v$ is a~classical valuation witnessing $\Gamma_\Hmc,\zeta\not\models_\CPL\bot$ and set $\zeta'=\bigodot_{r_\lambda\in\zeta}\lambda$. We define $\Mfrak_v$ using~\eqref{equ:valuesets} as in the previous part of the proof. This gives us that $\Mfrak_v$~witnesses $\Gamma,\zeta'\not\models_\FP\bot$. It remains to see that $\Gamma,\zeta'\models_\FP\delta$. To do that, assume for contradiction that some probabilistic model is witnessing $\Gamma,\zeta'\not\models_\FP\delta$. We define $v_\Hmc$ as shown in~\eqref{equ:Hornvaluation}. It is clear that $v_\Hmc$ would witness $\Gamma_\Hmc,\zeta\not\models_\CPL\delta_\Hmc$, contrary to the assumption. Hence, $\zeta'$ is a~sufficient solution to~$\Pmbb$, as required.
\end{proof}
\begin{proof}[Proof of Item~2]
We show only the $\np$-hardness since $\np$-membership follows from Theorem~\ref{theorem:fullsolutionexistenceNP}. We provide a~polynomial reduction from $\CPL$-satisfiability of conjunctive normal forms. The idea of the proof is similar to that of Theorem~\ref{theorem:PSCsolutionexistence} (Item~2).

Let $\phi\in\LCPL$ be in CNF and $r\notin\Var(\phi)$. We can equivalently rewrite~$\phi$ as a~conjunction of implicative clauses $\bigwedge^n_{j=1}\left(\left(\bigwedge^m_{i=1}p^j_i\wedge\bigwedge^l_{i=1}\neg q_i\right)\rightarrow\bot\right)$. Now, define
\begin{align*}
\Gamma_\Cmsf&=\left\{\left(\bigodot^m_{i=1}\Prob(p^j_i){\geq}\one\odot\bigodot^l_{i=1}\Prob(q^j_i){\leq}\zero\right){\rightarrow}\bot\mid j\in\{1,\ldots,n\}\right\}\cup\{\Prob(r){\geq}\one\}&
\delta_\Cmsf&=\Prob(r){\geq}\one&\Hmsf_\Cmsf&=\{p^1_1\}&\Vmc_\Cmsf&=\{0,1\}\\
\Pmbb_\Cmsf&=\langle\Gamma_\Cmsf,\delta_\Cmsf,\Hmsf_\Cmsf,\Vmc_\Cmsf\rangle
\end{align*}
Observe that $\Gamma_\Cmsf\cup\{\delta_\Cmsf\}$ is indeed an SPCF theory. First, there are no $p,q\in\Var$ s.t.\ $\FP\models\Prob(p){\geq}\one\oplus\Prob(q){\leq}\zero$. Second, all events in~$\Pmbb_\Cmsf$ are propositional variables, whence $\Emc[\Pmbb]$~is chained positive.

Now assume that $\phi$ is $\CPL$-satisfiable. In particular, let $v(\phi)=1$ for some classical valuation~$v$. We define~$\tau_v$ and~$\theta_v$ as shown in~\eqref{equ:completePITfromvaluation}. Using the same reasoning as in the proof of Theorem~\ref{theorem:PSCsolutionexistence}, we can prove that~$\theta$ is a~full solution to~$\Pmbb_\Cmsf$. For the converse direction, assume that $\theta=\Prob(\tau){\approx}\one$ is a~full solution to~$\Pmbb_\Cmsf$. Again, using the reasoning from Theorem~\ref{theorem:PSCsolutionexistence}, we can show that the following classical valuation~--- $v_\theta(h)=1$ if $h\in\tau$, $v_\theta(h)=0$ if $\neg h\in\tau$~--- satisfies~$\phi$.
\end{proof}
\section{Proofs of Section~\ref{sec:probabduction}}
\PrAPsolutionsconditionalprobability*
\begin{proof}
Assume that $\tau$ is a~solution to~$\Pmbb_\pfrak$. Then $\phi,\tau\models_\CPL\chi$ and there is a~probabilistic model~$\Mfrak_\pfrak=\langle2^\Vmbf,\mu_\pfrak\rangle$ coherent with~$\pfrak$ s.t.\ $\mu_\pfrak(\|\phi\wedge\tau\|_\Mfrak)>0$. Now, $\phi,\tau\models_\CPL\chi$ implies that $\|\phi\wedge\tau\|_\Mfrak\subseteq\|\chi\|_\Mfrak$ for all probabilistic models $\Mfrak=\langle2^\Vmbf,\mu\rangle$. It follows that $\mu(\|\phi\wedge\tau\|_\Mfrak)\leq\mu(\|\chi\|_\Mfrak)$. This gives us condition~1. To see that condition~2 is satisfied, we proceed as follows. Note that we have $\|\phi\wedge\tau\|_{\Mfrak_\pfrak}\subseteq\|\chi\|_{\Mfrak_\pfrak}$ because $\phi,\tau\models_\CPL\chi$. Hence, $\|\phi\wedge\tau\|_{\Mfrak_\pfrak}\cap\|\chi\|_{\Mfrak_\pfrak}=\|\phi\wedge\tau\|_{\Mfrak_\pfrak}$, i.e., $\mu_\pfrak(\|\phi\wedge\tau\|_{\Mfrak_\pfrak}\cap\|\chi\|_{\Mfrak_\pfrak})=\mu(\|\phi\wedge\tau\|_{\Mfrak_\pfrak})$. Finally, $\mu_\pfrak(\|\phi\wedge\tau\|_\Mfrak)>0$ implies that $\Probfrak_\mu(\chi\mid\phi\wedge\chi)=1$, as required.

For the converse direction, assume that \emph{$\tau$ is not a~solution to~$\Pmbb_\pfrak$}. Then either (i)~$\phi,\tau\not\models_\CPL\chi$ or (ii)~$\mu(\|\phi\wedge\tau\|_\Mfrak)=0$ in all probabilitstic models $\Mfrak=\langle2^\Vmbf,\mu\rangle$ s.t.\ $\mu$ is coherent with~$\pfrak$. In the first case, let $v$~be a~classical valuation witnessing $\phi,\tau\not\models_\CPL\chi$. Let, further $X=\{p\mid v(p)=1\}$. It is clear that $\{X\}$ is an atom in $2^{2^\Vmbf}$. Thus, setting $\mu_v(\{X\})=1$ completely defines a~measure on~$2^{2^\Vmbf}$ Now, let $\Mfrak_v=\langle2^\Vmbf,\mu_v\rangle$. It follows that $\mu_v(\|\phi\wedge\tau\|_{\Mfrak_v})=1$ and $\mu_v(\|\chi\|_{\Mfrak_v})=0$. Thus, condition~1 is not satisfied. In the second case, observe that if $\mu(\|\phi\wedge\tau\|_\Mfrak)=0$ in all probabilitstic models $\Mfrak=\langle2^\Vmbf,\mu\rangle$ s.t.\ $\mu$ is coherent with~$\pfrak$, then $\Probfrak_\mu(\chi\mid\phi\wedge\tau)$ is undefined in all such models. This is because $\Probfrak_\mu(\chi\mid\phi\wedge\tau)=\dfrac{\mu(\|\phi\wedge\tau\wedge\chi\|_\Mfrak)}{\mu(\|\phi\wedge\tau\|_\Mfrak)}$. Hence, condition~2 is not satisfied. The result now follows.
\end{proof}
\PrAPtoFPLuk*
\begin{proof}
We begin with Item~1. Assume that $\tau$ is a~solution to~$\Pmbb_\pfrak$. By Proposition~\ref{prop:PrAPsolutionsconditionalprobability}, we have that $\mu(\|\phi\wedge\tau\|_\Mfrak)\leq\mu(\|\chi\|_\Mfrak)$, whence, $\Imc_\Mfrak(\Prob(\phi\wedge\tau))\leq\Imc_\Mfrak(\Prob(\chi))$ in all probabilistic models for~$\Vmbf$. Hence, $\Prob(\phi\wedge\tau)\rightarrow\Prob(\chi)$ is $\FP$-valid, as required. Furthermore, as $\tau$ is a~solution, there is a~probabilistic model $\Mfrak'=\langle2^\Vmbf,\mu\rangle$ s.t.\ $\mu$~is coherent with~$\pfrak$ and $\mu(\|\phi\wedge\tau\|_\Mfrak)>0$. By Definitions~\ref{def:FPLuksemantics} and~\ref{def:probabilisticAPFPLukcounterpart}, we have that $\Imc_{\Mfrak'}(\xi)=1$ for all $\xi\in\Xi_\pfrak$ and $\Imc_{\Mfrak'}(\Prob(\phi\wedge\tau))>0$. Thus, $\Xi_\pfrak,\neg\triangle\neg\Prob(\phi\wedge\tau)\not\models_\FP\bot$, as required.

For the converse direction, assume that $\tau$ is not a~solution to~$\Pmbb_\pfrak$. Hence, either (i)~$\phi,\tau\not\models_\CPL\chi$ or (ii)~$\mu(\|\phi\wedge\tau\|_\Mfrak)=0$ in all probabilitstic models $\Mfrak=\langle2^\Vmbf,\mu\rangle$ s.t.\ $\mu$ is coherent with~$\pfrak$. In the first case, by the same reasoning as in Proposition~\ref{prop:PrAPsolutionsconditionalprobability}, it follows that there is some probabilistic model $\Mfrak=\langle2^\Vmbf,\mu\rangle$ s.t.\ $\mu(\|\phi\wedge\tau\|_\Mfrak)>\mu(\|\chi\|_\Mfrak)$. Thus, $\Prob(\phi\wedge\tau)\rightarrow\Prob(\chi)$ is not $\FP$-valid. In the second case, using Definitions~\ref{def:FPLuksemantics} and~\ref{def:probabilisticAPFPLukcounterpart}, we obtain that $\Xi_\pfrak,\neg\triangle\neg\Prob(\phi\wedge\tau)\models_\FP\bot$.

Consider now Item~2. Assume first that $\tau$ is a~preferred solution to~$\Pmbb_\pfrak$ and $\sigma$ is another solution to~$\Pmbb_\pfrak$. Since $\tau$ is a~solution to $\Pmbb_\pfrak$, it follows that the conditions from Item~1 are satisfied. It remains to see that $\Xi_\pfrak,\Prob(\sigma)\rightarrow\Prob(\tau)\not\models_\FP\bot$. Now, recall from Definition~\ref{def:PrAPsolution} that if~$\tau$ is a~preferred solution to~$\Pmbb_\pfrak$ and $\sigma$ is some other solution to~$\Pmbb_\pfrak$, then there must be a~probabilistic model $\Mfrak'=\langle2^\Vmbf,\mu'\rangle$ s.t.\ $\mu$~is coherent with~$\pfrak$ and $\mu(\|\tau\|_{\Mfrak'})\geq\mu(\|\sigma\|_{\Mfrak'})$. Again, by Definitions~\ref{def:FPLuksemantics} and~\ref{def:probabilisticAPFPLukcounterpart}, it follows that $\Xi_\pfrak,\Prob(\sigma)\rightarrow\Prob(\tau)\not\models_\FP\bot$, as required.

Conversely, let $\tau$ be not a~preferred solution. If $\tau$ is not a~solution at all, we can reuse the reasoning from Item~1 and see that one of the conditions in Item~1 is falsified. If $\tau$ is a~solution, but not a~preferred one, then there must be some solution $\sigma$ s.t.\ $\mu(\|\sigma\|_\Mfrak)>\mu(\|\tau\|_\Mfrak)$ in every probabilistic model $\Mfrak=\langle2^\Vmbf,\mu\rangle$ s.t.\ $\mu$~is coherent with~$\pfrak$. Again, by Definitions~\ref{def:FPLuksemantics} and~\ref{def:probabilisticAPFPLukcounterpart}, it follows that $\Xi_\pfrak,\Prob(\sigma)\rightarrow\Prob(\tau)\models_\FP\bot$.
\end{proof}
\PrAPcomplexity*
\begin{proof}
We begin with hardness for Items~1 and~3. Recall that classical abduction problems are a~particular case of probabilistic abduction problems. As is well-known from~\cite{EiterGottlob1995}, solution recognition for classical abduction problems is $\DP$-complete and solution existence is $\Sigma^\Pmsf_2$-complete. It follows that solution recognition for PrAPs is $\DP$-hard and solution existence is $\DP$-complete.

Let us now proceed to $\DP$-membership (Item~1). Let $\Pmbb_\pfrak=\langle\{\phi\},\chi,\Hmsf,\Embb,\pfrak\rangle$ By Theorem~\ref{theorem:PrAPtoFPLuk}, an $\LCPL$-term $\tau$ composed of literals from~$\Hmsf$ is a~solution to $\Pmbb_\pfrak$ iff (i)~$\Xi_\pfrak,\neg\triangle\neg\Prob(\phi\wedge\tau)\not\models_\FP\bot$ and (ii)~$\FP\models\Prob(\phi\wedge\tau)\rightarrow\Prob(\chi)$ (recall Definition~\ref{def:probabilisticAPFPLukcounterpart} for~$\Xi_\pfrak$). By Proposition~\ref{prop:FPLukNPcoNP}, it follows that verifying~(i) is in~$\np$ and verifying~(ii) is in~$\conp$. Since these two checks are independent, we have the desired membership result.

Consider now Item~2. For membership, we will provide an $\Sigma^\Pmsf_2$-procedure that resolves the complementary problem: given an $\LCPL$-term $\tau$ and a~PrAP $\Pmbb_\pfrak=\langle\{\phi\},\chi,\Hmsf,\Embb,\pfrak\rangle$, decide whether $\tau$ is \emph{not} a~preferred solution. By Theorem~\ref{theorem:PrAPtoFPLuk}, $\tau$ is \emph{not} a~preferred solution to~$\Pmbb_\pfrak$ iff either (1)~it is not a~solution to~$\Pmbb_\pfrak$ at all or (2)~there is some other solution $\sigma$ to~$\Pmbb_\pfrak$ s.t.\ $\Xi_\pfrak,\Prob(\sigma)\rightarrow\Prob(\tau)\models_\FP\bot$. By Item~1, it takes us $\DP$ time to verify whether (1)~holds. If \emph{it does not} (i.e., if $\tau$ is a~solution to~$\Pmbb_\pfrak$), we verify~(2) as follows. We guess an $\LCPL$-term $\sigma$ composed from literals in~$\Hmsf$ and then check in $\conp$ time that $\Xi_\pfrak,\Prob(\sigma)\rightarrow\Prob(\tau)\models_\FP\bot$. If this check succeeds, it follows that $\tau$ was not a~preferred solution to~$\Pmbb_\pfrak$.

Finally, the $\Sigma^\Pmsf_2$-membership of solution existence for PrAPs follows immediately from Item~1. Indeed, it suffices to guess an $\LCPL$-term composed from hypotheses and then verify in $\DP$ time that it is indeed a~solution to~$\Pmbb_\pfrak$. The result now follows.
\end{proof}
\conprecognitionPrAP*
\begin{proof}
Let us set $\Vmbf=\Var[\Pmbb_\pfrak]$ to simplify the notation. We observe briefly that $\pfrak$ is indeed a~probability distribution on $2^\Vmbf$. Note that all formulas in $\Embb$ correspond to (a~subset of) atoms of Boolean algebra $2^{2^\Vmbf}$ and that they are assigned with values that sum up to~$1$. Furthermore, there is only one probabilistic model $\Mfrak=\langle2^\Vmbf,\mu_\pfrak\rangle$ s.t.\ $\mu_\pfrak$ is coherent with~$\pfrak$, namely, the model where $\mu_\pfrak$ is the probability measure induced by the probability distribution~$\pfrak$.

We begin with $\conp$-hardness. To show it, we construct a~polynomial-time reduction from $\CPL$-validity. Namely, let $\chi\in\LCPL$, $p\notin\Var(\chi)$, $\Hmsf=\{p\}$, $\Embb=\{p\wedge\bigwedge_{q\in\Var(\chi)}q\}$, and $\pfrak(p\wedge\bigwedge_{q\in\Var(\chi)}q)=1$. We show that $p$~is a~solution to $\Pmbb_\pfrak=\langle\{p\vee\neg p\},\chi,\Hmsf,\Embb,\pfrak\rangle$ iff $\chi$~is $\CPL$-valid. Assume that $\chi$ is $\CPL$-valid. It follows that $p,p\vee\neg p\models_\CPL\chi$. Now, one can see that $\mu(\|p\wedge(p\vee\neg p)\|_\Mfrak)=1$, whence $p$ is indeed a~solution to~$\Pmbb_\pfrak$. Conversely, assume that $\chi$ is not $\CPL$-valid. Then $p,p\vee\neg p\not\models_\CPL\chi$, i.e., $p$~is not a~solution to~$\Pmbb_\pfrak$.

Now, for $\conp$-membership, we proceed as follows. Let $\tau$ be an $\LCPL$-term and $\Pmbb_\pfrak=\langle\{\phi\},\chi,\Hmsf,\Embb,\pfrak\rangle$. It is clear that it takes $\conp$ time to verify $\phi,\tau\models_\CPL\chi$. Now, we need to establish whether $\mu_\pfrak(\|\phi\wedge\tau\|_\Mfrak)>0$. We check whether $\pi\models_\CPL\phi\wedge\tau$ for each $\pi\in\Embb$ s.t.\ $\pfrak(\pi)>0$. As $\Embb$ contains only $\LCPL$-terms over \emph{all} variables in~$\Vmbf$, it follows that each $\pi$~determines one valuation of variables in~$\Vmbf$. Thus, it takes polynomial time w.r.t.\ the size of the representation of~$\pfrak$ to check whether there is some $\pi\in\Embb$ s.t.\ $\pfrak(\pi)>0$ and $\pi\models_\CPL\phi\wedge\tau$. If this check succeeds and if $\phi,\tau\models_\CPL\chi$ holds, then $\tau$~is a~solution to~$\Pmbb_\pfrak$; otherwise it is not. $\conp$-membership now follows.
\end{proof}
\end{document}